\documentclass[11pt,fleqn,a4paper]{article} %,draft

\usepackage{amsmath,amssymb,amsthm,amsfonts} % ,enumerate,backref,refcheck
\usepackage[mathscr]{eucal}

\usepackage{hyperref}
\hypersetup{colorlinks, linkcolor=blue, citecolor=blue, urlcolor=blue}

\makeatletter
\long\def\@makecaption#1#2{%
  \vskip\abovecaptionskip\footnotesize
  \sbox\@tempboxa{#1. #2}%
  \ifdim \wd\@tempboxa >\hsize
    #1. #2\par
  \else
    \global \@minipagefalse
    \hb@xt@\hsize{\hfil\box\@tempboxa\hfil}%
  \fi
  \vskip\belowcaptionskip}
\makeatother

\allowdisplaybreaks

\newcommand{\p}{\partial}
\newcommand{\sgn}{\mathop{\rm sgn}\nolimits}
\newcommand{\ord}{\mathop{\rm ord}\nolimits}
\newcommand{\const}{{\rm const}}
\newcommand{\lsemioplus}{\mathbin{\mbox{$\lefteqn{\hspace{.77ex}\rule{.4pt}{1.2ex}}{\in}$}}}

\newlength{\mylength}
\newcommand{\solution}{\hspace*{-\mylength}\bullet\quad}
\newcommand{\solutionRedEq}{\hspace*{-\mylength}\circ\quad}

\newtheorem{theorem}{Theorem}%[section]
\newtheorem{lemma}[theorem]{Lemma}
\newtheorem{corollary}[theorem]{Corollary}

{\theoremstyle{definition}

\newtheorem{remark}[theorem]{Remark}

}

\newcommand{\todo}[1][\null]{\ensuremath{\clubsuit}}

\newcommand{\noprint}[1]{}

\usepackage{xcolor}

\begin{document}

\par\noindent {\LARGE\bf
Hidden symmetries, hidden conservation laws and\\ exact solutions
of dispersionless Nyzhnyk equation
\par}

\vspace{4mm}\par\noindent{\large
Oleksandra O.\ Vinnichenko$^\dag$, Vyacheslav M.\ Boyko$^{\dag\ddag}$ and Roman O.\ Popovych$^{\dag\S}$
}

\vspace{4mm}\par\noindent{\it\small
$^\dag$\,Institute of Mathematics of NAS of Ukraine, 3 Tereshchenkivska Str., 01024 Kyiv, Ukraine
\par}

\vspace{2mm}\par\noindent{\it\small
$^\ddag$\,Department of Mathematics, Kyiv Academic University, 36 Vernads'koho Blvd., 03142 Kyiv, Ukraine
\par}

\vspace{2mm}\par\noindent{\it\small
$^\S$\,Mathematical Institute, Silesian University in Opava, Na Rybn\'\i{}\v{c}ku 1, 746 01 Opava, Czech Republic
\par}

\vspace{4mm}\par\noindent{\small
E-mails:
oleksandra.vinnichenko@imath.kiev.ua,
boyko@imath.kiev.ua,
rop@imath.kiev.ua
\par}

\vspace{7mm}\par\noindent\hspace*{5mm}\parbox{150mm}{\small
Among Lie submodels of the (real symmetric potential) dispersionless Nyzhnyk equation,
we single out a remarkable submodel as such that, despite being the only one,
is associated with a family of in general inequivalent one-dimensional subalgebras
of the maximal Lie invariance algebra of this equation,
which are parameterized by an arbitrary function of the time variable.
The large family of invariant solutions of the dispersionless Nyzhnyk equation that are related to the above submodel
is expressed in terms of an arbitrary function of the time variable
and the double quadrature of the well-known (implicit) general solution
of the inviscid Burgers equation with respect to a space-like submodel invariant variable.
The singled out submodel possesses many other interesting properties.
In particular, we show that it is Lie-remarkable,
and its maximal Lie invariance algebra completely defines its point symmetry pseudogroup,
which provides the second but simpler example of the latter phenomenon in literature.
Moreover, only hidden Lie symmetries of the dispersionless Nyzhnyk equation
that are associated with this submodel are essential for finding its exact solutions.
Using Lie reductions, we construct new families of exact solutions
of the inviscid Burgers equation and the dispersionless Nyzhnyk equation in closed or parametric form.
We also exhaustively described generalized symmetries, cosymmetries and conservation laws of the submodel,
which gives the corresponding nonlocal and hidden structures for
the inviscid Burgers equation and the dispersionless Nyzhnyk equation, respectively.\looseness=-1
}\par\vspace{4mm}

\section{Introduction}%\label{sec:Introduction}

One of the first integrable systems of differential equations in more
than two independent variables was the (1+2)-dimensional Nyzhnyk%
\footnote{%
The commoner version of transliterating this surname in the literature is ``Nizhnik'',
but the version used in the paper is more correct.
}
system suggested in~\cite[Eq.~(4)]{nizh1980a}.
By introducing potentials, the real symmetric version of this system is reduced
to a (1+2)-dimensional single partial differential equation,
which is called the (real symmetric potential) Nyzhnyk equation.
Using the technique of limit transitions to dispersionless counterparts of
(1+2)-dimensional integrable differential equations
and of the corresponding Lax representations~\cite[p.~167]{zakh1994a},
it is easy to show
that the (real symmetric potential) dispersionless Nyzhnyk equation
\begin{gather}\label{eq:dN}
u_{txy}=(u_{xx}u_{xy})_x + (u_{xy}u_{yy})_y,
\end{gather}
which is the dispersionless counterpart of the above Nyzhnyk equation,
possesses a nonlinear Lax representation~\cite{pavl2006a}.

Correcting, enhancing and significantly extending results from~\cite{moro2021a},
in~\cite{boyk2024a,vinn2024a} we carried out classical symmetry analysis
of the equation~\eqref{eq:dN}.
The maximal Lie invariance pseudoalgebra~$\mathfrak g$ of~\eqref{eq:dN}
is infinite-dimensional and is spanned by the vector fields
\begin{gather*}%\label{eq:dNMIA}
\begin{split}&
D^t(\tau)=\tau\p_t+\tfrac13\tau_tx\p_x+\tfrac13\tau_ty\p_y-\tfrac1{18}\tau_{tt}(x^3+y^3)\p_u,\quad
D^{\rm s}=x\p_x+y\p_y+3u\p_u,\\ &
P^x(\chi)=\chi\p_x-\tfrac12\chi_tx^2\p_u,\quad
P^y(\rho)=\rho\p_y-\tfrac12\rho_ty^2\p_u,\\ &
R^x(\alpha)=\alpha x\p_u,\quad
R^y(\beta)=\beta y\p_u,\quad
Z(\sigma)=\sigma\p_u,
\end{split}
\end{gather*}
where $\tau$, $\chi$, $\rho$, $\alpha$, $\beta$ and $\sigma$
run through the set of smooth functions of~$t$,
see~\cite{boyk2024a,moro2021a}.
The point-symmetry pseudogroup~$G$ of the equation~\eqref{eq:dN}
was computed in~\cite[Theorem~2]{boyk2024a} using
the original megaideal-based version of the algebraic method that was suggested in~\cite{malt2024a}.
It is generated by the transformations of the form
\begin{gather}\label{eq:dNPointSymForm}
\begin{split}
&\tilde t=T(t),\quad
\tilde x=CT_t^{1/3}x+X^0(t),\quad
\tilde y=CT_t^{1/3}y+Y^0(t),\\
&\tilde u=C^3u-\frac{C^3T_{tt}}{18T_t}(x^3+y^3)
-\frac{C^2}{2T_t^{1/3}}(X^0_tx^2+Y^0_ty^2)+W^1(t)x+W^2(t)y+W^0(t)
\end{split}
\end{gather}
and the transformation
$\mathscr J$: $\tilde t=t$, $\tilde x=y$, $\tilde y=x$, $\tilde u=u$.
%\[\mathscr J\colon\ \tilde t=t,\quad \tilde x=y,\quad \tilde y=x,\quad \tilde u=u.\]
Here $T$, $X^0$, $Y^0$, $W^0$, $W^1$ and $W^2$ are arbitrary smooth functions of~$t$
with $T_t\neq0$, and $C$ is an arbitrary nonzero constant.
The identity component~$G_{\rm id}$ of the pseudogroup~$G$
consists of the transformations of the form~\eqref{eq:dNPointSymForm} with $T_t>0$ and $C>0$.

We can single out subgroups of~$G$ each of which is parameterized by a single parameter
among the constant and functional parameters
involved in the representation~\eqref{eq:dNPointSymForm} for transformations from~$G$.
For this purpose, we should set all of these parameters, except the single associated one,
to the trivial values corresponding to the identity transformation,
which are 1, $t$, 0, 0, 0, 0 and~0 for $C$, $T$, $X^0$, $Y^0$, $W^0$, $W^1$ and $W^2$,
thus obtaining the subgroups
\[
\{\mathscr D^t(T)\},\ \ \{\mathscr D^{\rm s}(C)\},\ \
\{\mathscr P^x(X^0)\},\ \ \{\mathscr P^y(Y^0)\},\ \
\{\mathscr R^x(W^1)\},\ \ \{\mathscr R^y(W^2)\}\ \ \mbox{and}\ \ \{\mathscr Z(W^0)\}
\]
of~$G$
associated with the subalgebras~$\{D^t(\tau)\}$, $\langle D^{\rm s}\rangle$,
$\{P^x(\chi)\}$, $\{P^y(\rho)\}$, $\{R^x(\alpha)\}$,
$\{R^y(\beta)\}$ and $\{Z(\sigma)\}$ of~$\mathfrak g$, respectively.
Here all the parameter functions run through the specified sets of their values.
We call transformations from these subgroups and the transformation~$\mathscr J$
elementary point symmetry transformations of the equation~\eqref{eq:dN}.

A complete list of discrete point symmetry transformations of the equation~\eqref{eq:dN}
that are independent up to composing with each other and with transformations from~$G_{\rm id}$
is exhausted by three commuting involutions, which can be chosen to be
the permutation~$\mathscr J$ of the variables~$x$ and~$y$
and two transformations $\mathscr I^{\rm i}:=\mathscr D^t(-t)$ and $\mathscr I^{\rm s}:=\mathscr D^{\rm s}(-1)$
changing the signs of $(t,x,y)$ and of $(x,y,u)$, respectively.

The above-mentioned results of~\cite{boyk2024a} created
a basis for the exhaustive classification of Lie reductions of~\eqref{eq:dN}
to partial differential equations in two independent variables
and to ordinary differential equations in~\cite{vinn2024a}.
Among the listed inequivalent subalgebras of~$\mathfrak g$,
there is a family of one-dimensional subalgebras such that
the corresponding reduced equations are of the same form,
which further reduces to the inviscid Burgers equation%
\footnote{\label{fnt:InviscidBurgersEq}
The inviscid Burgers equation~\eqref{eq:dNs1.3RhoNe1InviscidBurgersEq}
is the simplest nonlinear transport equation, called also Hopf's equation.
It possesses the well-known implicit representation of its general solution
${F(h,z_2-hz_1)=0}$ with an arbitrary nonconstant sufficiently smooth function~$F$,
see \cite[Chapter~E, Eq.~2.44]{kamk1959B} and \cite[Section~1.1.1.18]{poly2012A}.
}
\begin{gather}\label{eq:dNs1.3RhoNe1InviscidBurgersEq}
h_1+hh_2=0.
\end{gather}
Here and in what follows,
the subscripts~1 and~2 of functions denote the differentiations
with respect to~$z_1$ and~$z_2$, respectively,
and each function is considered as its zero-order derivative.
More specifically, the list of $G$-inequivalent one-dimensional subalgebras of~$\mathfrak g$
from~\cite[Lemma~5]{vinn2024a} in particular contains the family of subalgebras
\[
\mathfrak s_{1.3}^\rho=\big\langle P^x(1)+P^y(\rho)\big\rangle,
\]
where $\rho=\rho(t)$ is an arbitrary smooth function of~$t$
with $\rho(t)\ne0$ for any~$t$ in the domain of~$\rho$
and $\rho\not\equiv1$ on each open interval of the domain of~$\rho$.
Within this family, subalgebras~$\mathfrak s_{1.3}^\rho$ and~$\smash{\mathfrak s_{1.3}^{\tilde\rho}}$
are $G$-inequivalent if and only if
$\tilde\rho(t)=\rho(at+b)$ for some $a,b\in\mathbb R$ or
$\tilde\rho=(\rho(\hat T))^{-1}$, where $\hat T$ is the inverse of a solution~$T$ of the equation $T_t=c\rho^{-3}$ for some $c\in\mathbb R$.
In the context of the discussion in \cite[Section~B]{fush1994b} and~\cite{popo1995b},
the optimal ansatz constructed using the subalgebra~$\mathfrak s_{1.3}^\rho$
with a fixed appropriate value of the parameter function~$\rho$ is
\begin{gather}\label{eq:dNs1.3RhoNe1ModifiedAnsatz}
u=w(z_1,z_2)-\frac{\rho_t}{6\rho}y^3, \quad
z_1=2\int\frac{\rho^3-1}{\rho^3}{\rm d}t,\quad
z_2=\frac y{\rho}-x,
\end{gather}
where the integral denotes a fixed antiderivative of the integrand.
Any of the ansatzes of this form reduces the equation~\eqref{eq:dN} to the partial differential equation
\begin{gather}\label{eq:dNs1.3RhoNe1ModifiedRedEq}
w_{122}+w_{22}w_{222}=0
\end{gather}
in two independent variables,
which was called the modified reduced equation~1.3 in~\cite{vinn2024a},
see \mbox{\cite[Eq.~(11)]{vinn2024a}}.
One can see that due to the choice of appropriate ansatzes,
the (in general) $G$-inequivalent subalgebras~$\mathfrak s_{1.3}^\rho$
result in reduced equations of the same form~\eqref{eq:dNs1.3RhoNe1ModifiedRedEq}.
This phenomenon was explicitly indicated for the first time in~\cite{vinn2024a} using the above reduction.
Hereafter the equation~\eqref{eq:dNs1.3RhoNe1ModifiedRedEq} is also called reduced equation~1.3
to relate the results of the present paper with those of~\cite{vinn2024a}.
Objects unambiguously associated with this equation, like various invariance algebras and symmetry (pseudo)groups,
are marked by the subscript ``1.3''.

It is obvious that the substitution~$w_{22}=h$ maps the reduced equation~\eqref{eq:dNs1.3RhoNe1ModifiedRedEq}
to the inviscid Burgers equation~\eqref{eq:dNs1.3RhoNe1InviscidBurgersEq}.
Combining this observation with the ansatz~\eqref{eq:dNs1.3RhoNe1ModifiedAnsatz}
and the known implicit general solution of the inviscid Burgers equation~\eqref{eq:dNs1.3RhoNe1InviscidBurgersEq},
see footnote~\ref{fnt:InviscidBurgersEq},
we obtain a large family of exact solutions of the dispersionless Nyzhnyk equation~\eqref{eq:dN},
\begin{gather}\label{eq:dNs1.3InvarSolutions}
\solution u=
%\int\left(\int h(z_1,z_2)\,{\rm d}z_2\right){\rm d}z_2
\int^{z_2}(z_2-s) h(z_1,s)\,{\rm d}s
-\frac{\rho_t}{6\rho}y^3, \quad
z_1:=2\int\frac{\rho^3-1}{\rho^3}{\rm d}t,\quad
z_2:=\frac y\rho-x.
\end{gather}
Here $\rho$ is an arbitrary sufficiently smooth function of~$t$
that is not equal to the constant functions~0 and~1,
and the function $h=h(z_1,z_2)$ is implicitly defined by the equation \[{F(h,z_2-hz_1)=0}\]
with an arbitrary nonconstant sufficiently smooth function~$F$ of its arguments.
Up to the $G$-equivalence, the integral with respect to~$z_2$ can be considered
as a fixed second antiderivative of~$h$ with respect to this variable.
The number of quadratures in~\eqref{eq:dNs1.3InvarSolutions} can be reduced by one
if we denote $\vartheta:=2\int(1-\rho^{-3})\,{\rm d}t$ and substitute $\rho=(1-\frac12\vartheta_t)^{-1/3}$,
when assuming $\vartheta$ as an arbitrary sufficiently smooth function of~$t$ instead of~$\rho$,
where the derivative~$\vartheta_t$ is not equal to the constant functions~0 and~2.
Nevertheless, even after replacement~$\rho$ by~$\vartheta$,
the formula~\eqref{eq:dNs1.3InvarSolutions} cannot be considered as a convenient representation
for exact solutions of the dispersionless Nyzhnyk equation~\eqref{eq:dN}
since it still contains the quadrature with an implicitly defined \mbox{function}.\looseness=1

This is why to single out those solutions in the family~\eqref{eq:dNs1.3InvarSolutions}
that admit simpler representations, see Theorem~\ref{thm:dNRedEq1.3CorrInvSolutions},
in the present paper
we carry out the Lie reduction procedure for the equation~\eqref{eq:dNs1.3RhoNe1ModifiedRedEq}
as the second step of the Lie reduction procedure for the equation~\eqref{eq:dN}
with using the subalgebras~$\mathfrak s_{1.3}^\rho$ for the first step of reduction.
In fact, we still mostly work with the inviscid Burgers equation~\eqref{eq:dNs1.3RhoNe1InviscidBurgersEq}
instead of the equation~\eqref{eq:dNs1.3RhoNe1ModifiedRedEq}.

We also comprehensively study local symmetry-like objects of the equation~\eqref{eq:dNs1.3RhoNe1ModifiedRedEq},
which includes its $z_2$-integrals, generalized symmetries, cosymmetries,
conserved currents, conservation-law characteristics and conservation laws,
and relate them to their counterparts for the inviscid Burgers equation~\eqref{eq:dNs1.3RhoNe1InviscidBurgersEq}.
It turns out for each of the above kind of local symmetry-like objects,
the equation~\eqref{eq:dNs1.3RhoNe1ModifiedRedEq} admits such objects of arbitrarily high order.
In total, this gives one more example,
in addition to only a few ones, of a comprehensive study of local symmetry-like objects
for a system of partial differential equations arising in applications. %\looseness=-1

More specifically,
the structure of the maximal Lie invariance pseudoalgebra~$\mathfrak a_{1.3}$ of reduced equation~1.3
including its megaideals is analyzed in Section~\ref{sec:dNRedEq1.3MIA}.

Using the essential megaideals among the found ones, in Section~\ref{sec:dNRedEq1.3PointSymGroup}
we apply, to the equation~\eqref{eq:dNs1.3RhoNe1ModifiedRedEq}, the megaideal-based version of the algebraic method
of constructing point-symmetry (pseudo)groups of systems of differential equations
that was suggested in~\cite{malt2024a} and developed in~\cite{boyk2024a}.
It turns out that the point symmetry pseudogroup~$G_{1.3}$ of~\eqref{eq:dNs1.3RhoNe1ModifiedRedEq}
has a remarkable property.
The algebraic condition that the pushforward~$\Phi_*$ of the algebra~$\mathfrak a_{1.3}$
by any element~$\Phi$ of~$G_{1.3}$ preserves this algebra, $\Phi_*\mathfrak a_{1.3}=\mathfrak a_{1.3}$,
completely defines the pseudogroup~$G_{1.3}$.
Therefore, the direct method is needed only to verify that the pseudogroup~$G_{1.3}$
is indeed the entire point-symmetry (pseudo)group of~\eqref{eq:dNs1.3RhoNe1ModifiedRedEq}.
After~\cite{boyk2024a}, this is the second but much simpler example of this kind in the literature.

Inspired by finding the above phenomenon, we study other defining properties of Lie symmetries
of the equation~\eqref{eq:dNs1.3RhoNe1ModifiedRedEq} in Section~\ref{sec:dNs1.3RhoNe1DefPropertiesOfLieSyms}.
We prove that
this equation is Lie-remarkable since it itself is completely defined
by 11- and 12-dimensional subalgebras of the algebra~$\mathfrak a_{1.3}$
in the classes of genuine and general partial differential equations of order not greater than three
in two independent variables, respectively,
whereas a six-dimensional subalgebra of the former subalgebra suffices
to define the local diffeomorphisms that stabilize the algebra~$\mathfrak a_{1.3}$.

In Section~\ref{sec:dNRedEq1.3OnInduction},
we study the induction of Lie and point symmetries of the reduced equation~\eqref{eq:dNs1.3RhoNe1ModifiedRedEq}
by their counterparts for the original equation~\eqref{eq:dN}.
This gives the first example of studying the induction of point symmetries,
including discrete ones, in the course of a Lie reduction,
which is a more complicated problem than that of inducing Lie symmetries.

In Section~\ref{sec:dNs1.3RhoNe1Subalgebras},
we classify, up to the $G_{1.3}$-equivalence,
one-dimensional subalgebras of~$\mathfrak a_{1.3}$ that are appropriate
for Lie reduction of the equation~\eqref{eq:dNs1.3RhoNe1ModifiedRedEq}.
This classification was carried out by means of reducing it
to the classification of one-dimensional subalgebras of the algebra~$\check{\mathfrak a}_{1.3}$
up to the \smash{$\check G_{1.3}$}-equivalence, which was presented in \cite[Table~2]{poch2017a}.
Here $\check{\mathfrak a}_{1.3}$ denotes
the algebra of Lie-symmetry vector fields of~\eqref{eq:dNs1.3RhoNe1InviscidBurgersEq}
that are induced by Lie-symmetry vector fields of~\eqref{eq:dNs1.3RhoNe1ModifiedRedEq},
see the end of Section~\ref{sec:dNRedEq1.3MIA}.
Analogously, \smash{$\check G_{1.3}$} denotes
the group of point symmetry transformations of~\eqref{eq:dNs1.3RhoNe1InviscidBurgersEq}
that are induced by point symmetry transformations of~\eqref{eq:dNs1.3RhoNe1ModifiedRedEq},
see the end of Section~\ref{sec:dNRedEq1.3PointSymGroup}.

Large families of Lie invariant solutions of
the equations~\eqref{eq:dNs1.3RhoNe1InviscidBurgersEq} and~\eqref{eq:dNs1.3RhoNe1ModifiedRedEq}
are constructed in Section~\ref{sec:dNs1.3LieInvSols}
in explicit form in terms of elementary functions and the Lambert $W$ function
as well as in parametric form.
To simplify the treatment,
we replace the Lie reduction procedure for the equation~\eqref{eq:dNs1.3RhoNe1ModifiedRedEq}
by that for the equation~\eqref{eq:dNs1.3RhoNe1InviscidBurgersEq}
and obtain Lie invariant solutions of~\eqref{eq:dNs1.3RhoNe1ModifiedRedEq}
by integrating twice the obtained invariant solutions of~\eqref{eq:dNs1.3RhoNe1InviscidBurgersEq} with respect to~$z_2$
and neglecting trivial summands of the form $\breve W^1(z_1)z_2+\breve W^0(z_1)$ arising in the course of the integration
due to the $G_{1.3}$-inequivalence on the solution set of~\eqref{eq:dNs1.3RhoNe1ModifiedRedEq}.
Here $\breve W^1$ and $\breve W^0$ are arbitrary sufficiently smooth functions of~$z_1$.
Then we complete the application
of the optimized procedure of step-by-step reductions involving hidden symmetries
\cite[Section~B]{kova2023b} in Theorem~\ref{thm:dNRedEq1.3CorrInvSolutions},
presenting the form of the corresponding solutions of the dispersionless Nyzhnyk equation~\eqref{eq:dN}.
This form is obtained
by extending solutions of the reduced equation~\eqref{eq:dNs1.3RhoNe1ModifiedRedEq}
by noninduced point symmetries of this equation
and substituting them into the ansatz~\eqref{eq:dNs1.3RhoNe1ModifiedAnsatz}. %\looseness=-1

Local symmetry-like objects associated with the equation~\eqref{eq:dNs1.3RhoNe1ModifiedRedEq}
are studied in Section~\ref{sec:dNs1.3RhoNe1SymLikeObjects}.
This includes the exhaustive descriptions of its
$z_2$-integrals (Section~\ref{sec:dNs1.3RhoNe1Integrals}),
generalized symmetries (Section~\ref{sec:dNs1.3RhoNe1GenSyms}),
cosymmetries (Section~\ref{sec:dNs1.3RhoNe1Cosyms}) and
conserved currents, conservation-law characteristics and conservation laws (Section~\ref{sec:dNs1.3RhoNe1CLs}).
Auxiliary statements on the general solutions of certain differential equations
for differential functions associated with the above objects
are collected in Section~\ref{sec:dNs1.3RhoNe1AuxiliaryResults}.
In Section~\ref{sec:InviscidBurgersEqInducedObjects},
we establish relations between
the local symmetry-like objects of the equation~\eqref{eq:dNs1.3RhoNe1ModifiedRedEq}
and their counterparts for the inviscid Burgers equation~\eqref{eq:dNs1.3RhoNe1InviscidBurgersEq}.
In the context of Section~\ref{sec:InviscidBurgersEqInducedObjects}
and the entire Section~\ref{sec:dNs1.3RhoNe1SymLikeObjects},
it is natural to also study local symmetry-like objects of
the equation $q_{12}+q_2q_{22}=0$, which is intermediate between
the equations~\eqref{eq:dNs1.3RhoNe1InviscidBurgersEq} and~\eqref{eq:dNs1.3RhoNe1ModifiedRedEq}
in the chain of the differential substitutions $h=q_2$ and $q=w_2$.
This study is carried out in Section~\ref{sec:IntermediateEq}.\looseness=1

The final Section~\ref{sec:Conclusion} contains a discussion of the results of the present paper
together with some possible further research perspective.

\section{Maximal Lie invariance algebra}\label{sec:dNRedEq1.3MIA}

The maximal Lie invariance pseudoalgebra~$\mathfrak a_{1.3}$ of the reduced equation~\eqref{eq:dNs1.3RhoNe1ModifiedRedEq}
was computed in~\cite{vinn2024a}
using the packages {\sf DESOLV} \cite{carm2000a} and {\sf Jets} \cite{BaranMarvan} for {\sf Maple};
the latter package is based on algorithms developed in~\cite{marv2009a}.
This pseudoalgebra is spanned by the vector fields
\begin{gather}\label{eq:dNMIASubmodel}
\begin{split}&
P^1=\p_{z_1},\quad
D^1=z_1\p_{z_1}-w\p_w,\quad
%\\&\todo \tilde D^1=D^1+\tfrac12 D^2=z_1\p_{z_1}+\tfrac12z_2\p_{z_2}+\tfrac12w\p_w,\quad
K=z_1^2\p_{z_1}+z_1z_2\p_{z_2}+(z_1w+\tfrac16z_2^{\,3})\p_w,\\&
D^2=z_2\p_{z_2}+3w\p_w,\quad
P^2=\p_{z_2},\quad
H=z_1\p_{z_2}+\tfrac12z_2^{\,2}\p_w,\\&
R(\alpha)=\alpha(z_1)z_2\p_w,\quad
Z(\sigma)=\sigma(z_1)\p_w.
\end{split}
\end{gather}
Here and in what follows,
the functional parameters~$\alpha$, $\beta$ and~$\sigma$ run through the set of smooth functions of
a single argument, $t$ or~$z_1$ depending on the context.
A computation by {\sf Jets} also shows that
the contact invariance algebra~$\mathfrak a_{1.3\rm c}$ of~\eqref{eq:dNs1.3RhoNe1ModifiedRedEq}
coincides with the first prolongation of~$\mathfrak a_{1.3}$.

Up to the antisymmetry of the Lie bracket,
the nonzero commutation relations between the vector fields~\eqref{eq:dNMIASubmodel} spanning~$\mathfrak a_{1.3}$
are exhausted by
\begin{gather*}%\label{eq:dNCommRelationsSubmodel}
%\begin{split}&
[P^1,D^1]=P^1,\quad
[P^1,K]=2D^1+D^2,\quad
[D^1,K]=K,\\ %[.5ex]%&
[P^1,H]=P^2,\quad
[P^1,R(\alpha)]=R(\alpha_{z_1}),\quad
[P^1,Z(\sigma)]=Z(\sigma_{z_1}),\\ %[.5ex]%&
[D^1,H]=H,\quad
[D^1,R(\alpha)]=R(z_1\alpha_{z_1}+\alpha),\quad
[D^1,Z(\sigma)]=Z(z_1\sigma_{z_1}+\sigma),\\ %[.5ex]%&
[K,P^2]=-H,\quad
[K,R(\alpha)]=R(z_1^2\alpha_{z_1}),\quad
[K,Z(\sigma)]=Z(z_1^2\sigma_{z_1}-z_1\sigma),\\ %[.5ex]%&
[D^2,P^2]=-P^2,\quad
[D^2,H]=-H,\quad
[D^2,R(\alpha)]=-2R(\alpha),\quad
[D^2,Z(\sigma)]=-3Z(\sigma),\\ %[.5ex]%&
[P^2,H]=R(1),\quad
[P^2,R(\alpha)]=Z(\alpha),\quad
[H,R(\alpha)]=Z(z_1\alpha).
%\end{split}
\end{gather*}

The commutation relations imply that the Lie pseudoalgebra~$\mathfrak a_{1.3}$
is the sum of its nine-dimensional Lie subalgebra
$\mathfrak a_{1.3}^{\rm ess}:=\langle P^1,D^1,K,D^2,P^2,H,R(1),Z(1),Z(z_1)\rangle$
and its infinite-dimensional abelian (pseudo)ideal
$\mathfrak a_{1.3}^{\rm triv}:=\langle R(\alpha),Z(\sigma)\rangle$,
$\mathfrak a_{1.3}=\mathfrak a_{1.3}^{\rm ess}+\mathfrak a_{1.3}^{\rm triv}$,
where $\mathfrak a_{1.3}^{\rm ess}\cap\mathfrak a_{1.3}^{\rm triv}=\langle R(1),Z(1),Z(z_1)\rangle$.
In fact, only the subalgebra~$\mathfrak a_{1.3}^{\rm ess}$ is essential
in the course of classifying Lie reductions of the equation~\eqref{eq:dNs1.3RhoNe1ModifiedRedEq}.
This is why we call $\mathfrak a_{1.3}^{\rm ess}$ the essential subalgebra of~$\mathfrak a_{1.3}$.

To compute the point-symmetry pseudogroup~$G_{1.3}$ of reduced equation~1.3 given by~\eqref{eq:dNs1.3RhoNe1ModifiedRedEq}
using the algebraic method, we construct \emph{megaideals} of the algebra~$\mathfrak a_{1.3}$,
i.e., linear subspaces of~$\mathfrak a_{1.3}$
that are stable under action of the automorphism group of~$\mathfrak a_{1.3}$~\cite{bihl2015a,popo2003a}.

Given a Lie algebra~$\mathfrak g$,
by $\mathfrak z(\mathfrak g)$, $\mathfrak g'$, $\mathfrak g''$, $\mathfrak g^k$, $k\in\mathbb N$,
and $\mathfrak g_{(k)}$, $k\in\mathbb N_0:=\mathbb N\cup\{0\}$, in this section
we denote the center, the derived algebra, the second derived algebra,
the $k$th Lie power and the $k$th element of the upper central series
of~$\mathfrak g$, respectively,
$\mathfrak z(\mathfrak g):=\{v\in\mathfrak g\mid [v,w]=0\ \forall w\in\mathfrak g\}$,
$\mathfrak g':=[\mathfrak g,\mathfrak g]$,
$\mathfrak g'':=[\mathfrak g',\mathfrak g']$,
$\mathfrak g^1:=\mathfrak g$, $\mathfrak g^{k+1}:=[\mathfrak g,\mathfrak g^k]$, $k\in\mathbb N$,
$\mathfrak g_{(0)}:=\{0\}$ and $\mathfrak g_{(k+1)}/\mathfrak g_{(k)}$ is a center of $\mathfrak g/\mathfrak g_{(k)}$,
$k\in\mathbb N_0$.
In particular, $\mathfrak g_{(1)}=\mathfrak z(\mathfrak g)$ and $\mathfrak g^2=\mathfrak g'$.
All the listed subalgebras of~the algebra~$\mathfrak g$ as well as its radical
are its megaideals \cite{bihl2015a,popo2003a}.
In view of the commutation relations of~$\mathfrak a_{1.3}$, the only following megaideal is obvious:
\[
\mathfrak m_1:=\mathfrak a_{1.3}'=\big\langle P^1,\,2D^1+D^2,\,K,\,P^2,\,H,\,R(\alpha),\,Z(\sigma)\big\rangle.
\]

To find other megaideals of~$\mathfrak a_{1.3}$, we prove the following assertion.

\begin{lemma}\label{lem:radicalOfg}
The radical $\mathfrak r$ of~$\mathfrak a_{1.3}$ coincides with
$\big\langle D^2,\,P^2,\,H,\,R(\alpha),\,Z(\sigma)\big\rangle$.
\end{lemma}

\begin{proof}
We use ideas from the proof of Lemma~1 in~\cite{malt2024a} and
denote the span from lemma's statement by~$\mathfrak s$.
To prove that $\mathfrak r=\mathfrak s$,
it suffices to show that $\mathfrak s$ is the maximal solvable ideal of~$\mathfrak a_{1.3}$.
The commutation relations %~\eqref{eq:dNCommRelationsSubmodel}
between the vector fields spanning $\mathfrak a_{1.3}$
imply that $\mathfrak s$ is an ideal of~$\mathfrak a_{1.3}$.
The third derived algebra~$\mathfrak s^{(3)}$ of~$\mathfrak s$ is equal to $\{0\}$,
and thus the ideal $\mathfrak s$ is solvable (of solvability rank three).
It remains to check that the solvable ideal $\mathfrak s$ of~$\mathfrak a_{1.3}$ is maximal in~$\mathfrak a_{1.3}$.

Consider an ideal $\mathfrak s_1$ of $\mathfrak a_{1.3}$ that properly contains $\mathfrak s$.
Then a vector field~$Q$ of the form $Q=c_1P^1+c_2D^1+c_3K$ with $(c_1,c_2,c_3)\ne(0,0,0)$ belongs to~$\mathfrak s_1$.
Since $\mathfrak s_1$ is an ideal of $\mathfrak a_{1.3}$, the commutators
$[Q,P^1]=-c_2P^1-c_3(2D^1+D^2)$,
$[Q,D^1]=c_1P^1-c_3K$ and
$[Q,K]=c_1(2D^1+D^2)+c_2K$ belong to~$\mathfrak s_1$ as well.
Successively commuting each of these commutators with~$P^1$, $D^1$ and~$K$
and linearly recombining the obtained elements,
we derive that $c_iP^1$, $c_i(2D^1+D^2)$ and $c_iK$ also belong to~$\mathfrak s_1$ for any $i\in\{1,2,3\}$,
which means that $P^1,2D^1+D^2,K\in\mathfrak s_1$, i.e., $\mathfrak s_1=\mathfrak a_{1.3}$.
Since the algebra~$\mathfrak a_{1.3}$ is not solvable,
the span~$\mathfrak s$ is maximal as a solvable ideal of~$\mathfrak a_{1.3}$.
\end{proof}

Thus, $\mathfrak a_{1.3}=\mathfrak f\lsemioplus\mathfrak r$,
where $\mathfrak f=\big\langle P^1,2D^1+D^2,K\big\rangle$ is a ``Levi subalgebra'' of~$\mathfrak a_{1.3}$,
which is isomorphic to ${\rm sl}(2,\mathbb R)$.
We set $\mathfrak m_2:=\mathfrak r$.
Knowing~$\mathfrak r$ and using properties of megaideals~\cite{bihl2015a,popo2003a},
we can easily construct several other proper megaideals of the algebra~$\mathfrak a_{1.3}$,
\begin{gather*}
\mathfrak m_3:=\mathfrak m_2'=\mathfrak m_1\cap\mathfrak m_2=\big\langle P^2,H,R(\alpha),Z(\sigma)\big\rangle,\\
\mathfrak m_4:=\mathfrak m_{2(2)}=\big\langle R(\alpha),Z(\sigma)\big\rangle,\quad
\mathfrak m_2''=\big\langle R(1),Z(\sigma)\big\rangle,\quad
\mathfrak m_5:=\mathfrak z(\mathfrak m_3)=\big\{ Z(\sigma)\big\},\\
\mathfrak m_6:=\big\langle R(1),Z(1),Z(z_1)\big\rangle,\quad
\mathfrak m_7:=(\mathfrak m_2')^3=\big\langle Z(1),Z(z_1)\big\rangle.
\end{gather*}
In particular, to find the megaideal~$\mathfrak m_6$,
we use Proposition~1 from~\cite{card2013a} with
$\mathfrak i_0=\mathfrak m_2''$, $\mathfrak i_1=\mathfrak m_1$ and~$\mathfrak i_2=\mathfrak m_7$.

Overall, for the algebra~$\mathfrak a_{1.3}$ we obtain the proper megaideal $\mathfrak m_1=\mathfrak a_{1.3}'$
and the hierarchy
\[
\mathfrak r=:\mathfrak m_2\varsupsetneq\mathfrak m_3
\varsupsetneq\mathfrak m_4\varsupsetneq\mathfrak m_2''
\varsupsetneq\begin{array}{c}\mathfrak m_5\\\mathfrak m_6\end{array}
\varsupsetneq\mathfrak m_7
\]
of proper megaideals contained in its radical~$\mathfrak r$.
The only proper megaideal~$\mathfrak m_2''$ and the entire algebra~$\mathfrak a_{1.3}$ as its improper nonzero megaideal
is the sum of other found proper megaideals,
$\mathfrak m_2''=\mathfrak m_5+\mathfrak m_6$ and $\mathfrak a_{1.3}=\mathfrak m_1+\mathfrak m_2$.
This is not the case for the other listed megaideals,
and, therefore, they can be essential in the course of computing
the point-symmetry pseudogroup of the equation~\eqref{eq:dNs1.3RhoNe1ModifiedRedEq}
by the algebraic method.
Among them, only the megaideals~$\mathfrak m_6$ and~$\mathfrak m_7$
are finite-dimensional and, moreover, they are respectively three- and two-dimensional.
Note that within the framework of the above elementary approach,
we cannot check whether or not
the entire set of proper megaideals of the (infinite-dimensional) algebra~$\mathfrak a_{1.3}$
is exhausted by the megaideals~$\mathfrak m_j$, $j=1,\dots,7$, and~$\mathfrak m_2''$.

The maximal Lie invariance algebra~$\mathfrak a_{\rm iB}$ of the equation~\eqref{eq:dNs1.3RhoNe1InviscidBurgersEq}
is much wider than the algebra~$\mathfrak a_{1.3}$.
More specifically,
\[
\mathfrak a_{\rm iB}=\big\langle
\theta(z_1,z_2,h)(\p_{z_1}+h\p_{z_2}),\,
\varphi(h,z_2-hz_1)\p_{z_2},\,
\psi(h,z_2-hz_1)(z_1\p_{z_2}+\p_h)
\big\rangle,
\]
where $\theta$, $\varphi$ and~$\psi$ run through the sets of smooth functions
of the corresponding arguments, see~\cite{katk1965a} or \cite[Section~11.2]{CRC_v1}.
The differential substitution~$w_{22}=h$ induces a homomorphism~$\boldsymbol\upsilon$
of the algebra~$\mathfrak a_{1.3}$ into the algebra~$\mathfrak a_{\rm iB}$,
which can be represented as the composition of the second prolongation of the vector fields from~$\mathfrak a_{1.3}$
with the pushforward of the prolonged vector fields by the natural projection
from ${\rm J}^2(\mathbb R^2_{z_1,z_2}\times\mathbb R_w)$ onto $\mathbb R^2_{z_1,z_2}\times\mathbb R_{w_{22}}$
and substituting $w_{22}=h$.
Thus, the homomorphism~$\boldsymbol\upsilon$ maps
the Lie-symmetry vector fields $P^1$, $D^1$, $K$, $D^2$, $P^2$, $H$, $R(\alpha)$ and~$Z(\sigma)$
of the equation~\eqref{eq:dNs1.3RhoNe1ModifiedRedEq}
to the Lie-symmetry vector fields $\check P^1$, $\check D^1$, $\check K$, $\check D^2$, $\check P^2$, $\check H$, $0$ and~$0$
of the equation~\eqref{eq:dNs1.3RhoNe1InviscidBurgersEq}, respectively, where
\begin{gather*}%\label{eq:dNMIASubmodelModified}
\begin{split}&
\check P^1=\p_{z_1},\quad
\check D^1=z_1\p_{z_1}-h\p_h,\quad
\check K=z_1^2\p_{z_1}+z_1z_2\p_{z_2}+(z_2-z_1h)\p_h,\\&
\check D^2=z_2\p_{z_2}+h\p_h,\quad
\check P^2=\p_{z_2},\quad
\check H=z_1\p_{z_2}+\p_h.
\end{split}
\end{gather*}
In other words, $\ker\boldsymbol\upsilon=\mathfrak a_{1.3}^{\rm triv}=\langle R(\alpha),Z(\sigma)\rangle$, and
$
\check{\mathfrak a}_{1.3}:=\mathop{\rm im}\boldsymbol\upsilon
=\langle\check P^1,\check D^1,\check K,\check D^2,\check P^2,\check H\rangle
$
is a~subalgebra of~$\mathfrak a_{\rm iB}$, which can be called
the algebra of Lie-symmetry vector fields of~\eqref{eq:dNs1.3RhoNe1InviscidBurgersEq}
that are induced by Lie-symmetry vector fields of~\eqref{eq:dNs1.3RhoNe1ModifiedRedEq}.
The algebra~$\check{\mathfrak a}_{1.3}$ coincides, up to notation of variables and vector fields,
with the algebra~$\mathfrak g$ from \cite[Section~3]{poch2017a}.
It is obvious that the algebra~$\check{\mathfrak a}_{1.3}$ is isomorphic to
the quotient algebra of the essential subalgebra~$\mathfrak a_{1.3}^{\rm ess}$ of~$\mathfrak a_{1.3}$
by its ideal $\langle R(1),Z(1),Z(z_1)\rangle=\mathfrak a_{1.3}^{\rm ess}\cap\ker\boldsymbol\upsilon$.
Note that the algebra~$\check{\mathfrak a}_{1.3}$ is also isomorphic to
the Lie algebra ${\rm aff}(2,\mathbb R)$ of the planar group ${\rm Aff}(2,\mathbb R)$.

\section{Point symmetry pseudogroup}\label{sec:dNRedEq1.3PointSymGroup}

\begin{theorem}\label{thm:dNRedEq1.3PointSymGroup}
The point-symmetry pseudogroup~$G_{1.3}$ of the equation~\eqref{eq:dNs1.3RhoNe1ModifiedRedEq}
is constituted by the transformations of the form
\begin{gather}\label{eq:dNs1.3RhoNe1ModifiedRedEqPointSymForm}
\begin{split}&
\tilde z_1=\frac{c_1z_1+c_2}{c_3z_1+c_4},\quad
\tilde z_2=\frac{z_2+c_5z_1+c_6}{c_3z_1+c_4},\\&
\tilde w=\frac{w}{\Delta(c_3z_1+c_4)}
-\frac{c_3}{\Delta(c_3z_1+c_4)^2}\frac{z_2^3}6-\frac{c_3c_6-c_4c_5}{\Delta(c_3z_1+c_4)^2}\frac{z_2^2}2
+W^1(z_1)z_2+W^0(z_1),
\end{split}
\end{gather}
where $c_1$, \dots, $c_6$ are arbitrary constants with $\Delta:=c_1c_4-c_2c_3\ne0$,
and $W^0$ and $W^1$ are arbitrary smooth functions of~$z_1$.
\end{theorem}

\begin{proof}
Since the maximal Lie invariance algebra~$\mathfrak a_{1.3}$
of the equation~\eqref{eq:dNs1.3RhoNe1ModifiedRedEq} is infinite-dimensional
and has a number of megaideals,
it is convenient to find the pseudogroup~$G$ using the modification of the megaideal-based method
that was suggested in~\cite{malt2024a}.%
\footnote{%
The method to be applied is called \emph{algebraic}
in contrast with the \emph{direct} method, which is described,
e.g., in~\cite{king1998a} and \cite[Sections~2.2 and 4]{bihl2011b}.
The first version of the algebraic method for computing the point-symmetry group
of a system of differential equations, which was suggested by Hydon \cite{hydo1998a,hydo1998b,hydo2000b,hydo2000A},
is based on knowing the automorphism group of the corresponding maximal Lie invariance algebra~$\mathfrak g$,
and hence it is applicable only if the algebra~$\mathfrak g$ is finite-dimensional,
and, moreover, its dimension is not too large.
This is why Hydon's approach can be reinforced using classical results
on finite-dimensional Lie algebras~\cite{kont2019a}.
The other version of the algebraic method involves less knowledge on the structure of~$\mathfrak g$,
which is just a collection of megaideals of~$\mathfrak g$ and can be obtained even if $\dim\mathfrak g=\infty$.
It was suggested for the first time in~\cite{bihl2011b} and was developed and efficiently applied
in~\cite{boyk2024a,card2011a,card2013a,card2021a,malt2024a,opan2020a}.
}
The consideration is based on the following fact as a necessary condition to be satisfied
by any point symmetry of the equation~\eqref{eq:dNs1.3RhoNe1ModifiedRedEq}.
If a point transformation~$\Phi$ in the space with the coordinates $(z_1,z_2,w)$,
\[
\Phi\colon\ (\tilde z_1,\tilde z_2,\tilde w)=(Z^1,Z^2,W)
\]
with a tuple of smooth functions $(Z^1,Z^2,W)$ of $(z_1,z_2,w)$ with nonvanishing Jacobian,
is a point symmetry of~\eqref{eq:dNs1.3RhoNe1ModifiedRedEq},
then the pushforward~$\Phi_*$ of vector fields by~$\Phi$ generates
an automorphism of the algebra~$\mathfrak g:=\mathfrak a_{1.3}$.
Hence $\Phi_*\mathfrak g\subseteq\mathfrak g$ and, moreover,
$\Phi_*\mathfrak m_j\subseteq\mathfrak m_j$, $j=1,\dots,7$.

We choose the following linearly independent elements of~$\mathfrak g$:
\begin{gather}\label{eq:dNs1.3RhoNe1DefSubalgBasis}
Q^1:=Z(1),\quad Q^2:=Z(z_1),\quad Q^3:=R(1),\quad Q^4:=P^2,\quad Q^5:=H,\quad Q^6:=D^2.
\end{gather}
Since
$Q^1,Q^2\in\mathfrak m_7$, $Q^3\in\mathfrak m_6$, $Q^4,Q^5\in\mathfrak m_3$ and $Q^6\in\mathfrak m_2$, then
\begin{gather}\label{eq:dNRedEq1.3MainPushforwards}
\begin{split}&
\Phi_*Q^i=a_{i1}\tilde Z(1)+a_{i2}\tilde Z(\tilde z_1),\quad i=1,2,
\\&
\Phi_*Q^i=a_{i1}\tilde Z(1)+a_{i2}\tilde Z(\tilde z_1)+a_{i3}\tilde R(1),\quad i=3,
\\&
\Phi_*Q^i=\tilde Z(\tilde\sigma^i)+\tilde R(\tilde\alpha^i)+a_{i4}\tilde P^2+a_{i5}\tilde H,\quad i=4,5,
\\&
\Phi_*Q^i=\tilde Z(\tilde\sigma^i)+\tilde R(\tilde\alpha^i)+a_{i4}\tilde P^2+a_{i5}\tilde H
+a_{i6}\tilde D^2,\quad i=6.
\end{split}
\end{gather}
Here
$a_{ij}$ are constants with $\hat\Delta a_{33}(a_{44}a_{55}-a_{45}a_{54})\ne0$,
the other parameters are smooth functions of~$\tilde z_1$,
and we denote $a_{11}a_{22}-a_{12}a_{21}=:\hat\Delta$.

For each~$i\in\{1,\dots,6\}$,
we expand the $i$th equation from~\eqref{eq:dNRedEq1.3MainPushforwards}, split it componentwise and pull the result back by~$\Phi$.
We simplify the obtained constraints, taking into account constraints derived in the same way for preceding values of~$i$
and omitting the constraints satisfied identically in view of other constraints.

Thus, for $i=1,2,3$, we get
$Z^1_w=Z^2_w=0$, $W_w=a_{11}+a_{12}Z^1$, $z_1W_w=a_{21}+a_{22}Z^1$ and $z_2W_w=a_{31}+a_{32}Z^1+a_{33}Z^2$.
Hence
\begin{gather*}
Z^1=-\frac{a_{11}z_1-a_{21}}{a_{12}z_1-a_{22}},
\quad
Z^2=\frac{-\hat\Delta}{a_{33}(a_{12}z_1-a_{22})}z_2
+\frac{a_{32}}{a_{33}}\,\frac{a_{11}z_1-a_{21}}{a_{12}z_1-a_{22}}-\frac{a_{31}}{a_{33}},
%\\Z^2=\frac{-(a_{11}a_{22}-a_{12}a_{21})z_2+(a_{11}a_{32}-a_{12}a_{31})z_1+a_{31}a_{22}-a_{32}a_{21}}{a_{33}(a_{12}z_1-a_{22})}
\\
W_w=\frac{-\hat\Delta}{a_{12}z_1-a_{22}}.
\end{gather*}

The equations~\eqref{eq:dNRedEq1.3MainPushforwards} with $i=4,5$ result in the constraints
\begin{gather}\label{eq:dNRedEq1.3MainPushforwards4,5a}
Z^2_2=a_{44}+a_{45}Z^1,\quad
z_1Z^2_2=a_{54}+a_{55}Z^1,\\[.8ex]\label{eq:dNRedEq1.3MainPushforwards4,5b}
W_2=\frac{a_{45}}2(Z^2)^2+\tilde\alpha^4(Z^1)Z^2+\tilde\sigma^4(Z^1),\\[-.5ex]\label{eq:dNRedEq1.3MainPushforwards4,5c}
z_1W_2+\frac{z_2^2}2W_w=\frac{a_{55}}2(Z^2)^2+\tilde\alpha^5(Z^1)Z^2+\tilde\sigma^5(Z^1).
\end{gather}
Solving~\eqref{eq:dNRedEq1.3MainPushforwards4,5a} as a system
of linear algebraic equations with respect to $(Z^1,Z^2_2)$
and comparing the obtained expressions with the above ones,
we derive that the tuple $(a_{44},a_{45},a_{54},a_{55})$ is proportional to  $(a_{11},a_{12},a_{21},a_{22})$
with the multiplier $1/a_{33}$.
In particular, $a_{45}=a_{12}/a_{33}$ and $a_{55}=a_{22}/a_{33}$.
Therefore, integrating the equation~\eqref{eq:dNRedEq1.3MainPushforwards4,5b} gives
\[
W=\frac{-\hat\Delta}{a_{12}z_1-a_{22}}w
+\frac{a_{12}\hat\Delta^2}{6 a_{33}^3(a_{12}z_1-a_{22})^2}z_2^3
+W^2(z_1)z_2^2+W^1(z_1)z_2+W^0(z_1),
\]
where the function~$W^0$ arises due to the integration with respect to~$z_2$,
and the functions~$W^1$ and~$W^2$ are expressed via $a_{i,j}$, $i=1,2,3$, $j=1,2$, $a_{33}$,
$\tilde\sigma^4(Z^1)$ and $\tilde\alpha^4(Z^1)$
but the precise form of these expressions is not essential.
The equation~\eqref{eq:dNRedEq1.3MainPushforwards4,5c} results in the constraint
$a_{33}^3=\hat\Delta$
and expressions for $\tilde\sigma^5(Z^1)$ and $\tilde\alpha^5(Z^1)$,
which both are inessential as well.

From the equation~\eqref{eq:dNRedEq1.3MainPushforwards} with $i=6$,
we obtain the following explicit form of the coefficient~$W^2$:
\[W^2=\frac{-a_{32}\hat\Delta}{2 (a_{12}z_1-a_{22})^2}.\]
Re-denoting the constant parameters,
$c_1=a_{11}\hat\Delta^{-2/3}$,
$c_2=-a_{21}\hat\Delta^{-2/3}$,
$c_3=-a_{12}\hat\Delta^{-2/3}$,
$c_4=a_{22}\hat\Delta^{-2/3}$,
$c_5=(a_{31}a_{12}-a_{32}a_{11})\hat\Delta^{-1}$ and
$c_6=(a_{21}a_{32}-a_{22}a_{31})\hat\Delta^{-1}$,
we derive the representation~\eqref{eq:dNs1.3RhoNe1ModifiedRedEqPointSymForm}
for the point transformation~$\Phi$.

We can straightforwardly check by the direct substitution that
any point transformation of the form~\eqref{eq:dNs1.3RhoNe1ModifiedRedEqPointSymForm}
is a point symmetry of the equation~\eqref{eq:dNs1.3RhoNe1ModifiedRedEq}.

It is easy to check by the direct substitution that
any point transformation of this form is a point symmetry
of the equation~\eqref{eq:dNs1.3RhoNe1ModifiedRedEq}.
This means that the condition $\Phi_*\mathfrak g\subseteq\mathfrak g$
is not only necessary but also sufficient
for a point transformation~$\Phi$ belongs to
the point-symmetry pseudogroup~$G_{1.3}$ of the equation~\eqref{eq:dNs1.3RhoNe1ModifiedRedEq}.
\end{proof}

\begin{remark}
The proof of Theorem~\ref{thm:dNRedEq1.3PointSymGroup} shows
that the following implications hold true:
\begin{gather*}
\Phi_*\mathfrak m_7\subseteq\mathfrak m_7\ \Rightarrow\ \Phi_*\mathfrak m_5\subseteq\mathfrak m_5,
\\
\Phi_*\mathfrak m_7\subseteq\mathfrak m_7,\,\Phi_*\mathfrak m_6\subseteq\mathfrak m_6
\ \Rightarrow\ \Phi_*\mathfrak m_4\subseteq\mathfrak m_4,
\\
\Phi_*\mathfrak m_7\subseteq\mathfrak m_7,\,\Phi_*\mathfrak m_6\subseteq\mathfrak m_6,\,\Phi_*\mathfrak m_3\subseteq\mathfrak m_3,\,\Phi_*\mathfrak m_2\subseteq\mathfrak m_2
\ \Rightarrow\ \Phi_*\mathfrak m_1\subseteq\mathfrak m_1.
\end{gather*}
Therefore, the collection of
the megaideals~$\mathfrak m_7$, $\mathfrak m_6$, $\mathfrak m_3$ and~$\mathfrak m_2$
is optimal in the course of applying the megaideal-based method
for computing the pseudogroup~$G_{1.3}$.
Nevertheless, the claim that the conditions $\Phi_*\mathfrak m_i\subseteq\mathfrak m_i$, $i=1,4,5$,
give no new constraints for~the components of~$\Phi$
in comparison with the conditions $\Phi_*\mathfrak m_i\subseteq\mathfrak m_i$, $i=2,3,6,7$,
becomes evident only in the course of proving Theorem~\ref{thm:dNRedEq1.3PointSymGroup}.
\end{remark}

\begin{remark}
Moreover, applying the modified version of the megaideal-based method from~\cite{malt2024a}
to the equation~\eqref{eq:dNs1.3RhoNe1ModifiedRedEq},
we can replace the collection of the conditions $\Phi_*Q\subseteq\mathfrak m_i$
for $Q$ from a set of vector fields generating the megaideal~$\mathfrak m_i$, $i=2,3,6,7$,
by a selection of a finite number (which is six) of these conditions,
$\Phi_*Q^1,\Phi_*Q^2\subseteq\mathfrak m_7$, $\Phi_*Q^3\subseteq\mathfrak m_6$,
$\Phi_*Q^4,\Phi_*Q^4\subseteq\mathfrak m_3$ and $\Phi_*Q^6\subseteq\mathfrak m_2$
under the notation~\eqref{eq:dNs1.3RhoNe1DefSubalgBasis}.
\end{remark}

\begin{remark}
The span~$\mathfrak h$ %=\langle Q^1,\dots,Q^6\rangle
of the selected linearly independent vector fields~$Q^1$,~\dots,~$Q^6$
is closed with respect to the Lie bracket of vector fields, i.e.,
it is a subalgebra of the algebra~$\mathfrak a_{1.3}$.
This phenomenon in the course of applying the above modification of the megaideal-based method
was also observed in~\cite[Remark~6]{malt2024a} and~\cite[Remark~24]{boyk2024a},
but it is still not clear whether its appearance is occasional
or one can always choose such appropriate vector fields
that they substitute a basis of a subalgebra of the corresponding maximal Lie invariance algebra.
\end{remark}

\begin{remark}\label{rem:dNRedEq1.3OnOtherMegaidealChain}
When proving Theorem~\ref{thm:dNRedEq1.3PointSymGroup},
we can replace the megaideal~$\mathfrak m_2$ by~$\mathfrak m_1$,
and then the selected subalgebra $\mathfrak h:=\langle Z(1),Z(z_1),R(1),P^2,H,D^2\rangle$
should be replaced by the subalgebra $\tilde{\mathfrak h}:=\langle Z(1),Z(z_1),R(1),P^2,H,P^1,2D^1+D^2\rangle$
of greater dimension.
At the same time, this replacement complicates the related computations.
\end{remark}

\begin{corollary}
The contact-symmetry pseudogroup~$G_{1.3\rm c}$ of reduced equation~1.3
coincides with the first prolongation~$G_{1.3(1)}$ of the pseudogroup~$G_{1.3}$.
\end{corollary}

\begin{proof}
Mimicking part~(ii) of the proof of \cite[Theorem~2]{boyk2024a},
we follow the proof of Theorem~\ref{thm:dNRedEq1.3PointSymGroup}
and use the same version of the algebraic method,
just extending it to contact vector fields and contact transformations.
In other words, we replace
the maximal Lie invariance algebra~$\mathfrak a_{1.3}$,
its megaideals~$\mathfrak m_j$, $j=1,\dots,7$,
and a point transformation~$\Phi$
by the contact invariance algebra~$\mathfrak a_{1.3\rm c}$
coinciding with the first prolongation $\mathfrak a_{1.3(1)}$ of~$\mathfrak a_{1.3}$,
the first prolongations~$\mathfrak m_{j(1)}$ of megaideals~$\mathfrak m_j$, $j=1,\dots,7$,
and a contact transformation
\[
\Psi\colon (\tilde z_1,\tilde z_2,\tilde w,\tilde w_{\tilde z_1},\tilde w_{\tilde z_2})
=(Z^1,Z^2,W,W^{z_1},W^{z_2}),
\]
respectively.
The right-hand side of the above equality is given by a tuple of smooth functions
of $(z_1,z_2,w,w_1,w_2)$ with nonvanishing Jacobian,
which additionally satisfies the contact condition
\begin{gather*}
(Z^k_l+Z^k_ww_l)W^{z_k}=W_l+W_ww_l,\quad
Z^k_{w_l}W^{z_k}=W_{w_l}.
\end{gather*}
Here and in what follows, the indices~$k$ and~$l$ run through the set $\{1,2\}$,
we assume summation for repeated indices,
and supplementing subscripts with ``(1)'' denotes the first prolongation of the corresponding object.
If the transformation~$\Psi$ is a contact symmetry of the equation~\eqref{eq:dNs1.3RhoNe1ModifiedRedEq},
then $\Psi_*\mathfrak m_{j(1)}\subseteq\mathfrak m_{j(1)}$, $j=1,\dots,7$,
where $\Psi_*$ is the pushforward of contact vector fields by~$\Psi$.
The counterpart of the collection of equations~\eqref{eq:dNRedEq1.3MainPushforwards} for the contact case
is handled in the same way as described after~\eqref{eq:dNRedEq1.3MainPushforwards}.
Successively considering the equations with $i=1$, $i=2$ and $i=3$,
we in particular derive the constraints $Z^k_w=0$, $Z^k_{w_1}=0$ and $Z^k_{w_2}=0$, respectively.
In view of the contact condition, this implies that $W_{w_l}=0$ as well.
Therefore, the contact transformation~$\Psi$ is the first prolongation
of a point transformation in the space $\mathbb R^3_{z_1,z_2,w}$,
which means that $G_{1.3\rm c}=G_{1.3(1)}$.
\end{proof}

\begin{remark}\label{rem:dNs1.3RhoNe1ModifiedDefSymGroup}
According to the proof of Theorem~\ref{thm:dNRedEq1.3PointSymGroup},
the necessary condition $\Phi_*\mathfrak a_{1.3}\subseteq\mathfrak a_{1.3}$
for elements~$\Phi$ of the pseudogroup~$G_{1.3}$ in fact defines this pseudogroup completely.
In other words, the pseudogroup~$G_{1.3}$ coincides with the stabilizer of the algebra~$\mathfrak a_{1.3}$
in the pseudogroup of local diffeomorphisms in the space $\mathbb R^3_{z_1,z_2,w}$.
Thus, the application of the direct method in the course of the second part
of the computing~$G_{1.3}$ within the framework of the algebraic method
reduces to the trivial check that all the point transformations
singled out by the condition $\Phi_*\mathfrak a_{1.3}\subseteq\mathfrak a_{1.3}$
are indeed point symmetries of the equation~\eqref{eq:dNs1.3RhoNe1ModifiedRedEq}.
In the literature, there is only one example of a system of differential equations
with the above property, given by the dispersionless Nyzhnyk equation~\eqref{eq:dN}, see~\cite{boyk2024a}.
The analogous claims also hold for the algebra~$\mathfrak a_{1.3\rm c}$ and the pseudogroup~$G_{1.3\rm c}$.
\end{remark}

By analogy with the algebra~$\mathfrak a_{1.3}$,
considering the modified composition of transformations \cite{kova2023b,kova2023a} as the pseudogroup operation,
we can represent the pseudogroup~$G_{1.3}$ as the product of
its subgroup~$G_{1.3}^{\rm ess}$
and its normal pseudosubgroup $G_{1.3}^{\rm triv}$.
Here
the subgroup~$G_{1.3}^{\rm ess}$ consists of the transformations
of the form~\eqref{eq:dNs1.3RhoNe1ModifiedRedEqPointSymForm} with $W^1_1=W^0_{11}=0$
and their natural domains.
The normal pseudosubgroup $G_{1.3}^{\rm triv}$ is singled out from~$G_{1.3}$
by the constraints $c_1=c_4=1$ and $c_2=c_3=c_5=c_6=0$,
i.e., it is constituted by the transformations of the form
$\tilde z_1=z_1$, $\tilde z_2=z_2$, $\tilde w=w+W^1(z_1)z_2+W^0(z_1)$.
The intersection $G_{1.3}^{\rm ess}\cap G_{1.3}^{\rm triv}$ is a normal subgroup of~$G_{1.3}^{\rm ess}$
and consists of the globally defined transformations
$\tilde z_1=z_1$, $\tilde z_2=z_2$, $\tilde w=w+a_2z_2+a_1z_1+a_0$,
where $a_0$, $a_1$ and~$a_2$ are arbitrary constants.

The vector fields~\eqref{eq:dNMIASubmodel} spanning the algebra~$\mathfrak a_{1.3}$
are respectively associated with the following parameter families of transformations
from the pseudogroup~$G_{1.3}$:
\begin{gather*}%\label{eq:dNs1.3ElementaryTransformations}
\arraycolsep=0ex
\begin{array}{llllll}
\mathscr P^1(c_2)\colon\quad & \tilde z_1=z_1+c_2,              \quad& \tilde z_2=z_2,                  \quad& \tilde w=w,                \\[1.3ex]
\mathscr D^1(c_1)\colon\quad & \tilde z_1=c_1z_1,               \quad& \tilde z_2=z_2,                  \quad& \tilde w=c_1^{-1}w,        \\
\mathscr K(c_3)  \colon\quad & \tilde z_1=\dfrac{z_1}{1-c_3z_1},\quad& \tilde z_2=\dfrac{z_2}{1-c_3z_1},\quad& \tilde w=\dfrac w{1-c_3z_1}+\dfrac{c_3z_2^3}{6(1-c_3z_1)^2},\\[1.8ex]
\mathscr D^2(\tilde c_4)\colon\quad & \tilde z_1=z_1,           \quad& \tilde z_2=\tilde c_4z_2,        \quad& \tilde w=\tilde c_4^{\,3}w,\\[1.3ex]
\mathscr P^2(c_6)\colon\quad & \tilde z_1=z_1,                  \quad& \tilde z_2=z_2+c_6,              \quad& \tilde w=w,                \\[1.3ex]
\mathscr H(c_5)  \colon\quad & \tilde z_1=z_1,                  \quad& \tilde z_2=z_2+c_5z_1,           \quad& \tilde w=w+\frac12c_5z_2^2+\frac12c_5^{\,2}z_1z_2+\frac16c_5^{\,3}z_1^2,\\[1.3ex]
\mathscr R(W^1)  \colon\quad & \tilde z_1=z_1,                  \quad& \tilde z_2=z_2,                  \quad& \tilde w=w+W^1(z_1)z_2,    \\[1.3ex]
\mathscr Z(W^0)  \colon\quad & \tilde z_1=z_1,                  \quad& \tilde z_2=z_2,                  \quad& \tilde w=w+W^0(z_1),
\end{array}
\end{gather*}
where $c_1$, $c_2$, $c_3$, $\tilde c_4$, $c_5$ and $c_6$ are arbitrary constants with $c_1\ne0$
and $W^0$ and $W^1$ are arbitrary smooth functions of~$z_1$.
Each of these families is singled out from~$G_{1.3}$ by setting all
the constant and functional parameters in~\eqref{eq:dNs1.3RhoNe1ModifiedRedEqPointSymForm}, except the single associated one,
to the trivial values corresponding to the identity transformation,
which are 1, 0, 0, 1, 0, 0, 0, 0, for $c_1$, \dots, $c_6$, $W^0$ and~$W^1$, respectively.
The exception is the family $\{\mathscr D^2(\tilde c_4)\}$,
for which we confine, to the trivial values, all the parameters except $c_1$ and $c_4$,
set $c_1=c_4$ and re-denote $\tilde c_4=1/c_4$.
Thus, each of these families is a (pseudo)subgroup of~$G_{1.3}$ parameterized by a single constant or functional parameters,
and each element of~$G_{1.3}$ can be represented as a composition of transformations from these (pseudo)subgroups.
It is natural to consider these transformations as elementary point symmetry transformations of the equation~\eqref{eq:dNs1.3RhoNe1ModifiedRedEq}.
The (pseudo)subgroups
$\{\mathscr P^1(c_2)\}$,
$\{\mathscr D^1(c_1)\mid c_1>0\}$,
$\{\mathscr K(c_3)\}$,
$\{\mathscr D^2(\tilde c_4)\mid \tilde c_4>0\}$,
$\{\mathscr P^2(c_6)\}$,
$\{\mathscr H(c_5)\}$,
$\{\mathscr R(W^1)\}$,
$\{\mathscr Z(W^0)\}$
are generated by the vector fields~$P^1$, $D^1$, $K$, $D^2$, $P^2$ and $H$
and the collection of vector fields $\{R(\alpha)\}$ and $\{Z(\sigma)\}$, respectively.
The families $\{\mathscr D^1(c_1)\}$ and $\{\mathscr D^2(\tilde c_4)\}$
also contain the compositions of elements of their subfamilies
$\{\mathscr D^1(c_1)\mid c_1>0\}$ and $\{\mathscr D^2(\tilde c_4)\mid \tilde c_4>0\}$
with the discrete point symmetry transformation $(\tilde z_1,\tilde z_2,\tilde w)=(-z_1,z_2,-w)$
and the transformation $(\tilde z_1,\tilde z_2,\tilde w)=(z_1,-z_2,-w)$ belonging to the identity component of~$G_{1.3}$, respectively.

The Lie algebra spanned by the vector fields $P^1$, $2D^1+D^2$ and $K$
is isomorphic to the algebra ${\rm sl}(2,\mathbb R)$.
This is why it is convenient to replace the basis element~$K$ by $Q^+:=P^1+K$
since the modified basis agrees with the Iwasawa decomposition of the corresponding Lie group.
The one-parameter group generated by $Q^+$ consists of the transformations
\begin{gather}\label{eq:dNPointTransformP1K}
\begin{split}
\mathscr Q^+(\tilde c_3)\colon\ &\tilde z_1=\frac{\sin\tilde c_3+z_1\cos\tilde c_3}{\cos\tilde c_3-z_1\sin\tilde c_3},\quad
\tilde z_2=\frac{z_2}{\cos\tilde c_3-z_1\sin\tilde c_3},\\&
\tilde w=\frac{w}{\cos\tilde c_3-z_1\sin\tilde c_3}
+\frac{\sin\tilde c_3}{(\cos\tilde c_3-z_1\sin\tilde c_3)^2}\frac{z_2^3}6,
\end{split}
\end{gather}
where $\tilde c_3$ is an arbitrary constant parameter,
which is determined by the corresponding transformation up to the summands $2\pi k$, $k\in\mathbb Z$.
The transformation~\eqref{eq:dNPointTransformP1K} with $\tilde c_3=-\pi/2$ is
\[
\mathscr K'\colon\quad \tilde z_1=-\dfrac1{z_1},\quad \tilde z_2=\dfrac{z_2}{z_1},\quad \tilde w=\dfrac w{z_1}-\dfrac{z_2^3}{6z_1^2},
\]
which it can be represented as the composition~$\mathscr P^1(-1)\circ\mathscr K(-1)\circ\mathscr P^1(-1)$.
The value $\tilde c_3=\pi$ corresponds to the transformation
\[
\mathscr J\colon\quad \tilde z_1=z_1,\quad \tilde z_2=-z_2,\quad \tilde w=-w.
\]
This shows that under the chosen interpretation of linear fractional transformations
as that in \cite{kova2023b,kova2023a},
the transformations~$\mathscr K'$ and~$\mathscr J$ belong to the identity components of~$G_{1.3}^{\rm ess}$ and of~$G_{1.3}$,
and thus they are not discrete point symmetry transformations of the equation~\eqref{eq:dNs1.3RhoNe1ModifiedRedEq}
although they look to be those.
Therefore, a complete list of discrete point symmetry transformations of the equation~\eqref{eq:dNs1.3RhoNe1ModifiedRedEq}
that are independent up to composing with each other
and with continuous point symmetry transformations of~\eqref{eq:dNs1.3RhoNe1ModifiedRedEq}
is exhausted by the single involution alternating the signs of $(z_1,w)$,
$\mathscr I\colon(\tilde z_1,\tilde z_2,\tilde w)=(-z_1,z_2,-w)$.

Similar to the algebras~$\mathfrak a_{1.3}$ and~$\mathfrak a_{\rm iB}$,
the substitution~$w_{22}=h$ induces a homomorphism~$\Upsilon$
of the pseudogroup~$G_{1.3}$ into the point symmetry pseudogroup~$G_{\rm iB}$
of the inviscid Burgers equation~\eqref{eq:dNs1.3RhoNe1InviscidBurgersEq}.
The homomorphism~$\Upsilon$ can be represented as the composition of the second prolongation of the transformations from~$G_{1.3}$
with the pushforward of the prolonged transformations by the natural projection
from ${\rm J}^2(\mathbb R^2_{z_1,z_2}\times\mathbb R_w)$ onto $\mathbb R^2_{z_1,z_2}\times\mathbb R_{w_{22}}$.
Thus, the transformation components for $z_1$ and $z_2$ are preserved,
and the transformation component for $h$ coincides with that for $w_{22}$.
In other words, the pseudogroup~$\mathop{\rm im}\Upsilon$
consists of the transformations
\begin{gather}\label{eq:dNSubmodelTransportInducedPointSymForm}
\tilde z_1=\frac{c_1z_1+c_2}{c_3z_1+c_4},\quad
\tilde z_2=\frac{z_2+c_5z_1+c_6}{c_3z_1+c_4},\quad
\tilde h=\frac{(c_3z_1+c_4)h-c_3z_2-c_3c_6+c_4c_5}\Delta,
\end{gather}
where $c_1$, \dots, $c_6$ are arbitrary constants with $\Delta:=c_1c_4-c_2c_3\ne0$.
Properly defining the domains of transformations of the form~\eqref{eq:dNSubmodelTransportInducedPointSymForm}
and their composition in the manner of \cite{kova2023b,kova2023a},
we can convert the pseudogroup~$\mathop{\rm im}\Upsilon$ into a group~$\check G_{1.3}$,
which is naturally isomorphic to the group constituted by the matrices of the form
\[
\begin{pmatrix}c_1&c_2&0\\c_3&c_4&0\\c_5&c_6&1\end{pmatrix}\quad\mbox{with}\quad
\Delta:=\begin{vmatrix}c_1&c_2\\c_3&c_4\end{vmatrix}\ne0
\]
and thus antiisomorphic to the general affine group~${\rm Aff}(2,\mathbb R)$.
Summing up, the kernel of the homomorphism~$\Upsilon$
coincides with the pseudosubgroup $G_{1.3}^{\rm triv}$,
whereas the group~$\check G_{1.3}$ associated with its image
is isomorphic to the quotient group of~$G_{1.3}^{\rm ess}$ by $G_{1.3}^{\rm ess}\cap G_{1.3}^{\rm triv}$.

By
$\check{\mathscr P}^1(c_2)$,
$\check{\mathscr D}^1(c_1)$,
$\check{\mathscr K}(c_3)  $,
$\check{\mathscr D}^2(\tilde c_4)$,
$\check{\mathscr P}^2(c_6)$,
$\check{\mathscr H}(c_5)$ and
$\check{\mathscr Q}^+(\tilde c_3)$
we denote the images of
$\mathscr P^1(c_2)$,
$\mathscr D^1(c_1)$,
$\mathscr K(c_3)  $,
$\mathscr D^2(\tilde c_4)$,
$\mathscr P^2(c_6)$,
$\mathscr H(c_5)$ and
$\mathscr Q^+(\tilde c_3)$
with respect to the homomorphism~$\Upsilon$, respectively.

\section{Defining properties of Lie symmetries}\label{sec:dNs1.3RhoNe1DefPropertiesOfLieSyms}

In view of Remark~\ref{rem:dNs1.3RhoNe1ModifiedDefSymGroup},
it is of interest to look for other defining properties of Lie symmetries
of the equation~\eqref{eq:dNs1.3RhoNe1ModifiedRedEq}.
The most interesting among them is this equation is completely defined
not only by its (infinite-dimensional) maximal Lie invariance algebra~$\mathfrak a_{1.3}$
but also by a proper (finite-dimensional) subalgebra of~$\mathfrak a_{1.3}$.

\begin{theorem}\label{thm:dNs1.3RhoNe1ModifiedDefSubalg}
(i) A genuine%
\footnote{
Here the attribute ``genuine'' of the corresponding partial differential equation means
that it cannot be represented in or transformed to a form
where one of the independent variables plays a role of a parameter.
}
partial differential equation of maximal rank of order not greater than three
in two independent variables $z_1$ and~$z_2$ and dependent variable~$w$ admits the algebra
\[
\mathfrak p:=\big\langle
P^1,\,P^2,\,Z(1),\,Z(z_1),\,Z(z_1^2),\,Z(z_1^3),\,R(1),\,R(z_1),\,R(z_1^2),\,H,\,D^2\big\rangle
\]
as its Lie invariance algebra if and only if it coincides with the equation~\eqref{eq:dNs1.3RhoNe1ModifiedRedEq}.

(ii) A differential equation of maximal rank of order not greater than three
in two independent variables $z_1$ and~$z_2$ and dependent variable~$w$ admits the algebra
\[
\mathfrak q:=\big\langle
P^1,\,P^2,\,Z(1),\,Z(z_1),\,Z(z_1^2),\,Z(z_1^3),\,R(1),\,R(z_1),\,R(z_1^2),\,H,\,2D^1+D^2,\,K
\big\rangle
\]
as its Lie invariance algebra if and only if it coincides with the equation~\eqref{eq:dNs1.3RhoNe1ModifiedRedEq}.

\end{theorem}

\begin{proof}
The ``if''-parts of both statements~(i) and~(ii) are trivial.
Let us prove their ``only if''-parts.
Since the algebras~$\mathfrak p$ and~$\mathfrak q$ have a quite large intersection,
\[
\mathfrak p\cap\mathfrak q=\big\langle
P^1,\,P^2,\,Z(1),\,Z(z_1),\,Z(z_1^2),\,Z(z_1^3),\,R(1),\,R(z_1),\,R(z_1^2),\,H
\big\rangle,
\]
the major part of the proof is the same for~(i) and~(ii).

Consider a differential equation $F=0$, where $F=F[w]$ is a differential function~\cite[p.~288]{olve1993A} of~$w$
with $\ord F\leqslant3$,
and denote by~$\mathcal M$ the manifold defined by it in the third-order jet space.
Suppose that this equation admits $\mathfrak p\cap\mathfrak q$ as its Lie invariance algebra.
Successively taking into account the invariance
with respect to~$P^1$, $P^2$, $Z(1)$, $Z(z_1)$, $Z(z_1^2)$, $Z(z_1^3)$, $R(1)$, $R(z_1)$ and~$R(z_1^2)$,
we obtain that up to factoring out an inessential nonvanishing multiplier,
the differential function~$F$ can be assumed
not to depend on~$z_1$, $z_2$, $w_{0,0}$, $w_{1,0}$, $w_{2,0}$, $w_{3,0}$, $w_{0,1}$, $w_{1,1}$ and $w_{2,1}$.
In other words, it depends at most on~$w_{0,2}$, $w_{1,2}$ and~$w_{0,3}$.
Then the invariance of the equation with respect to~$H$ implies
that up to factoring out a nonvanishing multiplier,
the differential function~$F$ can be assumed to depend at most on~$w_{0,3}$ and~$\omega:=w_{1,2}+w_{0,2}w_{0,3}$.

(i) Let the equation $F=0$ admit the entire algebra~$\mathfrak p$.
The dependence on the latter expression is essential for the equation to be a genuine partial differential one.
Since the equation $F=0$ is of maximal rank, we have $F_{w_{0,3}}\ne0$ or $F_\omega\ne0$ at each point of~$\mathcal M$.
Suppose that $F_{w_{0,3}}\ne0$ at some such point.
Then we can locally solve the equation $F=0$ with respect to~$w_{0,3}$,
$w_{0,3}=f(\omega)$ for some sufficiently smooth function~$f$ of~$\omega$.
The invariance with respect to~$D^2$ implies that the function~$f$ is constant,
which contradicts the supposition that the equation $F=0$ is a genuine partial differential equation.
Hence $F_{w_{0,3}}=0$ and $F_\omega\ne0$ on the entire manifold~$\mathcal M$.
For each point of~$\mathcal M$, we locally solve the equation $F=0$ with respect to~$\omega$,
obtaining the equation $\omega=c$ for some constant~$c$.
The last equation is invariant with respect to~$D^2$ only if $c=0$.

(ii) Let the equation $F=0$ admit the entire algebra~$\mathfrak q$.
Its invariance with respect to~$2D^1+D^2$ and~$K$ implies that
$2w_{0,3}F_{w_{0,3}}+3\omega F_\omega=0$ and
$(2z_1w_{0,3}-1)F_{w_{0,3}}+3z_1\omega F_\omega=0$ on~$\mathcal M$.
Hence $F_{w_{0,3}}=0$ on~$\mathcal M$, and thus $F_\omega\ne0$ on~$\mathcal M$ since the equation $F=0$ is maximal rank.
This means that $\omega=0$ on~$\mathcal M$,
i.e., the equation $F=0$ is equivalent to the equation $\omega=0$.
\end{proof}

A statement similar to Theorem~\ref{thm:dNs1.3RhoNe1ModifiedDefSubalg}
also holds for the equation~\eqref{eq:dNs1.3RhoNe1InviscidBurgersEq}.
More specifically, a genuine partial differential equation (resp.\ a differential equation)
of maximal rank of order one
in two independent variables $z_1$ and~$z_2$ and dependent variable~$h$ admits the algebra
$\check{\mathfrak p}:=\langle\check P^1,\,\check P^2,\,\check H,\,\check D^2\rangle$
(resp.\ $\check{\mathfrak q}:=\langle\check P^1,\,\check P^2,\,\check H,\,2\check D^1+\check D^2,\,\check K\rangle$)
as its Lie invariance algebra
if and only if it coincides with the equation~\eqref{eq:dNs1.3RhoNe1InviscidBurgersEq}~\cite{popo2025b}.
Using the terminology of \cite{gorg2019a,mann2014a,mann2007b,oliv2004a},
we can say that the equations~\eqref{eq:dNs1.3RhoNe1InviscidBurgersEq} and~\eqref{eq:dNs1.3RhoNe1ModifiedRedEq}
are (strongly) \emph{Lie-remarkable} in the context of partial differential equations of maximal rank.
See also \cite{andr2001a,krau1994a,nucc1996b,rose1986a} %,rose1988a
for studies on defining differential equations by their Lie or more general symmetries.

Theorem~\ref{thm:dNs1.3RhoNe1ModifiedDefSubalg} means that
the equation~\eqref{eq:dNs1.3RhoNe1ModifiedRedEq} is defined by its Lie symmetries
in a much more restrictive way than the dispersionless Nyzhnyk equation~\eqref{eq:dN} does.
More specifically, finite-dimensional subalgebras of the maximal Lie invariance algebra~$\mathfrak a_{1.3}$
of the equation~\eqref{eq:dNs1.3RhoNe1ModifiedRedEq} define not only the point-symmetry pseudogroup~$G_{1.3}$
of this equation, but also the equation~\eqref{eq:dNs1.3RhoNe1ModifiedRedEq} itself.
In contrast to this,
to completely define the dispersionless Nyzhnyk equation~\eqref{eq:dN} by its geometric properties,
one should involve, in addition to its Lie symmetries, e.g., its three simplest conservation laws
\cite[Theorem~19]{boyk2024a}.
Another point is that for defining certain properties of the equation~\eqref{eq:dNs1.3RhoNe1ModifiedRedEq},
even a narrower subalgebra of the algebra~$\mathfrak a_{1.3}$
than the subalgebra~$\mathfrak p$ from Theorem~\ref{thm:dNs1.3RhoNe1ModifiedDefSubalg} is sufficient.

Recall \cite[Definition~5]{boyk2024a}, see also \cite[Section~9]{malt2024a}.
A proper subalgebra~$\mathfrak s$ of a Lie algebra~$\mathfrak a$ of vector fields
is called a \emph{subalgebra defining the diffeomorphisms that stabilize~$\mathfrak a$}
if the conditions $\Phi_*\mathfrak a\subseteq\mathfrak a$ and $\Phi_*\mathfrak s\subseteq\mathfrak a$
for local diffeomorphisms (i.e., point transformations)~$\Phi$ in the underlying space are equivalent.
We have shown that
the point symmetry pseudogroup~$G_{1.3}$ of the equation~\eqref{eq:dNs1.3RhoNe1ModifiedRedEq}
coincides with the stabilizer of the algebra~$\mathfrak a_{1.3}$
in the pseudogroup of local diffeomorphisms in the space $\mathbb R^3_{z_1,z_2,w}$,
see Remark~\ref{rem:dNs1.3RhoNe1ModifiedDefSymGroup}.
Since admitting the subalgebra~$\mathfrak p$ as its Lie invariance algebra
completely defines the equation~\eqref{eq:dNs1.3RhoNe1ModifiedRedEq},
this subalgebra also defines the diffeomorphisms stabilizing the algebra~$\mathfrak a_{1.3}$.
At the same time, it turns out that these diffeomorphisms are also defined
by a proper subalgebra of~$\mathfrak p$ with essentially smaller dimension.

\begin{theorem}\label{thm:dNs1.3RhoNe1ModifiedDefSubalgOfMLIA}
The subalgebra~$\mathfrak h=\langle Z(1),Z(z_1),R(1),P^2,H,D^2\rangle$
of the algebra~$\mathfrak a_{1.3}$ defines the local diffeomorphisms that stabilize~$\mathfrak a_{1.3}$.
\end{theorem}

\begin{proof}
A local diffeomorphism stabilizes the algebra~$\mathfrak a_{1.3}$
if and only if it belongs to the pseudogroup~$G_{1.3}$,
i.e., it is the form~\eqref{eq:dNs1.3RhoNe1ModifiedRedEqPointSymForm}.
Therefore, it suffices to show that any local diffeomorphism~$\Phi$ stabilizing the subalgebra~$\mathfrak h$
is the form~\eqref{eq:dNs1.3RhoNe1ModifiedRedEqPointSymForm}.

We follow the proof of Theorem~\ref{thm:dNRedEq1.3PointSymGroup}
and use the same notation, including~\eqref{eq:dNs1.3RhoNe1DefSubalgBasis}.
However, for each of the selected elements~$Q^i$, $i=1,\dots,6$, of the algebra~$\mathfrak a_{1.3}$,
we employ the condition $\Phi_*Q^i\in\mathfrak a_{1.3}$ instead of the condition $\Phi_*Q^i\in\mathfrak m$,
where $\mathfrak m$ is the minimal megaideal of~$\mathfrak a_{1.3}$ containing the vector field~$Q^i$.
In other words, we replace the equations~\eqref{eq:dNRedEq1.3MainPushforwards} with the equations
\begin{gather}\label{eq:dNMainPushforwardsMod}
\Phi_*Q^i=a_{i1}\tilde P^1+a_{i2}\tilde D^1+a_{i3}\tilde K+a_{i4}\tilde D^2+a_{i5}\tilde P^2+a_{i6}\tilde H
+\tilde R(\tilde\alpha^i)+\tilde Z(\tilde\sigma^i),
\end{gather}
where
$a_{ij}$, $j=1,\dots,6$, are constants
and $\tilde\alpha^i$ and~$\tilde\sigma^i$ are smooth functions of~$\tilde z_1$.
In what follows, by $(i,\tilde z_1)$, $(i,\tilde z_2)$ and $(i,\tilde w)$
we denote the equations that are obtained by collecting the $\tilde z_1$-, $\tilde z_2 $- or $\tilde w$-components
in the equation~\eqref{eq:dNMainPushforwardsMod} with the same value of~$i$
and pulling the results back by~$\Phi$, respectively.

First, we consider the equations
\begin{gather*}
(1,\tilde z_1)\colon\quad Z^1_w=a_{13}(Z^1)^2+a_{12}Z^1+a_{11},\\
(2,\tilde z_1)\colon\quad z_1Z^1_w=a_{23}(Z^1)^2+a_{22}Z^1+a_{21},\\
(3,\tilde z_1)\colon\quad z_2Z^1_w=a_{33}(Z^1)^2+a_{32}Z^1+a_{31}.
\end{gather*}
If $Z^1_w\ne0$, then we can split the combination $z_1(1,\tilde z_1)-(2,\tilde z_1)$
with respect to $z_1$ and $Z^1$ and derive the equalities $a_{ij}=0$, $i=1,2$, $j=1,2,3$,
which contradicts the supposition $Z^1_w\ne0$ in view of~$(1,\tilde z_1)$.
Therefore, $Z^1_w=0$, and thus the equation~$(3,\tilde z_1)$ implies $a_{3j}=0$, $j=1,2,3$.

Under the derived conditions $a_{i3}=0$, $i=1,2,3$,
the equations $(i,\tilde z_2)$, $i=1,2,3$, take the following form:
\begin{gather*}
(1,\tilde z_2)\colon\quad Z^2_w=a_{14}Z^2+a_{15}+a_{16}Z^1,\\
(2,\tilde z_2)\colon\quad z_1Z^2_w=a_{24}Z^2+a_{25}+a_{26}Z^1,\\
(3,\tilde z_2)\colon\quad z_2Z^2_w=a_{34}Z^2+a_{35}+a_{36}Z^1.
\end{gather*}
Suppose that $Z^2_w\ne0$.
Then the splitting the combination $z_1(1,\tilde z_2)-(2,\tilde z_2)$ with respect to $Z^2$
leads to the equation $a_{14}z_1-a_{24}=0$, which further splits into $a_{14}=a_{24}=0$,
and to the equation $(a_{16}z_1-a_{26})Z^1+a_{15}z_1-a_{25}=0$.
The parameters $a_{15}$, $a_{16}$, $a_{25}$ and~$a_{26}$ do not simultaneously vanish
since otherwise the equation~$(1,\tilde z_2)$ immediately implies the contradiction
with the supposition $Z^2_w\ne0$.
Hence the tuples $(a_{15},a_{25})$ and $(a_{16},a_{26})$ are simultaneously nonzero
and thus $Z^1=-(a_{15}z_1-a_{25})/(a_{16}z_1-a_{26})$, i.e., $Z^1=Z^1(z_1)$ with $Z^1_1\ne0$
in view of the nondegeneracy of~$\Phi$.
Then we can split the combination $z_2(1,\tilde z_2)-(3,\tilde z_2)$ with respect to $(Z^1,z_2,Z^2)$
and get $a_{15}=a_{16}=a_{34}=a_{35}=a_{36}=0$, which contradicts the supposition as well.
As a result, $Z^2_w=0$.

In a similar way, we analyze the following collection of equations:
\begin{gather*}
(4,\tilde z_1)\colon\quad Z^1_2=a_{43}(Z^1)^2+a_{42}Z^1+a_{41},\\
(5,\tilde z_1)\colon\quad z_1Z^1_2=a_{53}(Z^1)^2+a_{52}Z^1+a_{51},\\
(6,\tilde z_1)\colon\quad z_2Z^1_2=a_{63}(Z^1)^2+a_{62}Z^1+a_{61}.
\end{gather*}
If $Z^1_2\ne0$, then we can split the combination $z_1(4,\tilde z_1)-(5,\tilde z_1)$
with respect to~$z_1$ and~$Z^1$ and derive the equalities $a_{ij}=0$, $i=4,5$, $j=1,2,3$,
contradicting the supposition $Z^1_2\ne0$ in view of the equation $(4,\tilde z_1)$.
Hence $Z^1_2=0$, and thus the nondegeneracy of~$\Phi$ and the equation $(6,\tilde z_1)$
respectively imply $Z^1_1Z^2_2\ne0$ and $a_{61}=a_{62}=a_{63}=0$.

Taking into account the previous results, consider the equations
\begin{gather*}
(4,\tilde z_2)\colon\quad Z^2_2=a_{44}Z^2+a_{45}+a_{46}Z^1,\\
(5,\tilde z_2)\colon\quad z_1Z^2_2=a_{54}Z^2+a_{55}+a_{56}Z^1,\\
(6,\tilde z_2)\colon\quad z_2Z^2_2=a_{64}Z^2+a_{65}+a_{66}Z^1.
\end{gather*}
The combination $z_1(4,\tilde z_2)-(5,\tilde z_2)$ splits with respect to $Z^2$
into the pair of the equations $a_{44}z_1-a_{54}=0$ and $(a_{46}z_1-a_{56})Z^1+a_{45}z_1-a_{55}=0$,
where the former equation further splits into the equalities $a_{44}=a_{54}=0$.
Then the latter equation, the equation~$(4,\tilde z_2)$ and the inequality $Z^2_2\ne0$
jointly imply that $(a_{46},a_{56})\ne(0,0)$ and thus
\begin{gather*}
Z^1=-\frac{a_{45}z_1-a_{55}}{a_{46}z_1-a_{56}},\quad
Z^2_2=-\frac{\hat\Delta}{a_{46}z_1-a_{56}}
\end{gather*}
with $\hat\Delta:=a_{45}a_{56}-a_{46}a_{55}\ne0$ since $Z^1_1\ne0$.
Substituting the expressions for $Z^1$ and $Z^2_2$ into the equation $(6,\tilde z_2)$
and differentiating the result with respect to~$z_2$, we obtain $a_{64}=1$,
and thus this equation gives
\begin{gather*}
Z^2=\frac{-\hat\Delta z_2+(a_{45}a_{66}-a_{46}a_{65})z_1+a_{65}a_{56}-a_{66}a_{55}}{a_{46}z_1-a_{56}}.
\end{gather*}

Finally, we analyze the equations
\begin{gather*}
(1,\tilde w)\colon\quad W_w=\tilde\alpha^1Z^2+\tilde\sigma^1,\\
(4,\tilde w)\colon\quad W_2=\tfrac12a_{46}(Z^2)^2+\tilde\alpha^4Z^2+\tilde\sigma^4,\\
(5,\tilde w)\colon\quad z_1W_2+\tfrac12z_2^2W_w=\tfrac12a_{56}(Z^2)^2+\tilde\alpha^5Z^2+\tilde\sigma^5,\\
(6,\tilde w)\colon\quad z_2W_2+3wW_w-3W=\tfrac12a_{66}(Z^2)^2+\tilde\alpha^6Z^2+\tilde\sigma^6.
\end{gather*}
The differential consequence $\p_{z_2}(1,\tilde w)-\p_w(4,\tilde w)$ is $\tilde\alpha^1=0$.
Hence the equation~$(1,\tilde w)$ is in fact of the form $W_w=\tilde\sigma^1$,
i.e., $W_w$ depends at most on~$z_1$.
Then, we can collect the coefficients of~$z_2^2$ in the combination $z_1(4,\tilde w)-(5,\tilde w)$
and derive the equation
\begin{gather*}
W_w=-\frac{\hat\Delta^2}{a_{46}z_1-a_{56}}.
\end{gather*}
Collecting the coefficients of~$z_2$ in the differential consequence
$(z_2\p_2+3w\p_w-2)(4,\tilde w)-\p_2(6,\tilde w)$ gives an expression for~$\tilde\alpha^4$,
$\tilde\alpha^4=a_{46}(z_2Z^2_2-Z^2)-a_{66}Z^2_2$,
%\[\tilde\alpha^4=\frac{a_{66}\hat\Delta-a_{46}((a_{45}a_{66}-a_{46}a_{65})z_1
%+a_{65}a_{56}-a_{66}a_{55})}{a_{46}z_1-a_{56}}.\]
which we then substitute into the equation $(4,\tilde w)$
to obtain a more specific expression for~$W_2$.

The last step is to substitute the found expressions for~$Z^2$, $W_2$ and~$W_w$ into $(6,\tilde w)$
and solve the resulting equation with respect to~$W$, which gives
\begin{gather*}
W=-\frac{\hat\Delta^2w}{a_{46}z_1-a_{56}}
+\frac{a_{46}\hat\Delta^2}{(a_{46}z_1-a_{56})^2}\frac{z_2^3}{6}
-\frac{a_{66}\hat\Delta^2}{(a_{46}z_1-a_{56})^2}\frac{z_2^2}{2}
+F^1(z_1)z_2+F^0(z_1),
\end{gather*}
where $F^1$ and $F^0$ are functions of~$z_1$ that are expressed in terms of parameters
of the equations~\eqref{eq:dNMainPushforwardsMod}.

We obtain that the point transformation~$\Phi$ is of the form~\eqref{eq:dNs1.3RhoNe1ModifiedRedEqPointSymForm}.
\end{proof}

%If we denote $a_{45}=\hat\Delta c_1$, $a_{55}=-\hat\Delta c_2$,
%$a_{46}=-\hat\Delta c_3$, $a_{56}=\hat\Delta c_4$,
%$a_{65}=(\hat\Delta c_5+a_{45}a_{66})/(a_{46})$,
%$a_{66}=(\hat\Delta c_6+a_{56}a_{65})/(a_{55})$ and
%$\hat\Delta=1/\Delta$, we get the transformations of the form~\eqref{eq:dNs1.3RhoNe1ModifiedRedEqPointSymForm}.

\begin{remark}
The proof of Theorem~\ref{thm:dNs1.3RhoNe1ModifiedDefSubalgOfMLIA} presents one more,
the most primitive algebraic way
for computing the point symmetry pseudogroup~$G_{1.3}$ of the equation~\eqref{eq:dNs1.3RhoNe1ModifiedRedEq},
which does not use megaideals of the algebra~$\mathfrak a_{1.3}$.
Although this computation is also based on pushing forward
only the six-dimensional subalgebra~$\mathfrak h$ of~$\mathfrak a_{1.3}$,
it is much more involved than the computation in the proof of Theorem~\ref{thm:dNRedEq1.3PointSymGroup},
and the simplification of the latter occurs precisely due to the use of known megaideals of~$\mathfrak a_{1.3}$.
\end{remark}

\begin{remark}
In Theorem~\ref{thm:dNs1.3RhoNe1ModifiedDefSubalgOfMLIA},
the subalgebra~$\mathfrak h$, which is contained in the megaideal~$\mathfrak m_2$ of~$\mathfrak a_{1.3}$,
can be replaced by the subalgebra \[\tilde{\mathfrak h}:=\langle Z(1),Z(z_1),R(1),P^2,H,P^1,2D^1+D^2\rangle,\]
which is contained in the megaideal~$\mathfrak m_1$ of~$\mathfrak a_{1.3}$,
but this leads to more complicated computations, cf.\ Remark~\ref{rem:dNRedEq1.3OnOtherMegaidealChain}.
\end{remark}

\section{On induction of Lie and point symmetries}\label{sec:dNRedEq1.3OnInduction}

The induction of Lie symmetries of a reduced system
by Lie symmetries of the original system of partial differential equations
is a well-known phenomenon and was first discussed already in \cite[Section~20.4]{ovsi1982A}.
For the dispersionless Nyzhnyk equation~\eqref{eq:dN} and its reduced equation~\eqref{eq:dNs1.3RhoNe1ModifiedRedEq},
this phenomenon reveals new features,
which have not been observed in the literature and deserve a~detailed consideration.

To find, for each fixed admissible value of the parameter function~$\rho$,
the algebra~$\mathfrak a_{1.3}^{\rho,\rm ind}$ of Lie-symmetry vector fields of the reduced equation~\eqref{eq:dNs1.3RhoNe1ModifiedRedEq}
that are induced under the Lie reduction of~\eqref{eq:dN} with respect to the subalgebra~$\mathfrak s_{1.3}^\rho$,
we make the following steps.
We first compute the normalizer ${\rm N}_{\mathfrak g}(\mathfrak s_{1.3}^\rho)$
of the subalgebra~$\mathfrak s_{1.3}^\rho$ in~$\mathfrak g$.
Then we push forward its elements by the point transformation
from the space with the coordinates $(t,x,y,u)$ to the space with the coordinates $(z_1,z_2,z_3,w)$
whose $z_1$-, $z_2$- and~$w$-components are defined in~\eqref{eq:dNs1.3RhoNe1ModifiedAnsatz}
and the $z_3$-component is, e.g., $z_3=y$.
And finally, we naturally project the pushed forward vector fields to the space with the coordinates $(z_1,z_2,w)$.
The normalizer ${\rm N}_{\mathfrak g}(\mathfrak s_{1.3}^\rho)$ depends on whether the derivative~$\rho_t$ vanishes,
\begin{gather*}
{\rm N}_{\mathfrak g}(\mathfrak s_{1.3}^\rho)=
\big\langle D^{\rm s},\,P^x(1),\,P^y(\rho),\,R^y(\beta)-R^x(\rho\beta),\,Z(\sigma)\big\rangle
\quad\mbox{if}\quad \rho_t\ne0,
\\
{\rm N}_{\mathfrak g}(\mathfrak s_{1.3}^\rho)=
\big\langle D^t(1),\,D^t(t),\,D^{\rm s},\,P^x(1),\,P^y(\rho),\,R^y(\beta)-R^x(\rho\beta),\,Z(\sigma)\big\rangle
\quad\mbox{if}\quad \rho_t=0.
\end{gather*}
The vector fields~$D^{\rm s}$, $P^x(1)+P^y(\rho)$, $P^y(\rho)$, $R^y(\beta)-R^x(\rho\beta)$, $Z(\sigma)$
and, if $\rho_t=0$, $D^t(1)$ and~$D^t(t)$ from \smash{${\rm N}_{\mathfrak g}(\mathfrak s_{1.3}^\rho)$}
induce the Lie-symmetry vector fields
$D^2$, 0, $P^2$,
$R(\tilde\alpha)$ with $\tilde\alpha(z_1)=\rho(t)\beta(t)$,
$Z(\tilde\sigma)$ with $\tilde\sigma(z_1)=\sigma(t)$ and, if $\rho_t=0$,
$P^1$ and $D^1+\frac13D^2$
of the reduced equation~\eqref{eq:dNs1.3RhoNe1ModifiedRedEq}, respectively.
All the elements of~$\mathfrak a_{1.3}$
from the set complement of the linear span~\smash{$\mathfrak a_{1.3}^{\rho,\rm ind}$}
of the above vector fields from~$\mathfrak a_{1.3}$
are genuinely hidden symmetries%
\footnote{%
The term \emph{hidden symmetries} in this sense was introduced in~\cite{yeho2004a}.
The same notion has other names in the literature, e.g.,
\emph{additional} \cite[Example~3.5]{olve1993A}
or \emph{Type-II hidden} \cite{abra2006a,abra2006b} symmetries
or \emph{noninduced} symmetries of the corresponding submodels \cite{fush1994a,fush1994b,popo1995b}.
Hidden symmetries of a system of partial differential equations were found for the first time in~\cite{kapi1978a};
an accessible description of these results was presented in \cite[Example~3.5]{olve1993A}.
See also \cite[footnote~3]{vinn2024a} for a brief discussion
and references to examples with comprehensive studies of hidden symmetries
of particular systems of differential equations.
}
of the equation~\eqref{eq:dN}.
Note that whether the Lie-symmetry vector fields $P^1$ and $D^1+\frac13D^2$ of~\eqref{eq:dNs1.3RhoNe1ModifiedRedEq}
are induced depends on the value of the parameter function~$\rho$,
which is involved neither in the reduced equation~\eqref{eq:dNs1.3RhoNe1ModifiedRedEq}
nor in its maximal Lie invariance algebra~$\mathfrak a_{1.3}$.

The above description of the induced Lie-symmetry vector fields of the reduced equation~\eqref{eq:dNs1.3RhoNe1ModifiedRedEq}
leads to the description of the induced continuous symmetry transformations of this equation.
Singling out the entire pseudosubgroup~$G_{1.3}^{\rho,\rm ind}$ of~$G_{1.3}$
constituted by the point symmetry transformations of~\eqref{eq:dNs1.3RhoNe1ModifiedRedEq}
that are induced under the Lie reduction of~\eqref{eq:dN} with respect to the subalgebra~$\mathfrak s_{1.3}^\rho$
is a much more difficult problem and depends on~$\rho$ in a more complicated way.
To solve this problem, we first find the stabilizer~${\rm St}_G(\mathfrak s_{1.3}^\rho)$
of the subalgebra~$\mathfrak s_{1.3}^\rho$ in the point-symmetry pseudogroup~$G$ of the equation~\eqref{eq:dN}
for each fixed admissible value of the parameter function~$\rho$.
Denote by $\hat G$ the pseudosubgroup of~$G$
that is constituted by the transformations~\eqref{eq:dNPointSymForm}.
Then $G\setminus\hat G=\mathscr J\circ\hat G$.
The pseudosubgroup ${\rm St}_G(\mathfrak s_{1.3}^\rho)\cap\hat G$
of the pseudosubgroup~${\rm St}_G(\mathfrak s_{1.3}^\rho)$ is singled out from~$\hat G$
by the constraints
\begin{gather*}
X^0_t=0,\quad
(\rho^{-1}Y^0)_t=0,\quad
W^1+\rho W^2=-\frac{C\rho_t}{2T_t^{2/3}}(Y^0)^2,\quad
T_{tt}=0,\quad \rho(T)=\rho(t),
\end{gather*}
i.e., $X^0=b_0$, $Y^0=b_1\rho$ and $T=b_2t+b_3$,
where $b_0$, \dots, $b_3$ are arbitrary constants with $b_2\ne0$ such that $\rho(T)=\rho(t)$.
Its complement ${\rm St}_G(\mathfrak s_{1.3}^\rho)\cap(G\setminus\hat G)$ in ${\rm St}_G(\mathfrak s_{1.3}^\rho)$
is singled out from $G\setminus\hat G$ by the constraints
\begin{gather*}
(X^0\rho)_t=0,\quad
Y^0_t=0,\quad
W^1+\rho W^2=\frac{C\rho_t}{2\rho T_t^{2/3}}(X^0)^2,\quad
(\rho^3T_t)_t=0,\quad \rho(T)=\frac1{\rho(t)},
\end{gather*}
i.e., $X^0=b_1/\rho$, $Y^0=b_0$ and $T=-b_2\int\rho^{-3}{\rm d}t+b_3$,
where $b_0$, \dots, $b_3$ are arbitrary constants with $b_2\ne0$
such that $\rho(T)=1/\rho(t)$.
Then we push forward the elements of~${\rm St}_G(\mathfrak s_{1.3}^\rho)$ by the point transformation
from the space with the coordinates $(t,x,y,u)$ to the space with the coordinates $(z_1,z_2,z_3,w)$
whose $z_1$-, $z_2$- and~$w$-components are defined in~\eqref{eq:dNs1.3RhoNe1ModifiedAnsatz}
and the $z_3$-component is, e.g., $z_3=y$.
Finally, we naturally project the pushed forward transformations to the space with the coordinates $(z_1,z_2,w)$.
As a result, we obtain that the pseudosubgroup~$G_{1.3}^{\rho,\rm ind}$ of~$G_{1.3}$
consists of transformations of the form
\begin{gather*}
\tilde z_1=\hat b_2z_1+\hat b_3, \quad
\tilde z_2=\hat C\hat b_2^{1/3}z_2+\hat b_1, \quad
\tilde w=\hat C^3w+\hat W^1(z_1)z_2+\hat W^0(z_1).
\end{gather*}
Here $\hat C$ and $\hat b_1$ are arbitrary constants with $\hat C\ne0$,
which correspond to the above constants $C$ and $b_1-b_0$, respectively,
and $\hat W^1$ and $\hat W^0$ are arbitrary sufficiently smooth functions of~$z_1$.
The expressions for these functions in terms of the parameters of~$G_{1.3}$ are defined by which transformations,
from ${\rm St}_G(\mathfrak s_{1.3}^\rho)\cap\hat G$
or from ${\rm St}_G(\mathfrak s_{1.3}^\rho)\cap(G\setminus\hat G)$, are considered.
For these two induction cases, we respectively have
\begin{gather*}
\hat W^1(z_1)=-W^1(t),\quad
\hat W^0(z_1)=W^0(t)+\frac{b_1^3}{6b_2}\rho_t(t)\rho^2(t),
\\
\hat W^1(z_1)=-W^1(t)+\frac{Cb_0^2\rho_t(t)}{2b_2^{2/3}\rho(t)},\quad
\hat W^0(z_1)=W^0(t)+\frac{b_1^3\rho_t(t)}{6b_2\rho(t)},
\end{gather*}
where $t$ and~$z_1$ are related via the second equality in~\eqref{eq:dNs1.3RhoNe1ModifiedAnsatz}.
Under the induction, the constants~$\hat b_2$ and~$\hat b_3$ are defined by
\begin{gather*}
\hat b_2=b_2, \quad
\hat b_3=-2b_2\int_{t_0}^{T^{-1}(t_0)}\frac{\rho^3(t)-1}{\rho^3(t)}{\rm d}t,
\end{gather*}
where $t_0$ is the fixed lower limit of the integral with variable upper limit~$t$
taken as the fixed antiderivative in~\eqref{eq:dNs1.3RhoNe1ModifiedAnsatz}.
Recall that $T=b_2t+b_3$ with $\rho(T)=\rho(t)$ and
$T=-b_2\int\rho^{-3}{\rm d}t+b_3$ with $\rho(T)=1/\rho(t)$
for the inducing transformations from ${\rm St}_G(\mathfrak s_{1.3}^\rho)\cap\hat G$
and ${\rm St}_G(\mathfrak s_{1.3}^\rho)\cap(G\setminus\hat G)$, respectively.
Therefore, the set that the tuple $(\hat b_2,\hat b_3)$ runs through
depends on the value of the parameter function~$\rho$.
If $\rho$ is a constant function, then there are no specific constraints on~$\hat b_2$ and~$\hat b_3$,
i.e., these constants are arbitrary with $\hat b_2\ne0$.
In other words, the pseudosubgroup~\smash{$G_{1.3}^{\rho,\rm ind}$} with constant~$\rho$ is singled out from~$G_{1.3}$
by the constraints $c_3=c_5=0$ and is maximal among such pseudosubgroups with respect to the inclusion relation.
In the case of general~$\rho$, we have $\hat b_2=1$ and $\hat b_3=0$,
i.e., the pseudosubgroup $G_{1.3}^{\rho,\rm ind}=:G_{1.3}^{\rm ind,gen}$ is singled out from~$G_{1.3}$
by the more constraints $c_2=c_3=c_1-c_4=c_5=0$
and is minimal among such pseudosubgroups with respect to the inclusion relation.
In both above cases for~$\rho$, the pseudosubgroup ${\rm St}_G(\mathfrak s_{1.3}^\rho)\cap\hat G$
and its complement ${\rm St}_G(\mathfrak s_{1.3}^\rho)\cap(G\setminus\hat G)$ in ${\rm St}_G(\mathfrak s_{1.3}^\rho)$
induce the same set of point symmetry transformations of~\eqref{eq:dNs1.3RhoNe1ModifiedAnsatz},
which coincides with~$G_{1.3}^{\rho,\rm ind}$.
For other values of the parameter function~$\rho$,
elements of~$G$ induce, up to composing with elements of~\smash{$G_{1.3}^{\rm ind,gen}$},
only discrete subsets of~$G_{1.3}$.
For example, if $\rho$ is a general periodic function with period~$\rm T$,
then the shifts of~$t$ by $n\rm T$, $n\in\mathbb Z$, as an element of~$G$
induce the shift of~$z_1$ by $n\hat{\rm T}$ with
$\hat{\rm T}:=2\int_0^{\rm T}\big(1-\rho^{-3}(t)\big){\rm d}t$,
which belongs~to~$G_{1.3}$. %\looseness=-1

\section{Classification of appropriate one-dimensional subalgebras}\label{sec:dNs1.3RhoNe1Subalgebras}

We compute the adjoint action of the pseudogroup~$G_{1.3}$ on the algebra~$\mathfrak a_{1.3}$
via pushing forward the spanning vector fields of~$\mathfrak a_{1.3}$ by the elementary transformations.
This way is more convenient in the infinite-dimensional case~\cite{bihl2012b,card2011a}
than the classical approach based on constructing inner automorphisms~\cite[Section~3.3]{olve1993A}.
In addition, it allows one to use not only the transformations from the identity component of~$G_{1.3}$
but also discrete elements of~$G_{1.3}$.
The nonidentity adjoint actions of the elementary transformations from~$G_{1.3}$
on vector fields spanning~$\mathfrak a_{1.3}$ are given by
\begin{gather*}
\mathscr P^1_*(c_2)D^1=D^1-c_2P^1,\quad
\mathscr P^1_*(c_2)K=K-c_2(2D^1+D^2)+c_2^2P^1,\quad
\mathscr P^1_*(c_2)H=H-c_2P^2,\\
\mathscr P^1_*(c_2)R(\alpha)=R(\tilde\alpha^1),\quad
\mathscr P^1_*(c_2)Z(\sigma)=Z(\tilde\sigma^1),
\\[2ex]
\mathscr D^1_*(c_1)P^1=c_1P^1,\quad
\mathscr D^1_*(c_1)K=c_1^{-1}K,\quad
\mathscr D^1_*(c_1)H=c_1^{-1}H,\\
\mathscr D^1_*(c_1)R(\alpha)=c_1^{-1}R(\tilde\alpha^2),\quad
\mathscr D^1_*(c_1)Z(\sigma)=c_1^{-1}Z(\tilde\sigma^2),
\\[2ex]
\mathscr K_*(c_3)P^1=P^1+c_3(2D^1+D^2)+c_3^2K,\quad
\mathscr K_*(c_3)D^1=D^1+c_3K,\\
\mathscr K_*(c_3)P^2=P^2+c_3H,\quad
\mathscr K_*(c_3)R(\alpha)=R(\tilde\alpha^3),\quad
\mathscr K_*(c_3)Z(\sigma)=Z\bigl((1+c_3z_1)\tilde\sigma^3\bigr),\quad
\\[2ex]
\mathscr D^2_*(\tilde c_4)P^2={\tilde c_4}P^2,\quad
\mathscr D^2_*(\tilde c_4)H={\tilde c_4}H,\quad
\mathscr D^2_*(\tilde c_4)R(\alpha)={\tilde c_4}^2R(\alpha),\quad
\mathscr D^2_*(\tilde c_4)Z(\sigma)={\tilde c_4}^3Z(\sigma),
\\[2ex]
\mathscr P^2_*(c_6)K=K-c_6H+\tfrac12c_6^2R(1)-\tfrac16c_6^3Z(1),\quad
\mathscr P^2_*(c_6)D^2=D^2-c_6P^2,\\
\mathscr P^2_*(c_6)H=H-c_6R(1)+\tfrac12c_6^2Z(1),\quad
\mathscr P^2_*(c_6)R(\alpha)=R(\alpha)-c_6Z(\alpha),
\\[2ex]
\mathscr H_*(c_5)P^1=P^1+c_5P^2+\tfrac12c_5^2R(1)-\tfrac16c_5^3Z(z_1),\quad
\mathscr H_*(c_5)D^1=D^1+c_5H,\\
\mathscr H_*(c_5)D^2=D^2-c_5H,\quad
\mathscr H_*(c_5)P^2=P^2+c_5R(1)-\tfrac12c_5^2Z(z_1),\\
\mathscr H_*(c_5)R(\alpha)=R(\alpha)-c_5Z(z_1\alpha),
\\[2ex]
\mathscr R_*(W^1)P^1=P^1+R\bigl(W^1_{z_1}\bigr),\quad
\mathscr R_*(W^1)D^1=D^1+R\bigl(z_1W^1_{z_1}+W^1\bigr),\\
\mathscr R_*(W^1)K=K+R\bigl(z_1^2W^1_{z_1}\bigr),\quad
\mathscr R_*(W^1)D^2=D^2-2R\bigl(W^1\bigr),\\
\mathscr R_*(W^1)P^2=P^2+Z\bigl(W^1\bigr),\quad
\mathscr R_*(W^1)H=H+Z\bigl(z_1W^1\bigr),
\\[2ex]
\mathscr Z_*(W^0)P^1=P^1+Z\bigl(W^0_{z_1}\bigr),\quad
\mathscr Z_*(W^0)D^1=D^1+Z\bigl(z_1W^0_{z_1}+W^0\bigr),\\
\mathscr Z_*(W^0)K=K+Z\bigl(z_1^2W^0_{z_1}-z_1W^0\bigr),\quad
\mathscr Z_*(W^0)D^2=D^2-3Z\big(W^0\big),
\\[2ex]
\mathscr Q^+_*(\tilde c_3)P^1=\mathsf c^2P^1+\mathsf c\,\mathsf s\,(2D^1+D^2)+\mathsf s^2K,\quad
\mathscr Q^+_*(\tilde c_3)K=\mathsf c^2K-\mathsf c\,\mathsf s\,(2D^1+D^2)+\mathsf s^2P^1,\\
\mathscr Q^+_*(\tilde c_3)D^1=D^1-\mathsf c\,\mathsf s\,(P^1-K)-\mathsf s^2(2D^1+D^2),\\
\mathscr Q^+_*(\tilde c_3)P^2=\mathsf c\,P^2+\mathsf s\,H,\quad
\mathscr Q^+_*(\tilde c_3)H=\mathsf c\,H-\mathsf s\,P^2,\\
\mathscr Q^+_*(\tilde c_3)R(\alpha)=R(\tilde\alpha^4),\quad
\mathscr Q^+_*(\tilde c_3)Z(\sigma)=Z(\tilde\sigma^4),
\end{gather*}
where $c_1$, $c_2$, $c_3$, $\tilde c_3$, $\tilde c_4$, $c_5$ and $c_6$ are arbitrary constants with $c_1\ne0$,
$W^0$~and~$W^1$ are arbitrary smooth functions of~$z_1$,
$\tilde\alpha^1(z_1):=\alpha^1(z_1-c_2)$,
$\tilde\sigma^1(z_1):=\sigma^1(z_1-c_2)$,
$\tilde\alpha^2(z_1):=\alpha^2(c_1^{-1}z_1)$,
$\tilde\sigma^2(z_1):=\sigma^2(c_1^{-1}z_1)$,
$\tilde\alpha^3(z_1):=\alpha^3\bigl(z_1(1+c_3z_1)^{-1}\bigr)$,
$\tilde\sigma^3(z_1):=\sigma^3\bigl(z_1(1+c_3z_1)^{-1}\bigr)$,
$\mathsf c:=\cos\tilde c_3$, $\mathsf s:=\sin\tilde c_3$
and
\[
\tilde\alpha^4(z_1)=\alpha^4\!\left(\frac{\mathsf c\,z_1-\mathsf s}{\mathsf s\,z_1+\mathsf c}\right),\quad
\tilde\sigma^4(z_1)=(\mathsf s\,z_1+\mathsf c)\,
\sigma^4\!\left(\frac{\mathsf c\,z_1-\mathsf s}{\mathsf s\,z_1+\mathsf c}\right).
\]

\begin{lemma}\label{lem:dNSubmodel1DSubalgsProperty}
Any one-dimensional subalgebra~$\mathfrak b$ of~$\mathfrak a_{1.3}$ that is appropriate
for Lie reduction  of the equation~\eqref{eq:dNs1.3RhoNe1ModifiedRedEq}
is $G_{1.3}$-equivalent to a subalgebra contained in the span $\mathfrak u:=\langle P^1,D^1,K,D^2,P^2,H\rangle$
or in the span $\langle P^2,R(\alpha)\rangle$.
\end{lemma}

\begin{proof}
A one-dimensional subalgebra~$\mathfrak b$ of~$\mathfrak a_{1.3}$ is appropriate
for Lie reduction  of the equation~\eqref{eq:dNs1.3RhoNe1ModifiedRedEq} if
its natural projection to $\mathfrak u$ is nonzero.
In other words, if $\mathfrak b=\langle Q\rangle$ with $Q=a_1P^1+a_2D^1+a_3K+a_4D^2+a_5P^2+a_6H+R(\alpha^0)+Z(\sigma^0)$,
then $a_i\ne0$ for some $i\in\{1,\dots,6\}$.

If at least one of the coefficients $a_1$, \dots, $a_4$ is nonzero,
then we successively push forward~$\mathfrak b$ by~$\mathscr R(W^1)$ and $\mathscr Z(W^0)$
with appropriate values of the parameter functions~$W^0$ and~$W^1$ to set $\alpha^0=0$ and $\sigma^0=0$.
This means that the pushed forward subalgebra is contained in~$\mathfrak u$.

If $a_1=\dots=a_4=0$, then $(a_5,a_6)\ne(0,0)$.
Successively pushing forward~$\mathfrak b$ by~$\mathscr Q^+(\tilde c_3)$ and~$\mathscr R(W^1)$
with appropriate values of constant~$\tilde c_3$ and the parameter function~$W^0$ and scaling~$Q$,
we can set $a_6=0$, $\sigma^0=0$ and $a_5=1$, respectively.
As a result, $Q=P^2+R(\alpha^0)$.
\end{proof}

The subalgebras contained in $\mathfrak u$ are $G_{1.3}$-equivalent
if and only if their images under the homomorphism
$\boldsymbol\upsilon\colon\mathfrak a_{1.3}\to\mathfrak a_{\rm iB}$
(see the end of Section~\ref{sec:dNRedEq1.3MIA}) are \smash{$\check G_{1.3}$}-equivalent.
The one-dimensional subalgebras of the algebra~$\check{\mathfrak a}_{1.3}$ up to the \smash{$\check G_{1.3}$}-equivalence
were classified in \cite[Table~2]{poch2017a},
and the exhaustive classification of subalgebras of the affine Lie algebra ${\rm aff}(2,\mathbb R)$,
which is isomorphic to~$\check{\mathfrak a}_{1.3}$, was carried out in~\cite[Theorem~11]{chap2024a}.
The classification list for dimension one consists of the subalgebras
\begin{gather*}
\check{\mathfrak b}_{1.0}^\alpha=\big\langle\check P^2\big\rangle,\quad
\check{\mathfrak b}_{1.1}=\big\langle\check D^2\big\rangle,\quad
\check{\mathfrak b}_{1.2}=\big\langle\check P^1\big\rangle,\quad
\check{\mathfrak b}_{1.3}=\big\langle\check P^1+\check H\big\rangle,\\
\check{\mathfrak b}_{1.4}=\big\langle\check P^1+\check D^2\big\rangle,\quad
\check{\mathfrak b}_{1.5}=\big\langle\check D^1+\check P^2\big\rangle,\quad
\check{\mathfrak b}_{1.6}^a=\big\langle\check D^1+a\check D^2\big\rangle_{a\geqslant\frac12\ (\!\bmod\check G_{1.3})},\\
\check{\mathfrak b}_{1.7}=\big\langle\check D^1+\check D^2+\check H\big\rangle,\quad
\check{\mathfrak b}_{1.8}^a=\big\langle\check P^1+\check K+a\check D^2\big\rangle_{a\geqslant0\ (\!\bmod\check G_{1.3})}.
\end{gather*}
In view of these notes and Lemma~\ref{lem:dNSubmodel1DSubalgsProperty}, we obtain the following assertion.

\begin{lemma}\label{lem:dNSubmodel1DSubalgs}
A complete list of $G_{1.3}$-inequivalent one-dimensional subalgebras of the algebra~$\mathfrak a_{1.3}$
that are appropriate for Lie reductions of the equation~\eqref{eq:dNs1.3RhoNe1ModifiedRedEq}
is exhausted by the following subalgebras:
\begin{gather*}
\mathfrak b_{1.0}^\alpha=\big\langle P^2+R(\alpha)\big\rangle,\quad
\mathfrak b_{1.1}=\big\langle D^2\big\rangle,\quad
\mathfrak b_{1.2}=\big\langle P^1\big\rangle,\quad
\mathfrak b_{1.3}=\big\langle P^1+H\big\rangle,\\
\mathfrak b_{1.4}=\big\langle P^1+D^2\big\rangle,\quad
\mathfrak b_{1.5}=\big\langle D^1+P^2\big\rangle,\quad
\mathfrak b_{1.6}^a=\big\langle D^1+aD^2\big\rangle_{a\geqslant\frac12\ (\!\bmod G_{1.3})},\\
\mathfrak b_{1.7}=\big\langle D^1+D^2+H\big\rangle,\quad
\mathfrak b_{1.8}^a=\big\langle P^1+K+aD^2\big\rangle_{a\geqslant0\ (\!\bmod G_{1.3})},
\end{gather*}
where $\alpha$ runs through the set of smooth function of~$z_1$ and $a$ is an arbitrary constant.
\end{lemma}

\begin{remark}
In Lemma~\ref{lem:dNSubmodel1DSubalgs}, we assume
that only $G_{1.3}$-equivalent subalgebras from the family $\{\mathfrak b_{1.0}^\alpha\}$ are chosen.
A subalgebra~$\mathfrak b_{1.0}^\alpha$ is mapped
by a transformation of the form~\eqref{eq:dNs1.3RhoNe1ModifiedRedEqPointSymForm}
to a subalgebra~$\mathfrak b_{1.0}^{\tilde\alpha}$ if and only if
$c_3=0$ and $W^1(z_1)=c_1^{-1}c_4^{-2}(c_5z_1+c_6)\big(\alpha(z_1)+c_5\big)$.
Hence subalgebras $\mathfrak b_{1.0}^\alpha$ and~$\mathfrak b_{1.0}^{\tilde\alpha}$ are $G_{1.3}$-equivalent
if and only if there exist constants $c_1$, $c_2$, $c_4$ and~$c_5$ with $c_1c_4\ne0$ such that
$\tilde\alpha(\tilde z_1)=c_1^{-1}\big(\alpha(z_1)+c_5\big)$, where $\tilde z_1=c_4^{-1}(c_1z_1+c_2)$.
\end{remark}

\section{Lie invariant solutions}\label{sec:dNs1.3LieInvSols}

We avoid directly constructing Lie invariant solutions of the equation~\eqref{eq:dNs1.3RhoNe1ModifiedRedEq}.
Instead of this, we apply an equivalent but simpler approach.
We carry out the Lie reduction procedure for the equation~\eqref{eq:dNs1.3RhoNe1InviscidBurgersEq},
integrate twice the obtained invariant solutions of~\eqref{eq:dNs1.3RhoNe1InviscidBurgersEq} with respect to~$z_2$
and, modulo the $G_{1.3}$-inequivalence on the solution set of~\eqref{eq:dNs1.3RhoNe1ModifiedRedEq},
neglect trivial summands of the form $\breve W^1(z_1)z_2+\breve W^0(z_1)$ arising in the course of the integration.
Here $\breve W^1$ and $\breve W^0$ are arbitrary sufficiently smooth functions of~$z_1$.

Each of the $\check G_{1.3}$-inequivalent one-dimensional subalgebras of the algebra~$\check{\mathfrak a}_{1.3}$
that are listed before Lemma~\ref{lem:dNSubmodel1DSubalgs}
are appropriate to be used for Lie reduction of~\eqref{eq:dNs1.3RhoNe1InviscidBurgersEq}.
The corresponding ansatzes and reduced equations are collected in Table~\ref{TableRedEqs},
where $\varphi=\varphi(\omega)$ is the new unknown function of the single invariant variable~$\omega$.

\begin{table}[ht!]\footnotesize
\begin{center}\renewcommand{\arraystretch}{1.4}
\caption{Lie reductions with respect to one-dimensional subalgebras of~$\mathfrak {\hat a}_{1.3}$.\strut}\label{TableRedEqs}
\begin{tabular}{|l|l|l|l|l|}
\hline
\hfil$\mathfrak \subset\mathfrak g$
&\hfil Basis &\hfil Ansatz, $\varphi=\varphi(\omega)$
& \hfil~$\omega$ &\hfil Reduced equation
\\
\hline
$\check{\mathfrak b}_{1.0}$   &$\check P^2$                 &$h=\varphi$                &$z_1$                           &$\varphi_\omega=0$
\\[1.5ex]
$\check{\mathfrak b}_{1.1}$   &$\check D^2$                 &$h=z_2\varphi$             &$z_1$                           &$\varphi_\omega+\varphi^2=0$
\\[1.5ex]
$\check{\mathfrak b}_{1.2}$   &$\check P^1$                 &$h=\varphi$                &$z_2$                           &$\varphi\varphi_\omega=0$
\\[1.5ex]
$\check{\mathfrak b}_{1.3}$   &$\check P^1+\check H$          &$h=\varphi+z_1$            &$z_2-\dfrac{z_1^2}2$            &$\varphi\varphi_\omega+1=0$
\\[1.5ex]
$\check{\mathfrak b}_{1.4}$   &$\check P^1+\check D^2$        &$h=e^{z_1}\varphi$         &$e^{-z_1}z_2$                   &$\varphi\varphi_\omega-\omega\varphi_\omega+\varphi=0$
\\[1.5ex]
$\check{\mathfrak b}_{1.5}$   &$\check D^1+\check P^2$        &$h=z_1^{-1}\varphi$        &$z_2-\ln|z_1|$                  &$\varphi\varphi_\omega-\varphi_\omega-\varphi=0$
\\[1.5ex]
$\check{\mathfrak b}_{1.6}^a$ &$\check D^1+a\check D^2$       &$h=z_1^{-1}|z_1|^a\varphi$ &$|z_1|^{-a}z_2$                 &$\varphi\varphi_\omega-a\omega\varphi_\omega+(a-1)\varphi=0$
\\[1.5ex]
$\check{\mathfrak b}_{1.7}$   &$\check D^1+\check D^2+\check H$ &$h=\varphi+\ln|z_1|$       &$\dfrac {z_2}{z_1}-\ln|z_1|$    &$\varphi\varphi_\omega-(\omega+1)\varphi_\omega+1=0$
\\[1.5ex]
$\check{\mathfrak b}_{1.8}^a$ &$\check P^1+\check K+a\check D^2$&$h=\dfrac{{\rm e}^{a\arctan z_1}}{\sqrt{z_1^2+1}}\varphi+\dfrac{z_1+a}{z_1^2+1}z_2$  &$\dfrac{{\rm e}^{-a\arctan z_1}}{\sqrt{z_1^2+1}}z_2\!$
&$\varphi\varphi_\omega+2a\varphi+(a^2+1)\omega=0\!$
\\[1.8ex]
\hline
\end{tabular}
\end{center}
\end{table}

After integrating each of the listed reduced equations, we present
the corresponding solutions~$h$ and~$w$ of the equations~\eqref{eq:dNs1.3RhoNe1InviscidBurgersEq} and~\eqref{eq:dNs1.3RhoNe1ModifiedRedEq}
up to the $\check G_{1.3}$- and $G_{1.3}$-equivalences, respectively,
omitting most of the related explanations.

Below $c_0$ and $c_1$ are arbitrary constants with $c_1\ne0$.
Reduced equations~1.0--1.2 and 1.4--1.6 have the solutions $\varphi=0$
that are trivial and will be neglected since
they correspond to the zero solutions of~\eqref{eq:dNs1.3RhoNe1InviscidBurgersEq} and~\eqref{eq:dNs1.3RhoNe1ModifiedRedEq}.
For readers' convenience,
we marked the constructed solutions of the reduced equation~\eqref{eq:dNs1.3RhoNe1ModifiedRedEq}
by the symbol~$\circ$
and the form of the corresponding inequivalent solutions of the dispersionless Nyzhnyk equation~\eqref{eq:dN}
by the bullet symbol~$\bullet$\,.

\medskip\par\noindent{\bf 1.0.}\
%$\check{\mathfrak b}_{1.0}=\big\langle\check P^2\big\rangle$: \
%$h=\varphi$, \ $\omega=z_1$; \ $\varphi_\omega=0$.
Reduced equation~1.0 trivially integrates to~$\varphi=c_0$.
Transformations from $\{\check{\mathscr H}(c_5)\}$ induce shifts of $\varphi$,
and thus we can set $\varphi=0$ modulo the equivalence induced by the action of~$\check G_{1.3}$,
which gives $h=0$ and $w=0$.

This case is singular in the sense of the correspondence between Lie reductions
of the equations~\eqref{eq:dNs1.3RhoNe1InviscidBurgersEq} and~\eqref{eq:dNs1.3RhoNe1ModifiedRedEq}.
More specifically, the $\check G_{1.3}$-inequivalent subalgebra~$\check{\mathfrak b}_{1.0}$ of~$\check{\mathfrak a}_{1.3}$
is associated with the family ${\big\{\mathfrak b_{1.0}^\alpha=\big\langle P^2+R(\alpha)\big\rangle\big\}}$
of the $G_{1.3}$-inequivalent subalgebras of~$\mathfrak a_{1.3}$.
An ansatz constructed for~$w$ using the subalgebra~$\mathfrak b_{1.0}^\alpha$
with a fixed value of the parameter function~$\alpha$ if
$w=\psi(\omega)+\frac12\alpha(z_1)z_2^{\,2}$,
where $\psi=\psi(\omega)$ is the new unknown function of the single invariant variable~$\omega=z_1$.
The corresponding reduced equation $\alpha_{z_1}=0$ is inconsistent if $\alpha\ne\const$
and is an identity otherwise.
In the latter case, the subalgebra~$\mathfrak b_{1.0}^\alpha$ is in fact a subalgebra of~$\mathfrak a_{1.3}^{\rm ess}$,
and each of the obtained solutions of~\eqref{eq:dNs1.3RhoNe1ModifiedRedEq}
is $G_{1.3}$-equivalent to the above zero solution $w=0$.

\medskip\par\noindent{\bf 1.1.}\
%$\check{\mathfrak b}_{1.1}=\big\langle\check D^2\big\rangle$: \
%$h=z_2\varphi$, \,\ $\omega=z_1$; \,\ $\varphi_\omega+\varphi^2=0$.
The general solution of reduced equation~1.1 is $\varphi=(\omega+c_0)^{-1}$.
The subgroup of~$\check G_{1.3}$ singled out by the constraint $c_5=c_6=0$
induces the point symmetry group of reduced equation~1.1 that consists of the transformations
\begin{gather*}
\tilde\omega=\frac{c_1\omega+c_2}{c_3\omega+c_4},\quad
\tilde\varphi=(c_3\omega+c_4)\frac{(c_3\omega+c_4)\varphi-c_3}\Delta
\end{gather*}
with the modified composition of transformations \cite{kova2023b,kova2023a} as the group operation,
where $c_1$,~\dots,~$c_4$ are arbitrary constants with $\Delta=c_1c_4-c_2c_3\ne0$,
cf.\ \eqref{eq:dNSubmodelTransportInducedPointSymForm}.
Any of the latter transformations with $c_3=1$ and $c_4=c_0$ maps $\varphi=(\omega+c_0)^{-1}$ to $\varphi=0$.

\medskip\par\noindent{\bf 1.2.}\
%$\check{\mathfrak b}_{1.2}=\big\langle\check P^1\big\rangle$: \
%$h=\varphi$, \,\ $\omega=z_2$; \,\ $\varphi\varphi_\omega=0$.
The solutions of reduced equation~1.2 are exhausted by the constant ones $\varphi=c_0$.
The corresponding solutions of the equations~\eqref{eq:dNs1.3RhoNe1InviscidBurgersEq} and~\eqref{eq:dNs1.3RhoNe1ModifiedRedEq}
are $\{\check{\mathscr H}(c_5)\}$- and $\{\mathscr H(c_5)\}$-equivalent to the zero solutions of these equations,
cf.\ Case~1.0.

\medskip\par\noindent{\bf 1.3.}\
%$\check{\mathfrak b}_{1.3}=\big\langle\check P^1+\check H\big\rangle$: \
%$h=\varphi+z_1$, \,\ $\omega=z_2-\frac12z_1^2$; \,\ $\varphi\varphi_\omega+1=0$.
Reduced equation~1.3 integrates to $\varphi=\pm(-2\omega+c_0)^{1/2}$.
Up to shifts with respect to~$z_2$, this gives $h=\pm(z_1^2-2z_2)^{1/2}+z_1$ and %$h=\pm(z_1^2-2z_2+c_0)^{1/2}+z_1$
\begin{gather*}
\solutionRedEq
w=\pm\frac{1}{15}(z_1^2-2z_2)^{5/2}+\frac12z_1z_2^2.
%w=\pm\frac{1}{15}(z_1^2-2z_2+c_0)^{5/2}+\tfrac12z_1z_2^2.
\end{gather*}

\noindent{\bf 1.4.}\
%$\check{\mathfrak b}_{1.4}=\big\langle\check P^1+\check D^2\big\rangle$: \
%$h=e^{z_1}\varphi$, \,\ $\omega={\rm e}^{-z_1}z_2$; \,\ $\varphi\varphi_\omega-\omega\varphi_\omega+\varphi=0$.
The solution set of reduced equation~1.4 consists of the functions
\begin{gather*}
\varphi=-\omega/\zeta,\quad
\zeta\in\big\{W_0(\tilde\omega),W_{-1}(\tilde\omega)\big\},\quad
\tilde\omega:=c_1\omega,
\end{gather*}
where~$W_0$ and~$W_{-1}$ are the principal real and the other real branches of the Lambert $W$ function, respectively.
Up to scalings induced by $\{\check{\mathscr D}^2(\tilde c_4)\}$, we can set $c_1=1$.
As a result, $\tilde\omega=\omega={\rm e}^{-z_1}z_2$, $h=-z_2/\zeta$ and
\begin{gather*}
\solutionRedEq
w=-z_2^3\frac{18\zeta^2+15\zeta+4}{108\zeta^3}.
\end{gather*}

\noindent{\bf 1.5.}\
%$\check{\mathfrak b}_{1.5}=\big\langle\check D^1+\check P^2\big\rangle$: \
%$h=z_1^{-1}\varphi$, \,\ $\omega=z_2-\ln|z_1|$; \,\ $\varphi\varphi_\omega-\varphi_\omega-\varphi=0$.
Reduced equation~1.5 also integrates in terms of the Lambert $W$ function,
\begin{gather*}
\varphi=-\zeta,\quad
\zeta\in\big\{W_0(\tilde\omega),W_{-1}(\tilde\omega)\big\},\quad
\tilde\omega:=c_1{\rm e}^{-\omega},
\end{gather*}
where~$W_0$ and~$W_{-1}$ are the principal real and the other real branches of the Lambert $W$ function, respectively,
and the integration constant~$c_1$ can be set to be equal $\sgn z_1$ modulo scalings induced by $\{\check{\mathscr D}^2(\tilde c_4)\}$.
We obtain $\tilde\omega={\rm e}^{-\omega}=z_1{\rm e}^{-z_2}$, $h=-z_1^{-1}\zeta$ and
\begin{gather*}
\solutionRedEq
%w=-\frac{\zeta}{z_1}\left(\frac16\zeta^2+\frac34\zeta+1\right).
w=-\zeta\frac{2\zeta^2+9\zeta+12}{12z_1}.
\end{gather*}

\noindent{\bf 1.6.}\
%$\check{\mathfrak b}_{1.6}^a=\big\langle \check D^1+a\check D^2 \big\rangle$,
%where $a\geqslant\frac12$: \
%$h=z_1^{-1}|z_1|^a\varphi$, \ $\omega=|z_1|^{-a}z_2$; \ $\varphi\varphi_\omega-a\omega\varphi_\omega+(a-1)\varphi=0$.
For any value of~$a$, reduced equation~1.6$^a$ has the solution $\varphi=\omega$.
The corresponding solution $h=z_1^{-1}z_2$ of the equation~\eqref{eq:dNs1.3RhoNe1InviscidBurgersEq}
is trivial since it is $\check G_{1.3}$-equivalent to the zero solution.
We neglect this solution below.

Recall that $a\geqslant\frac12$ ($\!{}\bmod\check G_{1.3}$) since
the pushforward $\check{\mathscr Q}^+_*(\frac12\pi)$ maps the subalgebra~$\check{\mathfrak b}_{1.6}^a$
to the subalgebra~$\check{\mathfrak b}_{1.6}^{1-a}$.
We separately consider the cases with $a=1$ and with general values of~$a$.
In the last case, we additionally single out two subcases, $a=2$ and $a=1/2$,
where the general solutions of the corresponding reduced equations can be represented explicitly.

In addition to $\varphi=\omega$,
the solution set of reduced equation~1.6$^1$, $(\varphi-\omega)\varphi_\omega=0$,
includes only the constant functions $\varphi=c_0$.
The corresponding solutions $h=c_1$ of the equation~\eqref{eq:dNs1.3RhoNe1InviscidBurgersEq}
are obviously $\check G_{1.3}$-equivalent to the zero solution.

Below $a\geqslant\frac12$ and $a\ne1$.
The general solution of reduced equation~1.6$^a$ can be represented implicitly in the form
\begin{gather}\label{eq:dNs1.3RhoNe1TransportEqRedEq1.6GenSolution}
\omega=\varphi-c_0|\varphi|^{\tfrac{a}{a-1}}.
\end{gather}
If $\varphi\ne\omega$, then $c_0\ne0$, and modulo the equivalence induced by the action of~$\check G_{1.3}$,
we can set $c_0$ to any nonzero value.

Choosing $c_0=1/4$, we easily solve the equation~\eqref{eq:dNs1.3RhoNe1TransportEqRedEq1.6GenSolution}
as a quadratic equation with respect to~$\varphi$ for the value $a=1/2$,
which results in an explicit solution of reduced equation~1.6$^{1/2}$
and the corresponding explicit solution of the equation~\eqref{eq:dNs1.3RhoNe1InviscidBurgersEq},
\noprint{
\begin{gather*}%\tilde c_1=4c_0
a=\frac12\colon\quad
\varphi=\frac12\Bigl(\omega\pm\sqrt{\omega^2+\tilde c_1}\,\Bigr),\quad
h=\frac1{2z_1}\Bigl(z_2\pm\sqrt{z_2^2+\tilde c_1z_1}\,\,\Bigr),
\end{gather*}
}
\begin{gather*}
\varphi=\frac12\Bigl(\omega+\sqrt{\omega^2+\varepsilon}\,\Bigr),\quad
h=\frac1{2z_1}\Bigl(z_2+\sqrt{z_2^2+z_1}\,\,\Bigr),
\end{gather*}
where $\varepsilon=\pm1$,
and in addition we simultaneously change
the signs of $(\omega,\varphi)$ and $(z_2,h)$ if necessary to set ``+'' before the square roots
the signs of $(z_1,h)$ if necessary to set $\varepsilon=1$ in~$h$.
The corresponding solution of the equations~\eqref{eq:dNs1.3RhoNe1ModifiedRedEq} is also explicit,
\noprint{
\begin{gather*}
\solutionRedEq
w=\frac1{12z_1}\left(z_2^3\pm\biggl((z_2^2+\tilde c_1z_1)^{3/2}
+3\tilde c_1z_1\Bigl(z_2\ln\Bigl(z_2+\sqrt{z_2^2+\tilde c_1z_1}\,\,\Bigr)
-\sqrt{z_2^2+\tilde c_1z_1}\,\,\Bigr)\biggr)\right),
\end{gather*}
}
\begin{gather*}
\solutionRedEq
w=\frac{z_2^3+(z_2^2+z_1)^{3/2}}{12z_1}
+\frac{z_2}4\ln\Bigl|z_2+\sqrt{z_2^2+z_1}\,\Bigr|
-\frac14\sqrt{z_2^2+z_1}.
\end{gather*}
In the same way, we can also construct solutions
of~\eqref{eq:dNs1.3RhoNe1InviscidBurgersEq} and~\eqref{eq:dNs1.3RhoNe1ModifiedRedEq} for the case $a=2$,
\begin{gather*}
\noprint{
a=2\colon\quad
\varphi=\frac2{\tilde c_1}\Bigl(1\pm\sqrt{1-\tilde c_1\omega}\,\Bigr),\quad
h=\frac2{\tilde c_1}\Bigl(z_1\pm\sqrt{z_1^2-\tilde c_1z_2}\,\,\Bigr),\\
w=\frac1{\tilde c_1}z_1z_2^2\pm\frac8{15\tilde c_1^3}(z_1^2-\tilde c_1z_2)^{5/2},
\\
a=2\colon\quad
\varphi=2\Bigl(1-\sqrt{1-\omega}\,\Bigr),\quad
}
h=2z_1-2\sqrt{z_1^2-z_2},\quad
w=z_1z_2^2-\frac8{15}(z_1^2-z_2)^{5/2},
\end{gather*}
but they are respectively $\check G_{1.3}$- and $G_{1.3}$-equivalent to those
obtained above using the reduction~1.3.

For the other values of~$a$,
we consider $\varphi$ in the equation~\eqref{eq:dNs1.3RhoNe1TransportEqRedEq1.6GenSolution} as a parameter
and denote it by~$s$, thus representing the general solution of reduced equation~1.6$^a$
in a parametric form in a~uniform way as
\[
\varphi=s,\quad \omega=s-c_0|s|^{\tfrac a{a-1}}.
\]
Modulo the induced equivalence, we can set, e.g., $c_0=1$.
The corresponding family of solutions of the equation~\eqref{eq:dNs1.3RhoNe1InviscidBurgersEq} in the parametric form is
\begin{gather}\label{eq:ReducedEq1.6ParamSolTransportEq}
h=\frac{|z_1|^a}{z_1}s, \quad \frac{z_2}{|z_1|^a}=s-c_0|s|^{\tfrac{a}{a-1}}.
\end{gather}
This leads to the following solutions of the equation~\eqref{eq:dNs1.3RhoNe1ModifiedRedEq}:
\begin{gather*}
\solutionRedEq
w=\frac{|z_1|^{3a}}{z_1}\left(\frac{s^3}{6}
-\frac{c_0a(4a-3)}{2(3a-2)(2a-1)}|s|^{\tfrac{3a-2}{a-1}}
+\frac{(c_0a)^2}{(2a-1)(3a-1)}s|s|^{\tfrac{2a}{a-1}}\right)\quad\mbox{if}\quad a\ne\frac23,
\\[1ex]
\solutionRedEq
w=z_1\left(\frac{s^3}6-c_0\ln|s|+\frac43c_0^2 s^{-3}\right)\quad\mbox{if}\quad a=\frac23,
\end{gather*}
where $s$ is defined by the second equation in~\eqref{eq:ReducedEq1.6ParamSolTransportEq}.

\medskip\par\noindent{\bf 1.7.}\
%$\check{\mathfrak b}_{1.7}=\big\langle\check D^1+\check D^2+\check H\big\rangle$: \
%$h=\varphi+\ln|z_1|$, \,\ $\omega=\dfrac {z_2}{z_1}-\ln|z_1|$; \,\ $\varphi\varphi_\omega-(\omega+1)\varphi_\omega+1=0$.
Similarly to Cases~1.4 and~1.5, we derive
\begin{gather*}
\varphi=\omega-\zeta,
\ \
\zeta\in\big\{W_0(\tilde\omega),W_{-1}(\tilde\omega)\big\},\ \
\tilde\omega:=c_1{\rm e}^{\omega},
\end{gather*}
where~$W_0$ and~$W_{-1}$ are the principal real and the other real branches of the Lambert $W$ function, respectively,
and the integration constant~$c_1$ can be set to be equal $\sgn z_1$ modulo scalings induced by $\{\check{\mathscr D}^2(\tilde c_4)\}$.
Hence $\tilde\omega={\rm e}^{-\omega}={\rm e}^{z_2/z_1}/z_1$, $h=z_2z_1^{-1}-\zeta$ and
\begin{gather*}
\solutionRedEq
w=\frac{z_2^3}{6z_1}-\frac{z_1^2}{2}\zeta\left(\frac13\zeta^2+\frac32\zeta+2\right).
\end{gather*}

\par\noindent{\bf 1.8.}\
%$\check{\mathfrak b}_{1.8}^a=\big\langle\check P^1+\check K+a\check D^2\big\rangle$, where $a\in\mathbb R$:\\
%\[
%h=\dfrac{{\rm e}^{a\arctan z_1}}{\sqrt{z_1^2+1}}\varphi+\dfrac{z_1+a}{z_1^2+1}z_2,\quad
%\omega=\dfrac{{\rm e}^{-a\arctan z_1}}{\sqrt{z_1^2+1}}z_2;\quad
%\varphi\varphi_\omega+2a\varphi+(a^2+1)\omega=0.
%\]
Reduced equation 1.8$^0$ can be easily integrated to $\varphi=\pm(c_1-\omega^2)^{1/2}$,
where $c_1>0$ for the solution to be real.
The scaling $(\tilde\omega,\tilde\varphi)=(b\omega,b\varphi)$ induced
by the scaling $\check{\mathscr D}^2(b)$ from the group~$\check G_{1.3}$, where $b=\pm c_1^{1/2}$,
reduces the above solution to the canonical form $\varphi=(1-\omega^2)^{1/2}$,
which gives the following explicit solutions
of the equations~\eqref{eq:dNs1.3RhoNe1InviscidBurgersEq} and~\eqref{eq:dNs1.3RhoNe1ModifiedRedEq}:
\noprint{
\begin{gather*}
h=\frac{z_1z_2\pm\sqrt{c_1(z_1^2+1)-z_2^2}}{z_1^2+1},\\
\solutionRedEq
w=\frac{z_1z_2^3}{6(z_1^2+1)}
\pm\left(\frac{c_1z_2}2\arctan\frac{z_2}{\sqrt{c_1(z_1^2+1)-z_2^2}}
+\left(\frac{c_1}3+\frac{z_2^2}{6(z_1^2+1)}\right)\sqrt{c_1(z_1^2+1)-z_2^2}\right)\!.
\end{gather*}
}
\begin{gather*}
h=\frac{z_1z_2+\sqrt{z_1^2+1-z_2^2}}{z_1^2+1},\\
\solutionRedEq
w=\frac{z_1z_2^3}{6(z_1^2+1)}
+\frac{z_2}2\arctan\frac{z_2}{\sqrt{z_1^2+1-z_2^2}}
+\frac16\left(2+\frac{z_2^2}{z_1^2+1}\right)\sqrt{z_1^2+1-z_2^2}.
\end{gather*}

The general solution of reduced equation~1.8$^a$ with $a\ne0$ can be represented in a parametric form.
Considering $\varphi/\omega$ as a parameter and denoting it by~$s$, we obtain
\[
\varphi=s\omega,\quad \omega=\frac{c_1{\rm e}^{a\arctan(s+a)}}{\sqrt{(s+a)^2+1}}.
\]
Up to the induced equivalence, we can set $c_1=1$.
The corresponding parametric solutions
of the equations~\eqref{eq:dNs1.3RhoNe1InviscidBurgersEq} and~\eqref{eq:dNs1.3RhoNe1ModifiedRedEq} are
\begin{gather*}
h=\frac{z_2}{z_1^2+1}(s+z_1+a),
\\
\solutionRedEq
w=\frac{z_2^3}{6(z_1^2+1)}
\left(z_1-\frac{3(a^2+1)s^2+2(a^2+1)^2-4as^3}{2a(9a^2+1)}+\frac{a^2-1}{2a}\right),
\end{gather*}
where
\[
\frac{{\rm e}^{-a\arctan z_1}}{\sqrt{z_1^2+1}}z_2=\frac{c_1{\rm e}^{a\arctan(s+a)}}{\sqrt{(s+a)^2+1}}.
\]

\begin{remark}
The point symmetry pseudogroup~$G_{\rm iB}$
of the inviscid Burgers equation~\eqref{eq:dNs1.3RhoNe1InviscidBurgersEq}
is much wider than its pseudosubgroup~$\check G_{1.3}$
consisting of the point symmetry transformations of~\eqref{eq:dNs1.3RhoNe1InviscidBurgersEq}
that are induced by the point symmetry transformations of~\eqref{eq:dNs1.3RhoNe1ModifiedRedEq}
via the substitution~$w_{22}=h$,
see the penultimate paragraph of Section~\ref{sec:dNRedEq1.3PointSymGroup}
and the last paragraph of Section~\ref{sec:dNRedEq1.3MIA}.
Any two solutions of~\eqref{eq:dNs1.3RhoNe1InviscidBurgersEq} are $G_{\rm iB}$-equivalent,
but generating solutions of~\eqref{eq:dNs1.3RhoNe1InviscidBurgersEq}
from a known explicit solution using point transformations from~$G_{\rm iB}$
does not in general lead to explicit solutions of~\eqref{eq:dNs1.3RhoNe1InviscidBurgersEq}.
The above solutions obtained by reductions~1.6$^2$ and~1.6$^{1/2}$
are related by the simple transformation
$\tilde h=-1/h$, $\tilde z_1=z_2$, $\tilde z_2=-z_1$ from~$G_{\rm iB}$.
\end{remark}

%\section{Exact solutions of dispersionless Nyzhnyk equation}

According to the optimized procedure of step-by-step reductions involving hidden symmetries
\cite[Section~B]{kova2023b},
to construct the corresponding exact solutions of the dispersionless Nyzhnyk equation~\eqref{eq:dN},
we extend the above solution families of the reduced equation~\eqref{eq:dNs1.3RhoNe1ModifiedRedEq}
by transformations from the pseudogroup~$G_{1.3}$ up to the equivalence
with respect to the induced symmetries of this equation
and substitute the extended families into ansatz~\eqref{eq:dNs1.3RhoNe1ModifiedAnsatz}.

\begin{theorem}\label{thm:dNRedEq1.3CorrInvSolutions}
Up to the $G$-equivalence, the set of exact solutions of the dispersionless Nyzhnyk equation~\eqref{eq:dN}
that can be constructed using the two-step Lie reductions,
where the first step is based on a subalgebra from the family $\{\mathfrak s_{1.3}^\rho\}$,
is exhausted by those of the form
\begin{gather}\label{eq:dNRedEq1.3CorrInvSolutions1}
\solution
u=\Delta(c_3z_1+c_4)\,\,w\!\left(\frac{c_1z_1+c_2}{c_3z_1+c_4},\frac{z_2+c_5z_1}{c_3z_1+c_4}\right)
+\frac{c_3z_2^3}{6(c_3z_1+c_4)}-\frac{c_4c_5 z_2^2}{2(c_3z_1+c_4)}-\frac{\rho_t}{6\rho}y^3,
\end{gather}
where $c_1$, \dots, $c_5$ are arbitrary constants with $\Delta=c_1c_4-c_2c_3=\pm1$, %and $c_3c_5=0$\todo
if $\rho$ is an arbitrary nonvanishing function of~$t$ with $\rho_t\ne0$,
and
\begin{gather}\label{eq:dNRedEq1.3CorrInvSolutions2}
\solution
u=w(z_1,z_2+c_5z_1)-\frac{c_5}2 z_2^2,\qquad\qquad
\solution
u=-z_1w(z_1^{-1},z_1^{-1}z_2+c_5)+\frac{z_2^3}{6z_1},
\end{gather}
where $c_5$ is an arbitrary constant,
if $\rho$ is an arbitrary constant with $\rho\ne0,1$.
In both cases,
$w\equiv0$ or $w(\cdot,\cdot)$ runs through the solutions of the equation~\eqref{eq:dNs1.3RhoNe1ModifiedRedEq}
listed in this section %Section~\ref{sec:dNs1.3LieInvSols}
and marked by~``$\circ$'',~and
\[
z_1=2\int\frac{\rho^3-1}{\rho^3}{\rm d}t,\quad
z_2=\frac y{\rho}-x.
\]
\end{theorem}

\begin{proof}
More specifically,
the inequivalent invariant solutions of the related intermediate reduced equation~\eqref{eq:dNs1.3RhoNe1ModifiedRedEq}
should be extended using a complete set of $\breve G_{1.3}^\rho$-inequivalent transformations from~$G_{1.3}$
under the left action of $\breve G_{1.3}^\rho$ on~$G_{1.3}$.
Recall that the pseudosubgroup~$\breve G_{1.3}^\rho$ of~$G_{1.3}$
consists of the point symmetry transformations of~\eqref{eq:dNs1.3RhoNe1ModifiedRedEq}
that are induced under the Lie reduction of~\eqref{eq:dN} with respect to the subalgebra~$\mathfrak s_{1.3}^\rho$,
see Section~\ref{sec:dNRedEq1.3OnInduction}.
For nonconstant values of the parameter function~$\rho$,
we assume $\breve G_{1.3}^\rho:=\breve G_{1.3}^{\rm gen}$,
thus neglecting the discrete extensions of~$\breve G_{1.3}^{\rm gen}$ for particular values of~$\rho$.
In other words, we extend the $G_{1.3}$-inequivalent solutions
of the equation~\eqref{eq:dNs1.3RhoNe1ModifiedRedEq}, which have been constructed in this section,
by acting the pseudosubgroup~$G_{1.3}$,
whose elements are of the form~\eqref{eq:dNs1.3RhoNe1ModifiedRedEqPointSymForm},
and then check which group parameters are inessential up to the $\breve G_{1.3}^\rho$-equivalence.
For $\rho_t\ne0$, these are $W^1(z_1)$, $W^0(z_1)$ and $c_6$, which can be set to zero.
For $\rho_t=0$, we can in addition set
either $c_1=c_4=1$, $c_2=c_3=0$ or $c_1=c_4=0$, $c_2=c_3=1$.
Finally, we substitute the obtained solutions into the ansatz~\eqref{eq:dNs1.3RhoNe1ModifiedAnsatz}.
\end{proof}

\begin{remark}
The $G$-inequivalent codimension-two Lie reductions of the dispersionless Nyzhnyk equation~\eqref{eq:dN}
from~\cite[Section~8.1]{vinn2024a} can be interpreted as two-step Lie reductions of this equation,
where the first steps involve one-dimensional subalgebras of~$\mathfrak g$
that are $G$-equivalent to subalgebras from the family~$\{\mathfrak s_{1.3}^\rho\}$.
Using $G_{1.3}$-inequivalent Lie reductions of the equation~\eqref{eq:dNs1.3RhoNe1ModifiedRedEq}
and extending the obtained exact solutions by hidden point symmetries of the equation~\eqref{eq:dN}
associated with its Lie reductions to~\eqref{eq:dNs1.3RhoNe1ModifiedRedEq}
in Theorem~\ref{thm:dNRedEq1.3CorrInvSolutions},
we construct much wider families of closed-form solutions of~\eqref{eq:dN}.
Any solution presented in~\cite[Section~8.1]{vinn2024a} is $G$-equivalent to
either a solution from the family~\eqref{eq:dNRedEq1.3CorrInvSolutions1} with $w=0$ and $c_5=0$
either a solution from the first family in~\eqref{eq:dNRedEq1.3CorrInvSolutions2},
where $c_5=0$ and $w$ is obtained by reductions~1.4 or~1.6.
\end{remark}

\section{Local symmetry-like objects}\label{sec:dNs1.3RhoNe1SymLikeObjects}

For a theoretical background on local symmetry-like objects of systems of differential equations,
which are generalized symmetries, cosymmetries, conservation-law characteristics and conservation laws,
see~\cite{olve1993A} as well as~\cite{blum2010A,blum1989A}.
We solve the equation~\eqref{eq:dNs1.3RhoNe1ModifiedRedEq} with respect to the derivative~$w_{222}$,
thus (locally) representing this equation in the Kovalevskaya form.
Therefore, we consider the derivatives of~$w$ with three or more differentiations with respect to~$z_2$
and the other derivatives of~$w$
as the principal and the parametric derivatives of the equation~\eqref{eq:dNs1.3RhoNe1ModifiedRedEq}, respectively.
In other words, the jet variables~$z_1$, $z_2$, $w_{k,l}$ with $k\in\mathbb N_0$ and $l\in\{0,1,2\}$,
where $w_{k,l}:=\p^{k+l}w/\p z_1^{\,k}\p z_2^{\,l}$,
constitute a coordinate system on the manifold~$\mathcal L$ defined by the equation~\eqref{eq:dNs1.3RhoNe1ModifiedRedEq}
and its differential consequences in the jet space $\mathrm J^\infty(\mathbb R^2_{z_1,z_2}\times\mathbb R_w)$.
(The preferable notation of derivatives of~$w$ in this section differs from that in the rest of the paper,
$w_{0,0}:=w$, $w_{1,0}:=w_1$, $w_{0,1}:=w_2$, $w_{2,0}:=w_{11}$, $w_{1,1}:=w_{12}$, $w_{0,2}:=w_{22}$, etc.)
The equation~\eqref{eq:dNs1.3RhoNe1ModifiedRedEq} possesses the two independent minimum-order $z_2$-integrals
\[
I^1:=w_{1,1}+\frac12(w_{0,2})^2, \quad %w_{12}+\frac12(w_{22})^2
I^2:=w_{2,0}-\frac13(w_{0,2})^3-z_2(w_{2,1}+w_{0,2}w_{1,2})=w_{2,0}-\frac13(w_{0,2})^3-z_2\mathrm D_1I^1,
%w_{11}-\frac13(w_{22})^3-z_2(w_{112}+w_{22}w_{122})=w_{11}-\frac13(w_{22})^3-z_2\mathrm D_1I^1,
\]
i.e., $\mathrm D_2I^1=\mathrm D_2I^2=0$ on solutions of~\eqref{eq:dNs1.3RhoNe1ModifiedRedEq}.
Here and in what follows, the symbols~$\mathrm D_1$ and~$\mathrm D_2$ denote
the operators of total derivatives with respect to the variables~$z_1$ and~$z_2$, respectively,
the index~$k$ runs~$\mathbb N_0$, $i,i'\in\{1,2\}$, and we assume summation with respect to repeated indices.
Then $\mathrm D_1^{\,k}I^1$ and~$\mathrm D_1^{\,k}I^2$ are $z_2$-integrals
of~\eqref{eq:dNs1.3RhoNe1ModifiedRedEq} as well.
Following the approach developed in~\cite{popo2024a,popo2025a,popo2025b},
we replace the above simple coordinate system on~$\mathcal L$ with the more sophisticated collection
\begin{gather}\label{eq:dNs1.3RhoNe1ModifiedRedEqCoords}
z_1,\ z_2,\ w_{0,0},\ w_{1,0},\ w_{0,1},\ w_{k,2},\ \zeta^{ik}:=\mathrm D_1^{\,k}I^i.
\end{gather}
%and supplement it to a coordinate system on~$\mathrm J^\infty(\mathbb R^2_{z_1,z_2}\times\mathbb R_w)$
%with the jet variables $w_{k,l}$, where $k\in\mathbb N_0$ and $l\in\mathbb N_0+3$.
We denote by~$f\{w\}$ a differential function of~$w$
that depends only on parametric variables and derivatives of~\eqref{eq:dNs1.3RhoNe1ModifiedRedEq}.
Up to the equivalence of integrals
(resp.\ of generalized symmetries, resp.\ of conserved currents, resp.\ of characteristics of conservation laws)
of~\eqref{eq:dNs1.3RhoNe1ModifiedRedEq},
we can consider the components or the characteristics of these objects to be such differential functions.
The restrictions~$\hat{\mathrm D}_1$ and~$\hat{\mathrm D}_2$
of the operators of total derivatives~$\mathrm D_1$ and~$\mathrm D_2$ take the form
\begin{gather*}
\begin{split}
\hat{\mathrm D}_1={}&\p_{z_1}+w_{1,0}\p_{w_{0,0}}
+\left(\zeta^{20}+\frac13(w_{0,2})^3+z_2\zeta^{11}\right)\p_{w_{1,0}}
+\left(\zeta^{10}-\frac12(w_{0,2})^2\right)\p_{w_{0,1}}\\
&+w_{k+1,2}\p_{w_{k,2}}+\zeta^{i,k+1}\p_{\zeta^{ik}},
\end{split}
\\[1ex]
\hat{\mathrm D}_2=\p_{z_2}+w_{0,1}\p_{w_{0,0}}
+\left(\zeta^{10}-\frac12(w_{0,2})^2\right)\p_{w_{1,0}}+w_{0,2}\p_{w_{0,1}}
-\hat{\mathrm D}_1^k\left(\frac{w_{1,2}}{w_{0,2}}\right)\p_{w_{k,2}}.
\end{gather*}
In particular, the condition for $z_2$-integrals~$f$ can be written as $\hat{\mathrm D}_2f=0$.
Note that $[\hat{\mathrm D}_1,\hat{\mathrm D}_2]=0$ since $[\mathrm D_1,\mathrm D_2]=0$.

In the coordinates~\eqref{eq:dNs1.3RhoNe1ModifiedRedEqCoords},
we define the partial orders $\ord_0$ and $\ord_i$
of a differential function~$f\{w\}$ with respect to the derivatives of
$w_{0,2}$ and~$I^i$, $i\in\{1,2\}$, respectively,
\begin{gather*}
\ord_0f:=\begin{cases}
\max\big\{k\mid f_{w_{k,2}}\ne0\big\}&\mbox{if this set is nonempty,}\\
-\infty&\mbox{otherwise,}
\end{cases}
\\
\ord_if:=\begin{cases}
\max\big\{k\mid f_{\zeta^{ik}}\ne0\big\}&\mbox{if this set is nonempty,}\\
-\infty&\mbox{otherwise.}
\end{cases}
\end{gather*}

Simultaneously with the coordinates~\eqref{eq:dNs1.3RhoNe1ModifiedRedEqCoords},
we use the even more sophisticated (local) coordinates on~$\mathcal L$,
\begin{gather}\label{eq:dNs1.3RhoNe1ModifiedRedEqCoords2}
w_{0,0},\ w_{1,0},\ w_{0,1},\ w_{0,2},\ w_{1,2},\
\theta^k:=\left(\frac{w_{0,2}}{w_{1,2}}\hat{\mathrm D}_2\right)^k(z_2-w_{0,2}z_1),\
\zeta^{ik}:=\mathrm D_1^{\,k}I^i.
\end{gather}
In the latter coordinates, the orders $\ord_if$ of a differential function~$f\{w\}$ are defined in the same way as above.
The expressions for~$\theta^k$ can be rewritten in the following form:
\begin{gather*}
\theta^0=z_2-w_{0,2}z_1,\quad
\theta^1=\frac{w_{0,2}}{w_{1,2}}+z_1,\quad
\theta^l=\left(\frac{w_{0,2}}{w_{1,2}}\hat{\mathrm D}_2\right)^{l-1}\frac{w_{0,2}}{w_{1,2}},\quad l=2,3,\dots.
%\\ \theta^2=\frac{w_{0,2}w_{2,2}-2(w_{1,2})^2}{(w_{1,2})^3}\quad\sim\quad w_{2,2}=(w_{1,2})^2\frac{w_{1,2}\theta^2+2}{w_{0,2}}
\end{gather*}
Since
$\big(\hat{\mathrm D}_1+w_{0,2}\hat{\mathrm D}_2\big)\theta^0
=\big(\hat{\mathrm D}_1+w_{0,2}\hat{\mathrm D}_2\big)(z_2-w_{0,2}z_1)=0$ and
\[
\left[\hat{\mathrm D}_1+w_{0,2}\hat{\mathrm D}_2,\frac{w_{0,2}}{w_{1,2}}\hat{\mathrm D}_2\right]=0,
\]
then $\big(\hat{\mathrm D}_1+w_{0,2}\hat{\mathrm D}_2\big)\theta^k=0$ for any $k$,
and hence, in the modified coordinates~\eqref{eq:dNs1.3RhoNe1ModifiedRedEqCoords2},
the operators~$\hat{\mathrm D}_1$ and~$\hat{\mathrm D}_2$ take the form
\begin{gather*}
\begin{split}
\hat{\mathrm D}_1={}&w_{1,0}\p_{w_{0,0}}
+\left(\zeta^{20}+\frac13(w_{0,2})^3+\left(\theta^0+w_{0,2}\theta^1-\frac{(w_{0,2})^2}{w_{1,2}}\right)\zeta^{11}\right)\p_{w_{1,0}}\\&{}
+\left(\zeta^{10}-\frac12(w_{0,2})^2\right)\p_{w_{0,1}}
+w_{1,2}\p_{w_{0,2}}+\frac{(w_{1,2})^2}{w_{0,2}}(w_{1,2}\theta^2+2)\p_{w_{1,2}}\\&{}
-w_{1,2}\theta^{k+1}\p_{\theta^k}+\zeta^{i,k+1}\p_{\zeta^{ik}},
\end{split}
\\[1ex]
\begin{split}
\hat{\mathrm D}_2
={}&w_{0,1}\p_{w_{0,0}}+\left(\zeta^{10}-\frac12(w_{0,2})^2\right)\p_{w_{1,0}}
+w_{0,2}\p_{w_{0,1}}-\frac{w_{1,2}}{w_{0,2}}\p_{w_{0,2}}\\&
-\left(\frac{w_{1,2}}{w_{0,2}}\right)^2(w_{1,2}\theta^2+1)\p_{w_{1,2}}
+\frac{w_{1,2}}{w_{0,2}}\theta^{k+1}\p_{\theta^k},
\end{split}
\end{gather*}
as well as
$z_1=\theta^1-w_{0,2}/w_{1,2}$ and $z_2=\theta^0+w_{0,2}\theta^1-(w_{0,2})^2/w_{1,2}$.

In what follows, all functions that arise in the course of integrating
are sufficiently smooth functions of their arguments,
which are indicated explicitly when the corresponding function first appears.
The symbol ``$*$'' in superscripts indicates indices running through finite sets of nonnegative integers.
For example, the dependence of a differential function~$f$ on~$\zeta^{1*}$ means
the dependence of~$f$ on $(\zeta^{10},\dots,\zeta^{1k_1})$ with $k_1:=\ord_1f$.

\subsection{Integrals}\label{sec:dNs1.3RhoNe1Integrals}

Basic symmetry-like objects associated with the equation~\eqref{eq:dNs1.3RhoNe1ModifiedRedEq}
are its integrals.

\begin{theorem}\label{thm:dNs1.3RhoNe1Integrals}
A differential function $\alpha\{w\}$ is a $z_2$-integral of the equation~\eqref{eq:dNs1.3RhoNe1ModifiedRedEq},
${\hat{\mathrm D}_2\alpha\{w\}=0}$,
if and only if it is a sufficiently smooth function of $z_1$, $I^1$ and~$I^2$
and a finite number of total derivatives of $I^1$ and~$I^2$ with respect to~$z_1$,
\[
\alpha=\alpha\big(z_1,(\zeta^{ik})_{k=0,\dots,r,\,i=1,2}\big)=
\alpha(z_1,I^1,\mathrm D_1I^1,\dots,\mathrm D_1^{\,r}I^1,I^2,\mathrm D_1I^2,\dots,\mathrm D_1^{\,r}I^2)
\quad\mbox{with}\quad r\in\mathbb N.
\]
\end{theorem}

\begin{proof}
The ``if''-part can be checked by the direct application of the chain rule.

We prove the ``only if''-part by contradiction.
Let $k_0:=\ord_0\alpha>-\infty$ for a $z_2$-integral~$\alpha\{w\}$ of~\eqref{eq:dNs1.3RhoNe1ModifiedRedEq}.
Then the condition \smash{$\p_{w_{k_0+1,2}}\hat{\mathrm D}_2\alpha=0$} implies $\alpha_{w_{k_0,2}}=0$,
which contradicts the supposition.
Therefore, $\alpha_{w_{k,2}}=0$ for any~$k\in\mathbb N_0$,
and thus successively
\smash{$\p_{w_{0,2}}^{\,\,2}\hat{\mathrm D}_2\alpha=\alpha_{w_{1,0}}=0$},
\smash{$\p_{w_{0,2}}\hat{\mathrm D}_2\alpha=\alpha_{w_{0,1}}=0$},
\smash{$\p_{w_{0,1}}\hat{\mathrm D}_2\alpha=\alpha_{w_{0,0}}=0$}.
As a result, \smash{$\hat{\mathrm D}_2\alpha=\alpha_{z_2}=0$} as well.
\end{proof}

\subsection{Auxiliary results}\label{sec:dNs1.3RhoNe1AuxiliaryResults}

\begin{lemma}\label{lem:dNs1.3RhoNe1AuxiliaryLemma1}
A differential function $\varrho\{w\}$ satisfies the equation
\begin{gather}\label{eq:dNs1.3RhoNe1ModifiedRedEqInTotalDers}
\big(\hat{\mathrm D}_1+w_{0,2}\hat{\mathrm D}_2\big)\varrho=0
\end{gather}
if and only if $\varrho$ is a function at most of~$w_{0,2}$ and a finite number of~$\theta^k$.
\end{lemma}

\begin{proof}
Denote $k_i:=\ord_i\varrho$.
Suppose that $k_1\geqslant1$.
Then the differentiation of the equation~\eqref{eq:dNs1.3RhoNe1ModifiedRedEqInTotalDers}
with respect to~$\zeta^{1,k_1+1}$ implies the condition $\varrho_{\zeta^{1k_1}}=0$
contradicting the supposition ${k_1>1}$.
Analogously, when supposing $k_2\geqslant0$, we derive the contradicting condition $\varrho_{\zeta^{2k_2}}=0$
by differentiating the equation~\eqref{eq:dNs1.3RhoNe1ModifiedRedEqInTotalDers}
with respect to~$\zeta^{2,k_2+1}$.
Therefore, $k_1\leqslant0$ and $k_2=-\infty$.
Then, successively considering the derivatives of the equation~\eqref{eq:dNs1.3RhoNe1ModifiedRedEqInTotalDers}
with respect to~$\zeta^{20}$, $\zeta^{11}$, $\zeta^{10}$ and~$w_{1,0}$,
we derive $\varrho_{w_{1,0}}=0$, $\varrho_{\zeta^{10}}=0$, $\varrho_{w_{0,1}}=0$ and $\varrho_{w_{0,0}}=0$, respectively.
Therefore, $\varrho$ is a function at most of~$w_{0,2}$ and a finite number of~$\theta^k$.
It is obvious that any such function~$\varrho$ is a solution of the equation~\eqref{eq:dNs1.3RhoNe1ModifiedRedEqInTotalDers}.
\end{proof}

\begin{lemma}\label{lem:dNs1.3RhoNe1AuxiliaryLemma2b}
Given $z_2$-integrals~$\kappa^0$ and~$\kappa^1$ of the equation~\eqref{eq:dNs1.3RhoNe1ModifiedRedEq}, the equation
\begin{gather}\label{eq:dNs1.3RhoNe1ModifiedRedEqInTotalDers2'}
\big(\hat{\mathrm D}_1+w_{0,2}\hat{\mathrm D}_2\big)f=\kappa^0+z_2\kappa^1
\end{gather}
for a differential function $f\{w\}$ has a solution if and only if
\begin{gather*}
\kappa^0=\hat{\mathrm D}_1\alpha+a_0\zeta^{10}+(b_{01}z_1+b_{00})\zeta^{20},\\
\kappa^1=-\hat{\mathrm D}_1^2\gamma+a_1(z_1\zeta^{11}+2\zeta^{10})
+a_2\zeta^{11}+(3b_{11}z_1+b_{10}+b_{21})\zeta^{20}
+(b_{11}z_1^{\,2}+b_{21}z_1+b_{20})\zeta^{21}\!
\end{gather*}
for some $z_2$-integrals~$\alpha$ and~$\gamma$ of~\eqref{eq:dNs1.3RhoNe1ModifiedRedEq}
and some constants~$a_j$ and~$b_{jj'}$, $j=0,1,2$, $j'=0,1$.
Then the general solution of~\eqref{eq:dNs1.3RhoNe1ModifiedRedEqInTotalDers2'}~is
\begin{gather*}
f=\check f+\alpha-z_2\hat{\mathrm D}_1\gamma+w_{0,2}\gamma
+(a_0+a_1\theta^0-a_2w_{0,2})(w_{0,1}-\tfrac12z_1(w_{0,2})^2)
\\\quad\quad{}
+a_1z_1z_2\zeta^{10}+a_2z_2\zeta^{10}
+b_{00}(w_{1,0}-z_2\zeta^{10}+\tfrac16z_1(w_{0,2})^3)
\\\quad\quad{}
+b_{01}\big(z_1(w_{1,0}-z_2\zeta^{10})+z_2w_{0,1}-w_{0,0}
+\tfrac1{12}z_1^2(w_{0,2})^3-\tfrac14z_2^2w_{0,2}\big)
\\\quad\quad{}
+b_{10}(z_2(w_{1,0}-z_2\zeta^{10})+w_{0,2}(z_2w_{0,1}-w_{0,0})
-\tfrac16z_2^2(w_{0,2})^2)
\\\quad\quad{}
+b_{11}\big(z_1\theta^0(w_{1,0}-z_2\zeta^{10})+\theta^0(z_2w_{0,1}-w_{0,0})
+\tfrac1{12}z_1^2(w_{0,2})^3\theta^0-\tfrac14z_2^2w_{0,2}\theta^0
+z_1^{\,2}z_2\zeta^{20}\big)
\\\quad\quad{}
+b_{20}(z_2\zeta^{20}-w_{0,2}(w_{1,0}-z_2\zeta^{10})-\tfrac16z_1(w_{0,2})^4)
\\\quad\quad{}
+b_{21}\big(z_1z_2\zeta^{20}-w_{0,2}(z_1(w_{1,0}-z_2\zeta^{10})
+(z_2w_{0,1}-w_{0,0})
+\tfrac1{12}z_1^2(w_{0,2})^3-\tfrac14z_2^2w_{0,2})\big),
\end{gather*}
where~$\check f$ is an arbitrary function at most of~$w_{0,2}$ and a finite number of~$\theta^k$.
\end{lemma}

\begin{proof}
Let us prove the ``only if''-part. Denote $k_i:=\max(\ord_i\kappa^0,\ord_i\kappa^1)$, $i=1,2$.
Then the equation~\eqref{eq:dNs1.3RhoNe1ModifiedRedEqInTotalDers2'} implies that
$\ord_1f=k_1-1$ if $k_1>1$,
$\ord_1f\leqslant0$ if $k_1\leqslant1$,
$\ord_2f=k_2-1$ if $k_2>0$ and
$\ord_2f=-\infty$ if $k_2\leqslant0$.

Let $k_1>1$ and $k_2>0$.
Denote by $\Delta_{ik}$ the equation obtained by the differentiation of~\eqref{eq:dNs1.3RhoNe1ModifiedRedEqInTotalDers2'}
with respect to $\zeta^{ik}$, $i\in\{1,2\}$, $k\in\mathbb N_0$,
\begin{gather*}%\label{eq:dNs1.3RhoNe1ModifiedRedEqInTotalDers2'}
\Delta_{ik}\colon\ f_{\zeta^{i,k-1}}+\big(\hat{\mathrm D}_1+w_{0,2}\hat{\mathrm D}_2\big)f_{\zeta^{ik}}
=\kappa^0_{\zeta^{ik}}+z_2\kappa^1_{\zeta^{ik}},\quad k>1\mbox{ \ if \ }i=1 \mbox{ \ and \ } k>0\mbox{ \ if \ }i=2,
\\
\Delta_{11}\colon\ f_{\zeta^{10}}+\big(\hat{\mathrm D}_1+w_{0,2}\hat{\mathrm D}_2\big)f_{\zeta^{11}}+z_2f_{w_{1,0}}
=\kappa^0_{\zeta^{11}}+z_2\kappa^1_{\zeta^{11}},
\\
\Delta_{10}\colon\ \big(\hat{\mathrm D}_1+w_{0,2}\hat{\mathrm D}_2\big)f_{\zeta^{10}}
+w_{0,2}f_{w_{1,0}}+f_{w_{0,1}}
=\kappa^0_{\zeta^{10}}+z_2\kappa^1_{\zeta^{10}},
\\
\Delta_{20}\colon\ \big(\hat{\mathrm D}_1+w_{0,2}\hat{\mathrm D}_2\big)f_{\zeta^{20}}+f_{w_{1,0}}
=\kappa^0_{\zeta^{20}}+z_2\kappa^1_{\zeta^{20}}.
\end{gather*}

For each $i\in\{1,2\}$,
we prove by induction downward with respect to~$k$ starting from $k=k_i-1$ to $k=1$ for $i=1$ and to $k=0$ for $i=2$
that
\begin{gather}\label{eq:dNs1.3RhoNe1ModifiedRedEqInTotalDersProofA}
\begin{split}&
f_{\zeta^{ik}}=\alpha^{ik}+z_2\beta^{ik}+w_{0,2}\gamma^{ik},\\&
k=1,\dots,k_1-1\mbox{ \ if \ }i=1\mbox{ \ and \ }k=0,\dots,k_2-1\mbox{ \ if \ }i=2,
\end{split}
\end{gather}
where $\alpha^{ik}$, $\beta^{ik}$ and~$\gamma^{ik}$ are $z_2$-integrals of~\eqref{eq:dNs1.3RhoNe1ModifiedRedEq}
whose $\ord_{i'}$ is less than~$k_{i'}$.
Indeed, the equation $\Delta_{ik_i}$ gives the base case with
$\alpha^{i,k_i-1}=\kappa^0_{\zeta^{ik_i}}$,
$\beta^{i,k_i-1}=\kappa^1_{\zeta^{ik_i}}$ and
$\gamma^{i,k_i-1}=0$.
For the induction step, suppose that the claim to be proved holds true for $k=l$.
Then the equation $\Delta_{il}$ implies
\[
f_{\zeta^{i,l-1}}
=\kappa^0_{\zeta^{il}}-\hat{\mathrm D}_1\alpha^{il}+z_2(\kappa^1_{\zeta^{il}}-\hat{\mathrm D}_1\beta^{il})
-w_{0,2}(\hat{\mathrm D}_1\gamma^{il}+\beta^{il}),
\]
i.e.,
$\alpha^{i,l-1}=\kappa^0_{\zeta^{il}}-\hat{\mathrm D}_1\alpha^{il}$,
$\beta^{i,l-1}=\kappa^1_{\zeta^{il}}-\hat{\mathrm D}_1\beta^{il}$ and
$\gamma^{i,l-1}=-\hat{\mathrm D}_1\gamma^{il}-\beta^{il}$.

The system~\eqref{eq:dNs1.3RhoNe1ModifiedRedEqInTotalDersProofA} implies that
\begin{gather}\label{eq:dNs1.3RhoNe1ModifiedRedEqInTotalDersProofB}
f=\alpha+z_2\beta+w_{0,2}\gamma+\bar f,
\end{gather}
where $\alpha$, $\beta$ and~$\gamma$ are $z_2$-integrals of~\eqref{eq:dNs1.3RhoNe1ModifiedRedEq}
whose $\ord_i$ is less than~$k_i$,
and $\bar f$ is a function of at most $(z_1,z_2,w_{0,0},w_{1,0},w_{0,1},w_{*,2},\zeta^{10})$.
Recall that $(w_{*,2}):=(w_{k,2}, k=0,\dots,\ord_0\bar f)$
according to the explanation in the introductive part of Section~\ref{sec:dNs1.3RhoNe1SymLikeObjects}.

Acting in a similar way in the case $k_1\leqslant1$ and $k_2>0$ as in the above case,
we derive the representation~\eqref{eq:dNs1.3RhoNe1ModifiedRedEqInTotalDersProofB} for~$f$
with $z_2$-integrals~$\alpha$, $\beta$ and~$\gamma$ of~\eqref{eq:dNs1.3RhoNe1ModifiedRedEq}
whose $\ord_1$ is less than or equal to~$0$ and $\ord_2$ is less than~$k_2$.
The treatment of the case $k_1>1$ and $k_2\leqslant0$ is analogous.
In the case $k_1\leqslant1$ and $k_2\leqslant0$,
we immediately have the representation~\eqref{eq:dNs1.3RhoNe1ModifiedRedEqInTotalDersProofB} for~$f$
with zero $\alpha$, $\beta$ and~$\gamma$.

We substitute the representation~\eqref{eq:dNs1.3RhoNe1ModifiedRedEqInTotalDersProofB}
into~\eqref{eq:dNs1.3RhoNe1ModifiedRedEqInTotalDers2'},
\[\hat{\mathrm D}_1\alpha+z_2\hat{\mathrm D}_1\beta+w_{0,2}(\beta+\hat{\mathrm D}_1\gamma)
+\big(\hat{\mathrm D}_1+w_{0,2}\hat{\mathrm D}_2\big)\bar f
=\kappa^0+z_2\kappa^1,
\]
and consider the obtained equation for three fixed values
$(z_2^\iota,w_{0,0}^\iota,w_{1,0}^\iota,w_{0,1}^\iota,w_{*,2}^\iota)$, $\iota=1,2,3$,
of the variable tuple $(z_2,w_{0,0},w_{1,0},w_{0,1},w_{*,2})$
such that the tuples $(1,z_2^\iota,w_{0,2}^\iota)$ are linearly independent.
This gives the system of linear algebraic equations
\[
\hat{\mathrm D}_1\alpha-\kappa^0+z_2^\iota(\hat{\mathrm D}_1\beta-\kappa^1)+w_{0,2}^\iota(\beta+\hat{\mathrm D}_1\gamma)
=\chi^{\iota1}\zeta^{20}+\chi^{\iota2}\zeta^{11}+\chi^{\iota0}
\]
with respect to $\hat{\mathrm D}_1\alpha-\kappa^0$, $\hat{\mathrm D}_1\beta-\kappa^1$ and~$\beta+\hat{\mathrm D}_1\gamma$
with nonzero determinant of its coefficient matrix,
where $\chi^{\iota0}$, $\chi^{\iota1}$ and~$\chi^{\iota2}$ are functions of~$(z_1,\zeta^{10})$ such that
$\chi^{\iota1}\zeta^{20}+\chi^{\iota2}\zeta^{11}+\chi^{\iota0}$ is the value of ${-\big(\hat{\mathrm D}_1+w_{0,2}\hat{\mathrm D}_2\big)\bar f}$
at $(z_2^\iota,w_{0,0}^\iota,w_{1,0}^\iota,w_{0,1}^\iota,w_{*,2}^\iota)$.
The solution of the system takes the form
\[
\beta+\hat{\mathrm D}_1\gamma=\hat\varphi\zeta^{20}+\tilde \varphi\zeta^{11}+\check\varphi,\quad
\kappa^0=\hat{\mathrm D}_1\alpha+\hat\kappa^0\zeta^{20}+\tilde \kappa^0\zeta^{11}+\check\kappa^0,\quad
\kappa^1=\hat{\mathrm D}_1\beta+\hat\kappa^1\zeta^{20}+\tilde \kappa^1\zeta^{11}+\check\kappa^1,
\]
where $\hat\kappa^0$, $\check\kappa^0$, $\tilde \kappa^0$, $\hat\kappa^1$, $\check\kappa^1$, $\tilde \kappa^1$, $\hat\varphi$, $\check\varphi$ and~$\tilde\varphi$ are functions of~$(z_1,\zeta^{10})$,
and hence
\begin{gather}\label{eq:dNs1.3RhoNe1ModifiedRedEqInTotalDersProofC}
\begin{split}
\big(\hat{\mathrm D}_1+w_{0,2}\hat{\mathrm D}_2\big)\bar f
={}&\hat\kappa^0\zeta^{20}+\tilde \kappa^0\zeta^{11}+\check\kappa^0+z_2(\hat\kappa^1\zeta^{20}+\tilde \kappa^1\zeta^{11}+\check\kappa^1)\\
&-w_{0,2}(\hat\varphi\zeta^{20}+\tilde\varphi\zeta^{11}+\check\varphi).
\end{split}
\end{gather}

Acting on the equation~\eqref{eq:dNs1.3RhoNe1ModifiedRedEqInTotalDersProofC}
by the operators~$\p_{\zeta^{20}}$ and $\p_{\zeta^{11}}-z_2\p_{\zeta^{20}}$,
we derive its differential consequences
\[
\bar f_{w_{1,0}}=\hat\kappa^0+z_2\hat\kappa^1-w_{0,2}\hat\varphi, \quad
\bar f_{\zeta^{10}}=-z_2(\hat\kappa^0+z_2\hat\kappa^1-w_{0,2}\hat\varphi)+\tilde \kappa^0+z_2\tilde \kappa^1-w_{0,2}\tilde\varphi.
\]
We make the cross differentiation of these differential consequences with respect to~$w_{1,0}$ and~$\zeta^{10}$
and split the obtained equation \smash{$\hat\kappa^0_{\zeta^{10}}+z_2\hat\kappa^1_{\zeta^{10}}-w_{0,2}\hat\varphi_{\zeta^{10}}=0$}
with respect to~$z_2$ and~$w_{0,2}$,
which gives the equations \smash{$\hat\kappa^0_{\zeta^{10}}=\hat\kappa^1_{\zeta^{10}}=\hat\varphi_{\zeta^{10}}=0$},
i.e., the coefficients~$\hat\kappa^0$, $\hat\kappa^1$ and $\hat\varphi$ depend at most on~$z_1$.
Separately differentiating the equation~\eqref{eq:dNs1.3RhoNe1ModifiedRedEqInTotalDersProofC}
with respect to~$\zeta^{10}$, $w_{1,0}$, $w_{0,0}$ and~$w_{0,1}$
and taking into account the previously derived differential consequences,
we obtain the following equations:
\begin{gather*}
\bar f_{w_{0,1}}
=\check\kappa^0_{\zeta^{10}}+z_2\check\kappa^1_{\zeta^{10}}-w_{0,2}\check\varphi_{\zeta^{10}}
+z_2(\hat\kappa^0_{z_1}+z_2\hat\kappa^1_{z_1}-w_{0,2}\hat\varphi_{z_1}+w_{0,2}\hat\kappa^1)
\\\qquad\quad{}
-(\tilde \kappa^0_{z_1}+z_2\tilde \kappa^1_{z_1}-w_{0,2}\tilde\varphi_{z_1}+w_{0,2}\tilde \kappa^1),
\\
\bar f_{w_{0,0}}=-(\hat\kappa^0_{z_1}+z_2\hat\kappa^1_{z_1}-w_{0,2}\hat\varphi_{z_1}+w_{0,2}\hat\kappa^1),
\\
\hat\kappa^0_{z_1z_1}+z_2\hat\kappa^1_{z_1z_1}-w_{0,2}\hat\varphi_{z_1z_1}+2w_{0,2}\hat\kappa^1_{z_1}=0,
\\
\check\kappa^0_{z_1\zeta^{10}}+z_2\check\kappa^1_{z_1\zeta^{10}}-w_{0,2}\check\varphi_{z_1\zeta^{10}}
+w_{0,2}\check\kappa^1_{\zeta^{10}}-(\tilde \kappa^0_{z_1z_1}+z_2\tilde \kappa^1_{z_1z_1}-w_{0,2}\tilde\varphi_{z_1z_1}+2w_{0,2}\tilde \kappa^1_{z_1})=0.
\end{gather*}
The cross differentiation of expressions for~$\bar f_{\zeta^{10}}$ and~$\bar f_{w_{0,1}}$
in addition gives the following equations:
\smash{$\check\kappa^0_{\zeta^{10}\zeta^{10}}=\tilde \kappa^0_{z_1\zeta^{10}}$},
\smash{$\check\kappa^1_{\zeta^{10}\zeta^{10}}=\tilde \kappa^1_{z_1\zeta^{10}}$} and
\smash{$\check\varphi_{\zeta^{10}\zeta^{10}}=
\tilde\varphi_{z_1\zeta^{10}}-\tilde \kappa^1_{\zeta^{10}}$}.
The last two equations split with respect to~$z_2$ and~$w_{0,2}$ into the equations
${\hat\kappa^0_{z_1z_1}=\hat\kappa^1_{z_1z_1}=0}$,
${\hat\varphi_{z_1z_1}=2\hat\kappa^1_{z_1}}$,
\smash{$\check\kappa^0_{z_1\zeta^{10}}=\tilde \kappa^0_{z_1z_1}$},
\smash{$\check\kappa^1_{z_1\zeta^{10}}=\tilde \kappa^1_{z_1z_1}$} and
\smash{$\check\varphi_{z_1\zeta^{10}}-\check\kappa^1_{\zeta^{10}}=\tilde\varphi_{z_1z_1}-2\tilde \kappa^1_{z_1}$}.
Therefore,
\begin{gather*}
\check\kappa^0=\Phi_{z_1},\quad
\check\kappa^1=\Psi_{z_1},\quad
\check\varphi=\Theta_{z_1}-\Psi,
\\
\tilde \kappa^0=\Phi_{\zeta^{10}}-a_0z_1,\quad
\tilde \kappa^1=\Psi_{\zeta^{10}}-a_1z_1,\quad
\tilde\varphi=\Theta_{\zeta^{10}}-a_1z_1^2-a_2z_1,
\\
\hat\kappa^0=b_{01}z_1+b_{00},\quad
\hat\kappa^1=b_{11}z_1+b_{10},\quad
\hat\varphi=b_{11}z_1^{\,2}+b_{21}z_1+b_{20},
\end{gather*}
where~$a_j$ and~$b_{jj'}$, $j=0,1,2$, $j'=0,1$, are arbitrary constants and
$\Phi$, $\Psi$ and~$\Theta$ are arbitrary functions of~$(z_1,\zeta^{10})$.
As a result, the function~$\bar f$ takes the form
\begin{gather*}
\bar f
=\hat f(z_1,z_2,w_{*,2})
+\Phi+z_2\Psi-w_{0,2}\Theta
+(a_0+a_1\theta^0-a_2w_{0,2})(w_{0,1}-z_1\zeta^{10})
\\\quad\quad{}
+\big(b_{01}z_1+b_{00}+b_{11}z_1\theta^0+b_{10}z_2-(b_{21}z_1+b_{20})w_{0,2}\big)(w_{1,0}-z_2\zeta^{10})
\\\quad\quad{}
+\big(b_{01}+b_{11}\theta^0+(b_{10}-b_{21})w_{0,2}\big)(z_2w_{0,1}-w_{0,0})
\end{gather*}
for some function~$\hat f$ of the indicated arguments.
We substitute this representation for~$\bar f$ into the equation~\eqref{eq:dNs1.3RhoNe1ModifiedRedEqInTotalDersProofC}
and obtain the reduced equation for~$\hat f$,
\begin{gather*}
\big(\hat{\mathrm D}_1+w_{0,2}\hat{\mathrm D}_2\big)\hat f(z_1,z_2,w_{*,2})
=\tfrac16(w_{0,2})^3\big(b_{01}z_1+b_{00}+b_{11}z_1\theta^0
+b_{10}z_2-(b_{21}z_1+b_{20})w_{0,2}\big)
\\\quad{}
-\tfrac12z_2(w_{0,2})^2\big(b_{01}+b_{11}\theta^0+(b_{10}-b_{21})w_{0,2}\big)
-\tfrac12(w_{0,2})^2(a_0+a_1\theta^0-a_2w_{0,2}).
\end{gather*}
We solve this equation with respect to~$\hat f$.
For convenient representation of the solution,
we consider the antiderivative of the fourth coefficient, which depends on~$z_1$,
in the right-hand side of the equation as a parameter function instead of the function involved in this coefficient,
\begin{gather*}
\hat f
=\check f
-\tfrac12 z_1(w_{0,2})^2(a_0+a_1\theta^0-a_2w_{0,2})
\\\quad\quad{}
+\tfrac16b_{00}z_1(w_{0,2})^3
+b_{01}\big(\tfrac1{12}z_1^2(w_{0,2})^3-\tfrac14z_2^2w_{0,2}\big)
\\\quad\quad{}
-\tfrac16b_{10}z_2^2(w_{0,2})^2
+b_{11}\big(\tfrac1{12}z_1^2(w_{0,2})^3\theta^0-\tfrac14z_2^2w_{0,2}\theta^0\big)
\\\quad\quad{}
-\tfrac16b_{20}z_1(w_{0,2})^4
-b_{21}\big(\tfrac1{12}z_1^2(w_{0,2})^4-\tfrac14z_2^2(w_{0,2})^2\big),
\end{gather*}
where $\check f$ is an arbitrary solution of the associated homogeneous equation
$\big(\hat{\mathrm D}_1+w_{0,2}\hat{\mathrm D}_2\big)\check f=0$.

Successively carrying out all the above substitutions
and re-denoting
$\alpha+\Phi-a_0z_1\zeta^{10}$ by $\alpha$ and
$\gamma-\Theta+a_1z_1^2\zeta^{10}+a_2z_1\zeta^{10}$ by $\gamma$
result in the expressions for~$\kappa^0$, $\kappa^1$ and~$f$ in lemma's statement.
\end{proof}

\begin{lemma}\label{lem:dNs1.3RhoNe1AuxiliaryLemma3}
Given a $z_2$-integral~$\alpha$ of the equation~\eqref{eq:dNs1.3RhoNe1ModifiedRedEq},
a solution~$\varrho$ of the equation~\eqref{eq:dNs1.3RhoNe1ModifiedRedEqInTotalDers}
and some constants~$a_0$, $a_1$ and~$a_2$, the equation
\begin{gather}\label{eq:dNs1.3RhoNe1ModifiedRedEqInTotalDers3}
\begin{split}
\hat{\mathrm D}_2 f={}&
\alpha+\varrho+(a_2z_1+a_1)(w_{1,0}-z_2\zeta^{10})+(a_2z_2+a_0)w_{0,1}-a_2w_{0,0}\\&{}
+\frac{z_1}{12}(a_2z_1+2a_1)(w_{0,2})^3-\frac{z_2}4(a_2z_2+2a_0)w_{0,2}
\end{split}
\end{gather}
has a solution if and only if
\begin{gather*}
\varrho=-R\frac{R\hat\varrho}{\theta^2}
+a_1(w_{0,2})^2\bigg(
-\frac23w_{0,2}\frac{(\theta^1)^2}{(\theta^2)^2}\theta^3
-2\frac{(\theta^1)^2}{\theta^2}
+2w_{0,2}\theta^1
-\frac12\theta^0
\bigg)
\\ \qquad{}
+a_2w_{0,2}\bigg(
2w_{0,2}\frac{(\theta^1)^2}{(\theta^2)^2}\theta^0\theta^3
+4\frac{(\theta^1)^2}{\theta^2}\theta^0
-2w_{0,2}\frac{(\theta^1)^3}{\theta^2}
+\frac54(\theta^0)^2-6w_{0,2}\theta^0\theta^1
\bigg)
\\ \qquad{}
+b_{00}w_{0,2}-b_{01}\theta^0
+b_{10}w_{0,2}\theta^0
-b_{11}\frac{(\theta^0)^2}2
+b_{20}\frac{(w_{0,2})^2}2
\\ \qquad{}
+\tilde b_{21}w_{0,2}\theta^1\bigg(
\frac{w_{0,2}}2\frac{\theta^1\theta^3}{(\theta^2)^2}
+\frac{\theta^1}{\theta^2}
-\frac32w_{0,2}\bigg),
\end{gather*}
where
$R:=-\p_{w_{0,2}}+\theta^{k+1}\p_{\theta^k}$,
the function~$\hat\varrho$ depends at most on~$w_{0,2}$ and a finite number of~$\theta^k$,
$b_{00}$, $b_{01}$, $b_{10}$, $b_{20}$, $b_{11}$ and~$\tilde b_{21}$ are arbitrary constants.
Then the general solution of~\eqref{eq:dNs1.3RhoNe1ModifiedRedEqInTotalDers3}~is
\begin{gather*}
f=\beta+z_2\alpha-\frac{w_{0,2}}{w_{1,2}\theta^2}R\hat\varrho+\hat\varrho
\\ \qquad{}
-\frac{a_0}2(z_2w_{0,1}-3w_{0,0})
+a_1\bigg(
\frac23\frac{(w_{0,2})^4}{w_{1,2}}\frac{(\theta^1)^2}{\theta^2}+\frac16z_1^2(w_{0,2})^4
+z_2w_{1,0}-z_2^2\zeta^{10}
\bigg)
\\ \qquad{}
-a_2\bigg(
2\frac{(w_{0,2})^3}{w_{1,2}}\frac{(\theta^1)^2}{\theta^2}\theta^0
+\frac23z_1^2(w_{0,2})^3\theta^0+\frac16z_1^3(w_{0,2})^4
\\ \qquad\qquad\quad{}
+z_2w_{0,0}-z_2^2w_{0,1}-z_1z_2w_{1,0}+z_1z_2^2\zeta^{10}
\bigg)
\\ \qquad{}
+b_{00}w_{0,1}
+\frac{b_{01}}2(2z_1w_{0,1}-z_2^2)
+b_{10}(2z_1w_{1,0}+z_2w_{0,1}-w_{0,0}-2z_1z_2\zeta^{10})
\\ \qquad{}
+b_{11}\bigg(z_1^2w_{1,0}+z_1z_2w_{0,1}-z_1w_{0,0}-\frac16z_2^3-z_1^2z_2\zeta^{10}\bigg),
\\ \qquad{}
-b_{20}(w_{1,0}-z_2\zeta^{10})
-\tilde b_{21}\bigg(
\frac{(w_{0,2})^3}{2w_{1,2}}\frac{(\theta^1)^2}{\theta^2}+\frac16z_1^2(w_{0,2})^3
+z_1w_{1,0}-z_2z_1\zeta^{10}
\bigg)
\end{gather*}
where $\beta$ is an arbitrary $z_2$-integral of~\eqref{eq:dNs1.3RhoNe1ModifiedRedEq}.
\end{lemma}

\begin{proof}
Replacing~$f$ by $f-z_2\alpha$, we can set $\alpha=0$.
We represent the equation~\eqref{eq:dNs1.3RhoNe1ModifiedRedEqInTotalDers3}
in the modified coordinates~\eqref{eq:dNs1.3RhoNe1ModifiedRedEqCoords2}, substituting
\begin{gather*}
z_1=\theta^1-\frac{w_{0,2}}{w_{1,2}},\quad
z_2=\theta^0+w_{0,2}z_1=\theta^0+w_{0,2}\theta^1-\frac{(w_{0,2})^2}{w_{1,2}}.
\end{gather*}

Suppose that $r:=\ord_0f\geqslant2$.
Then $\ord_0\varrho=r+1$.
Differentiating the equation~\eqref{eq:dNs1.3RhoNe1ModifiedRedEqInTotalDers3}
with respect to one of the coordinates~$\theta^{k+1}$ with $k\geqslant2$ or $\theta^2$,
we respectively derive the equations
\begin{gather*}
f_{\theta^k}=\frac{w_{0,2}}{w_{1,2}}\varrho_{\theta^{k+1}}-\frac{w_{0,2}}{w_{1,2}}\hat{\mathrm D}_2 f_{\theta^{k+1}},\quad k\geqslant2,
\quad
f_{\theta^1}-\frac{(w_{1,2})^2}{w_{0,2}}f_{w_{1,2}}
=\frac{w_{0,2}}{w_{1,2}}\varrho_{\theta^2}-\frac{w_{0,2}}{w_{1,2}}\hat{\mathrm D}_2 f_{\theta^2}.
\end{gather*}
We use these equations, going from $k=r$ downward to derive
\begin{gather*}
f_{\theta^k}=\frac{w_{0,2}}{w_{1,2}}\varrho_{\theta^{k+1}}-\frac{w_{0,2}}{w_{1,2}}\hat{\mathrm D}_2
\left(\frac{w_{0,2}}{w_{1,2}}\varrho_{\theta^{k+2}}-\frac{w_{0,2}}{w_{1,2}}\hat{\mathrm D}_2
\left(\dots\left(\frac{w_{0,2}}{w_{1,2}}\varrho_{\theta^{r+1}}
\right)\dots\right)\right),\quad k=2,\dots,r,
\\
f_{\theta^1}-\frac{(w_{1,2})^2}{w_{0,2}}f_{w_{1,2}}
=\frac{w_{0,2}}{w_{1,2}}\varrho_{\theta^2}-\frac{w_{0,2}}{w_{1,2}}\hat{\mathrm D}_2
\left(\frac{w_{0,2}}{w_{1,2}}\varrho_{\theta^3}-\frac{w_{0,2}}{w_{1,2}}\hat{\mathrm D}_2
\left(\dots\left(\frac{w_{0,2}}{w_{1,2}}\varrho_{\theta^{r+1}}
\right)\dots\right)\right).
\end{gather*}
This implies that
\[
f_{\theta^k}=\frac{w_{0,2}}{w_{1,2}}\hat\varphi^k+\check\varphi^k,\quad k=2,\dots,r,
\quad
f_{\theta^1}-\frac{(w_{1,2})^2}{w_{0,2}}f_{w_{1,2}}=\frac{w_{0,2}}{w_{1,2}}\hat\varphi^1+\check\varphi^1,
\]
where $\hat\varphi^k$ and~$\check\varphi^k$, $k=1,\dots,r$, are at most functions of $(w_{0,2},\theta^0,\dots,\theta^r)$.
The compatibility conditions of the last collection of equations as a system with respect to~$f$ are
$\hat\varphi^k_{\theta^l}=\hat\varphi^l_{\theta^k}$,
$\check\varphi^k_{\theta^l}+\delta_{l1}\hat\varphi^k=\check\varphi^l_{\theta^k}+\delta_{k1}\hat\varphi^l$,
$k,l=1,\dots,r$,
where $\delta_{kl}$ is the Kronecker delta.
Hence, integrating this system gives the following representation for~$f$:
\begin{gather}\label{eq:ReprOfFViaHatCheckBar}
f=\frac{w_{0,2}}{w_{1,2}}\hat f+\check f+\bar f,
\end{gather}
where $\hat f$ and~$\check f$ are at most functions of $(w_{0,2},\theta^0,\dots,\theta^r)$,
and the function~$\bar f$ depends at most on $(z_1,w_{0,0},w_{1,0},w_{0,1},w_{0,2},\theta^0,\zeta^{ik})$.
We substitute this representation into the equation~\eqref{eq:dNs1.3RhoNe1ModifiedRedEqInTotalDers3},
\begin{gather}\label{eq:dNs1.3RhoNe1ModifiedRedEqInTotalDers3bar}
\begin{split}
\hat{\mathrm D}_2\bar f={}&\phi-\frac{w_{1,2}}{w_{0,2}}\psi
+(a_2z_1+a_1)(w_{1,0}-z_2\zeta^{10})+(a_2z_2+a_0)w_{0,1}-a_2w_{0,0}\\&{}
+\frac{z_1}{12}(a_2z_1+2a_1)(w_{0,2})^3-\frac{z_2}4(a_2z_2+2a_0)w_{0,2},
\end{split}
\end{gather}
where we should substitute the above expressions for~$z_1$ and~$z_2$,
$z_1=\theta^1-w_{0,2}/w_{1,2}$ and $z_2=\theta^0+w_{0,2}\theta^1-(w_{0,2})^2/w_{1,2}$.
Since $\hat{\mathrm D}_2z_1=0$, the coefficients
\begin{gather}\label{eq:PhiPsiDef}
\phi:=\varrho-\frac{w_{0,2}}{w_{1,2}}\hat{\mathrm D}_2\hat f,\quad
\psi:=\theta^2\hat f+\frac{w_{0,2}}{w_{1,2}}\hat{\mathrm D}_2\check f
\end{gather}
depend at most on $(w_{0,2},\theta^0,\theta^1)$.
We replace the coordinates, using $z_1$ instead of $w_{1,2}$ as a coordinate.
%with $(z_1,w_{0,0},w_{1,0},w_{0,1},w_{0,2},\theta^k,\zeta^{ik})$.
Then the equation~\eqref{eq:dNs1.3RhoNe1ModifiedRedEqInTotalDers3bar} takes the form
\begin{gather}\label{eq:dNs1.3RhoNe1ModifiedRedEqInTotalDers3barB}
\begin{split}&
w_{0,1}\bar f_{w_{0,0}}+\left(\zeta^{10}-\frac12(w_{0,2})^2\right)\bar f_{w_{1,0}}+w_{0,2}\bar f_{w_{0,1}}
+\frac{\bar f_{\theta^0}\theta^1-\bar f_{w_{0,2}}}{\theta^1-z_1}
\\&\qquad
={}\phi-\frac\psi{\theta^1-z_1}
+\frac{z_1}{12}(a_2z_1+2a_1)(w_{0,2})^3-\frac{z_2}4(a_2z_2+2a_0)w_{0,2}
\\&\qquad\quad
+(a_2z_1+a_1)(w_{1,0}-z_2\zeta^{10})+(a_2z_2+a_0)w_{0,1}-a_2w_{0,0},
\end{split}
\end{gather}
where $z_2:=\theta^0+w_{0,2}z_1$.
We act on the last equation by the operator $(\theta^1-z_1)^2\p_{\theta^1}$,
\begin{gather}\label{eq:dNs1.3RhoNe1ModifiedRedEqInTotalDers3barConseq1}
\bar f_{w_{0,2}}-z_1\bar f_{\theta^0}
=(\theta^1-z_1)^2\phi_{\theta^1}-(\theta^1-z_1)\psi_{\theta^1}+\psi.
%=z_1^2\phi_{\theta^1}+z_1(\psi_{\theta^1}-2\theta^1\phi_{\theta^1})+\psi-\theta^1\psi_{\theta^1}+(\theta^1)^2\phi_{\theta^1}
\end{gather}
Differentiating the equation~\eqref{eq:dNs1.3RhoNe1ModifiedRedEqInTotalDers3barConseq1}
once more with respect to~$\theta^1$ and splitting the obtained equation with respect to~$z_1$,
we derive $\phi_{\theta^1\theta^1}=0$ and $\psi_{\theta^1\theta^1}=2\phi_{\theta^1}$.
%The third equation is $\theta^1\psi_{\theta^1\theta^1}-2\theta^1\phi_{\theta^1}=0$.
Hence
\[
\phi=\phi^1\theta^1+\phi^0,\quad
\psi=\phi^1(\theta^1)^2+\psi^1\theta^1+\psi^0,
\]
where $\phi^0$, $\phi^1$, $\psi^0$ and~$\psi^1$ are functions of $(w_{0,2},\theta^0)$.
Thus, the equation~\eqref{eq:dNs1.3RhoNe1ModifiedRedEqInTotalDers3barConseq1} reduces to
\begin{gather}\label{eq:dNs1.3RhoNe1ModifiedRedEqInTotalDers3barConseq2}
\bar f_{w_{0,2}}-z_1\bar f_{\theta^0}=\phi^1z_1^2+\psi^1z_1+\psi^0.
\end{gather}
We substitute the expression for $\bar f_{w_{0,2}}$
in view of~\eqref{eq:dNs1.3RhoNe1ModifiedRedEqInTotalDers3barConseq2}
into the equation~\eqref{eq:dNs1.3RhoNe1ModifiedRedEqInTotalDers3barB},
\begin{gather}\label{eq:dNs1.3RhoNe1ModifiedRedEqInTotalDers3barC}
\begin{split}&
w_{0,1}\bar f_{w_{0,0}}+\left(\zeta^{10}-\frac12(w_{0,2})^2\right)\bar f_{w_{1,0}}+w_{0,2}\bar f_{w_{0,1}}
+\bar f_{\theta^0}
\\&\qquad
={}\phi^0-\psi^1-z_1\phi^1
+\frac{z_1}{12}(a_2z_1+2a_1)(w_{0,2})^3-\frac{z_2}4(a_2z_2+2a_0)w_{0,2}
\\&\qquad\quad
+(a_2z_1+a_1)(w_{1,0}-z_2\zeta^{10})+(a_2z_2+a_0)w_{0,1}-a_2w_{0,0}.
\end{split}
\end{gather}
Differential consequences of~\eqref{eq:dNs1.3RhoNe1ModifiedRedEqInTotalDers3barConseq2}
and~\eqref{eq:dNs1.3RhoNe1ModifiedRedEqInTotalDers3barC} are
\begin{gather}\label{eq:BlockEq1}
-w_{0,2}\bar f_{w_{1,0}}+\bar f_{w_{0,1}}=V+\frac{z_1}4(a_2z_1+2a_1)(w_{0,2})^2-\frac{z_2}4(a_2z_2+2a_0),
\\\label{eq:BlockEq2}
-\bar f_{w_{1,0}}=(\p_{w_{0,2}}-z_1\p_{\theta^0})V+\frac{z_1}2(a_2z_1+2a_1)w_{0,2},
%\\\label{eq:BlockEq3}
%\bar f_{w_{0,1}}=V-w_{0,2}(\p_{w_{0,2}}-z_1\p_{\theta^0})V-\frac{z_1}4(a_2z_1+2a_1)(w_{0,2})^2-\frac{z_2}4(a_2z_2+2a_0),
\\\label{eq:BlockEq4}
-\bar f_{w_{0,0}}=\p_{\theta^0}V+(a_2z_1+a_1)w_{0,2}-\frac32(a_2z_2+a_0),
\\ \label{eq:V}
\p_{\theta^0}^2V=\frac52a_2,\quad
\p_{\theta^0}\p_{w_{0,2}}V=\frac32a_2z_1-a_1,\quad
\p_{w_{0,2}}^2V=-3a_1z_1,
\end{gather}
where
$V:=(\phi^0-\psi^1-z_1\phi^1)_{w_{0,2}}-(z_1\phi^0+\psi^0)_{\theta^0}$.
We integrate the equations~\eqref{eq:V} with respect to~$V$,
\begin{gather*}
V=\tfrac54a_2(\theta^0)^2+\left(\tfrac32a_2z_1-a_1\right)w_{0,2}\theta^0-\tfrac32a_1z_1(w_{0,2})^2
\\ \quad\quad{}
+(b_{11}z_1+b_{10})\theta^0+(b_{21}z_1+b_{20})w_{0,2}+b_{01}z_1+b_{00},
\end{gather*}
where $b_{ij}$, $i=0,1,2$, $j=0,1$, are arbitrary constants.
Recalling the definition of~$V$, we split the last equality with respect to~$z_1$,
and derive two equations for~$\phi^0$, $\phi^1$, $\psi^0$ and $\psi^1$,
whose general solution can be represented as
\begin{gather*}
\phi^0=\Phi_{w_{0,2}}-\tfrac34a_2w_{0,2}(\theta^0)^2-\tfrac12b_{11}(\theta^0)^2,
\\
\phi^1=-\Phi_{\theta^0}+\tfrac12a_1(w_{0,2})^3-\tfrac12b_{21}(w_{0,2})^2-b_{01}w_{0,2},
\\
\psi^0=\Psi_{w_{0,2}}+\tfrac12(a_1w_{0,2}-b_{10})(\theta^0)^2-\tfrac5{12}a_2(\theta^0)^3,
\\
\psi^1=\phi^0-\Psi_{\theta^0}-\tfrac12b_{20}(w_{0,2})^2-b_{00}w_{0,2},
\end{gather*}
where $\Phi$ and~$\Psi$ are arbitrary functions of~$(w_{0,2},\theta^0)$.
We substitute the expressions for~$\phi^0$, $\phi^1$, $\psi^0$, $\psi^1$ and~$V$
into the equations~\eqref{eq:dNs1.3RhoNe1ModifiedRedEqInTotalDers3barConseq2}--\eqref{eq:BlockEq4}
and integrate the obtained equations with respect to~$\bar f$,
\begin{gather*}
\bar f=z_1\Phi+\Psi+\beta
\\ \qquad{}
+\frac{a_0}2(3w_{0,0}-z_2w_{0,1})
+a_1\left(\frac{(w_{0,2})^2}{12}(2z_2^2-2z_2\theta^0+3(\theta^0)^2)+z_2w_{1,0}-z_2^2\zeta^{10}\right)
\\ \qquad{}
+a_2\left(\frac{w_{0,2}}{12}(z_1^2(w_{0,2})^2\theta^0-2z_2^3-3z_2^2\theta^0)
+z_1z_2w_{1,0}+z_2^2w_{0,1}-z_2w_{0,0}-z_1z_2^2\zeta^{10}\right)
\\ \qquad{}
+b_{10}\left(
-\frac{w_{0,2}}6(z_2^2+z_2\theta^0+(\theta^0)^2)
+z_1w_{1,0}+z_2w_{0,1}-w_{0,0}-z_1z_2\zeta^{10}\right)
\\ \qquad{}
+b_{11}\left(-z_1\frac{w_{0,2}}6(z_2^2+z_2\theta^0+(\theta^0)^2)
+z_1^2w_{1,0}+z_1z_2w_{0,1}-z_1w_{0,0}-z_1^2z_2\zeta^{10}
\right)
\\ \qquad{}
-(b_{01}z_1+b_{00})\left(\frac{z_1}2(w_{0,2})^2-w_{0,1}\right)
-(b_{21}z_1+b_{20})\left(\frac{z_1}6(w_{0,2})^3+w_{1,0}-z_2\zeta^{10}\right)
,
\end{gather*}
where $\beta$ is an arbitrary $z_2$-integral of~\eqref{eq:dNs1.3RhoNe1ModifiedRedEq}.
The equations~\eqref{eq:ReprOfFViaHatCheckBar} and~\eqref{eq:PhiPsiDef} imply
the following representation for~$f$ in terms of~$\check f$, $\bar f$ and~$\psi$:
\[
f=\frac{w_{0,2}}{w_{1,2}}\left(\frac\psi{\theta^2}-\frac1{\theta^2}\frac{w_{0,2}}{w_{1,2}}\hat{\mathrm D}_2\check f\right)
+\check f+\bar f.
\]
We substitute the expressions for~$\bar f$ and~$\psi$ and $b_{21}=\tilde b_{21}-b_{10}$ into this representation
and, denoting $\hat\varrho=\check f+\Phi\theta^1+\Psi+\tilde\Phi\theta^1+\tilde\Psi$
with
\begin{gather*}
\tilde\Phi=
\frac{a_1}6(w_{0,2})^3\theta^0-a_2(w_{0,2})^2(\theta^0)^2
-\frac{b_{00}}2(w_{0,2})^2+b_{01}w_{0,2}\theta^0-\frac{b_{10}}2(w_{0,2})^2\theta^0-\frac{b_{20}}6(w_{0,2})^3,
\\
\tilde\Psi=
\frac{a_1}4(w_{0,2})^2(\theta^0)^2-\frac5{12}a_2w_{0,2}(\theta^0)^3
+\frac{b_{01}}2(\theta^0)^2-\frac{b_{10}}2w_{0,2}(\theta^0)^2+\frac{b_{11}}6(\theta^0)^3,
\end{gather*}
obtain the expression for~$f$ from lemma's statement.
The expression for~$\varrho$ is found by solving~\eqref{eq:dNs1.3RhoNe1ModifiedRedEqInTotalDers3}
with respect to~$\varrho$.
\end{proof}

\subsection{Generalized symmetries}\label{sec:dNs1.3RhoNe1GenSyms}

The natural representatives of equivalence classes of generalized symmetries
of the equation~\eqref{eq:dNs1.3RhoNe1ModifiedRedEq}
are generalized vector fields in evolution form
whose characteristics are differential functions of~$w$
that depend only on parametric derivatives of this equation.

\begin{theorem}\label{thm:dNs1.3RhoNe1ModifiedRedEqGenSyms}
A differential function $f\{w\}$ is the characteristic
of a generalized symmetry of the equation~\eqref{eq:dNs1.3RhoNe1ModifiedRedEq}
if and only if it is a linear combination of the differential functions
\begin{gather*}
w_{1,0},\quad
z_1w_{1,0}+w_{0,0},\quad
z_1^2w_{1,0}+z_1z_2w_{0,1}-z_1w_{0,0}-\frac16z_2^{\,3},
\\[.5ex]
w_{0,1},\quad
2z_1w_{0,1}-z_2^{\,2},\quad
z_2w_{0,1}-3w_{0,0},\quad
\beta,\quad z_2\alpha,\quad
\frac{w_{0,2}}{w_{1,2}\theta^2}(\hat\varrho_{w_{0,2}}-\theta^{k+1}\hat\varrho_{\theta^k})+\hat\varrho,
\\[.5ex]
\frac{(w_{0,2})^3}{2w_{1,2}}\frac{(\theta^1)^2}{\theta^2}+\frac16z_1^2(w_{0,2})^3+z_1w_{1,0}, \quad
\frac23\frac{(w_{0,2})^4}{w_{1,2}}\frac{(\theta^1)^2}{\theta^2}+\frac16z_1^2(w_{0,2})^4+z_2w_{1,0}-z_2^2\zeta^{10},
\\[.5ex]
2\frac{(w_{0,2})^3}{w_{1,2}}\frac{(\theta^1)^2}{\theta^2}\theta^0+\frac23z_1^2(w_{0,2})^3\theta^0+\frac16z_1^3(w_{0,2})^4
+z_2w_{0,0}-z_2^2w_{0,1}-z_1z_2w_{1,0}+z_1z_2^2\zeta^{10},
\end{gather*}
where $\alpha$ and~$\beta$ are arbitrary $z_2$-integrals of~\eqref{eq:dNs1.3RhoNe1ModifiedRedEq},
the $\hat\varrho$ is an arbitrary function of~$w_{0,2}$ and a~finite number of~$\theta^k$.
\end{theorem}

\begin{proof}
The proof of the ``if''-part reduces to the substitution of each of the listed differential functions
into the generalized invariance condition for the equation~\eqref{eq:dNs1.3RhoNe1ModifiedRedEq},
\begin{gather}\label{eq:dNs1.3RhoNe1ModifiedRedEqGenInvCond}
\hat{\mathrm D}_1\hat{\mathrm D}_2^{\,2}f+w_{0,2}\hat{\mathrm D}_2^{\,3}f-\frac{w_{1,2}}{w_{0,2}}\hat{\mathrm D}_2^{\,2}f
=\hat{\mathrm D}_2\big(\hat{\mathrm D}_1+w_{0,2}\hat{\mathrm D}_2\big)\hat{\mathrm D}_2f
=0.
\end{gather}

Let us prove the ``only if''-part.
Suppose that a differential function $f\{w\}$ is the characteristic
of a generalized symmetry of the equation~\eqref{eq:dNs1.3RhoNe1ModifiedRedEq}.
Then it satisfies the generalized invariance condition~\eqref{eq:dNs1.3RhoNe1ModifiedRedEqGenInvCond},
which is equivalent to the condition
$%\begin{gather*}%\label{eq:dNs1.3RhoNe1ModifiedRedEqGenInvCondIntegrated}
\big(\hat{\mathrm D}_1+w_{0,2}\hat{\mathrm D}_2\big)\hat{\mathrm D}_2f=\alpha
$ %\end{gather*}
for some $z_2$-integral~$\alpha$ of the equation~\eqref{eq:dNs1.3RhoNe1ModifiedRedEq},
see Theorem~\ref{thm:dNs1.3RhoNe1Integrals} for the description of such integrals.
The further successive application of Lemmas~\ref{lem:dNs1.3RhoNe1AuxiliaryLemma1},
\ref{lem:dNs1.3RhoNe1AuxiliaryLemma2b} (with $\kappa^0=\alpha$ and $\kappa^1=0$)
and~\ref{lem:dNs1.3RhoNe1AuxiliaryLemma3} leads to the required statement.
\end{proof}

Since the possible number of arguments of the parameter functions grows
when the order of the corresponding symmetries does,
the generalized symmetry algebra  of the equation~\eqref{eq:dNs1.3RhoNe1ModifiedRedEq}
cannot be generated via commuting its lower-order generalized symmetries.

\subsection{Cosymmetries}\label{sec:dNs1.3RhoNe1Cosyms}

\begin{theorem}\label{thm:dNs1.3RhoNe1Cosyms}
A differential function $f\{w\}$ is a cosymmetry of the equation~\eqref{eq:dNs1.3RhoNe1ModifiedRedEq}
if and only if it is a linear combination of the differential functions
\begin{gather*}
\varrho,\quad \alpha,\quad w_{0,2}\gamma-z_2\hat{\mathrm D}_1\gamma,
\\[.5ex]
w_{0,1}-\tfrac12z_1(w_{0,2})^2,\quad %a_0
\big(w_{0,1}-\tfrac12z_1(w_{0,2})^2\big)w_{0,2}-z_2\zeta^{10},\quad  %-a_2
\big(w_{0,1}-\tfrac12z_1(w_{0,2})^2\big)\theta^0+z_1z_2\zeta^{10},%a_1
\\[.5ex]
z_1(w_{1,0}-z_2\zeta^{10})+z_2w_{0,1}-w_{0,0}+\tfrac1{12}z_1^2(w_{0,2})^3-\tfrac14z_2^2w_{0,2},\quad %b_{01}
w_{1,0}-z_2\zeta^{10}+\tfrac16z_1(w_{0,2})^3, %b_{00}
\\[.5ex]
\big(z_1(w_{1,0}-z_2\zeta^{10})+z_2w_{0,1}-w_{0,0}+\tfrac1{12}z_1^2(w_{0,2})^3-\tfrac14z_2^2w_{0,2}\big)\theta^0
+z_1^{\,2}z_2\zeta^{20}, %b_{11}
\\[.5ex]
\big(z_1(w_{1,0}-z_2\zeta^{10})+z_2w_{0,1}-w_{0,0}+\tfrac1{12}z_1^2(w_{0,2})^3-\tfrac14z_2^2w_{0,2}\big)w_{0,2}
-z_1z_2\zeta^{20}, %-b_{21}
\\[.5ex]
z_2(w_{1,0}-z_2\zeta^{10})+(z_2w_{0,1}-w_{0,0})w_{0,2}-\tfrac16z_2^2(w_{0,2})^2, %b_{10}
\\[.5ex]
(w_{1,0}-z_2\zeta^{10})w_{0,2}+\tfrac16z_1(w_{0,2})^4-z_2\zeta^{20}, %-b_{20}
\end{gather*}
where~$\varrho$ is an arbitrary function at most of~$w_{0,2}$ and a finite number of~$\theta^k$,
and $\alpha$ and~$\gamma$ are arbitrary $z_2$-integrals of~\eqref{eq:dNs1.3RhoNe1ModifiedRedEq}.
\end{theorem}

\begin{proof}
On can prove the ``if''-part by substituting each of the listed differential functions
into the condition defining cosymmetries of the equation~\eqref{eq:dNs1.3RhoNe1ModifiedRedEq},
which is formally adjoint to the generalized invariance condition for this equation,
\begin{gather}\label{eq:dNs1.3RhoNe1ModifiedRedEqCosymCond}
%-\hat{\mathrm D}_2\big(\hat{\mathrm D}_1+\hat{\mathrm D}_2\circ w_{0,2}\big)\hat{\mathrm D}_2f=
-\hat{\mathrm D}_2^2(\hat{\mathrm D}_1+w_{0,2}\hat{\mathrm D}_2\big)f=0.
\end{gather}

It remains to prove the ``only if''-part.
Suppose that a differential function $f\{w\}$ is the characteristic of a cosymmetry of~\eqref{eq:dNs1.3RhoNe1ModifiedRedEq}.
Then it satisfies the condition~\eqref{eq:dNs1.3RhoNe1ModifiedRedEqCosymCond},
which is equivalent to the condition
\begin{gather*}%\label{eq:dNs1.3RhoNe1ModifiedRedEqGenInvCondIntegrated}
\big(\hat{\mathrm D}_1+w_{0,2}\hat{\mathrm D}_2\big)f=\kappa^0+z_2\kappa^1
\end{gather*}
for some $z_2$-integrals~$\kappa^0$ and~$\kappa^1$ of the equation~\eqref{eq:dNs1.3RhoNe1ModifiedRedEq},
see Theorem~\ref{thm:dNs1.3RhoNe1Integrals} for the description of such integrals.
The further application of Lemma~\ref{lem:dNs1.3RhoNe1AuxiliaryLemma2b}
leads to the required statement.
\end{proof}

In view of~\eqref{eq:dNs1.3RhoNe1ModifiedRedEqGenInvCond}
and~\eqref{eq:dNs1.3RhoNe1ModifiedRedEqCosymCond},
the operator~$\hat{\mathrm D}_2$ is an inverse Noether operator
for the equation~\eqref{eq:dNs1.3RhoNe1ModifiedRedEq}.

\subsection{Conservation laws}\label{sec:dNs1.3RhoNe1CLs}

\begin{lemma}\label{lem:dNs1.3RhoNe1ModifiedRedEqGenFormOfCCs}
Any conserved current of the equation~\eqref{eq:dNs1.3RhoNe1ModifiedRedEq} is equivalent to
a linear combination of the tuples
\begin{gather*}
\left(\frac{w_{1,2}}{w_{0,2}}\varrho,\,w_{1,2}\varrho\right),
\quad
\big(0,\,\alpha\big),
\\[1ex]
\left(
\frac{(w_{0,2})^5}{24w_{1,2}}+\frac12w_{0,1}(w_{0,2})^2,\,
\frac{(w_{0,2})^6}{24w_{1,2}}-z_2(\zeta^{10})^2+\frac13w_{0,1}(w_{0,2})^3
+w_{1,0}\zeta^{10}\right),
\\[1ex]
\left(
-\frac{(w_{0,2})^5}{24w_{1,2}}\left(z_1+\frac{w_{0,2}}{3w_{1,2}}\right)
-\frac12w_{0,1}(z_1(w_{0,2})^2+w_{0,1}),\right.
\\\quad\left.
-\frac{(w_{0,2})^6}{24w_{1,2}}\left(z_1+\frac{w_{0,2}}{3w_{1,2}}\right)
+z_1z_2(\zeta^{10})^2-\frac13z_1w_{0,1}(w_{0,2})^3-z_1w_{1,0}\zeta^{10}
+w_{0,0}\zeta^{10}
\right),
\end{gather*}
where $\varrho$ is an arbitrary function at most of~$w_{0,2}$ and a finite number of~$\theta^k$,
and $\alpha$ is an arbitrary $z_2$-integral of~\eqref{eq:dNs1.3RhoNe1ModifiedRedEq}.
\end{lemma}

\begin{proof}
Let $(F^1,F^2)$ be a conserved current of the equation~\eqref{eq:dNs1.3RhoNe1ModifiedRedEq}.
Without loss of generality, up to the conserved-current equivalence
related to vanishing on the solution set of~\eqref{eq:dNs1.3RhoNe1ModifiedRedEq},
we can assume that the components~$F^1$ and~$F^2$ depend only on parametric derivatives of~\eqref{eq:dNs1.3RhoNe1ModifiedRedEq},
$F^1=F^1\{w\}$ and $F^2=F^2\{w\}$.
We use the coordinate system%~\eqref{eq:dNs1.3RhoNe1ModifiedRedEqCoords}\todo or
~\eqref{eq:dNs1.3RhoNe1ModifiedRedEqCoords2} on~$\mathcal L$.
Let $k_i:=\ord_i F^1$. %and $k_0=\ord_0 F^1$.
The condition that the tuple $(F^1,F^2)$ is a conserved current of~\eqref{eq:dNs1.3RhoNe1ModifiedRedEq}
reduces to the equation
\begin{gather}\label{eq:dNs1.3RhoNe1ModifiedRedEqCCCond}
\hat{\mathrm D}_1F^1+\hat{\mathrm D}_2F^2=0.
\end{gather}
We fix an arbitrary point
$\mathbf j_0^{}=(w_{0,0}^0,w_{1,0}^0,w_{0,1}^0,w_{0,2}^0,w_{1,2}^0,\theta^k_0,\zeta^{ik}_0,\,k\in\mathbb N_0,\,i=1,2)$
in the domain of $(F^1,F^2)$.
(Only a finite number of components of~$\mathbf j_0^{}$ is relevant for the proof.)
When integrating with respect to a jet variable,
we take the definite integral with respect to this variable with variable upper boundary
and lower bound equal to the corresponding component of~$\mathbf j_0^{}$
such that the integration line is contained in the domain of $(F^1,F^2)$.
We further consider various differential consequences,
marking them as the corresponding differential operator acting on~\eqref{eq:dNs1.3RhoNe1ModifiedRedEqCCCond}.
Note that the operator~$\hat{\mathrm D}_2$ commutes
with $\p_{\zeta^{1k}}$, $k\geqslant1$, and with $\p_{\zeta^{2k}}$, $k\geqslant0$.

Let $k_1\geqslant1$ and $k_2\geqslant0$.
We proceed first with the value $i=1$ and then with the value $i=2$.
The differential consequences
\begin{gather*}
\p_{\zeta^{ik}}\eqref{eq:dNs1.3RhoNe1ModifiedRedEqCCCond},\ k>k_i+1\colon\ \hat{\mathrm D}_2F^2_{\zeta^{ik}}=0,
\\
\p_{\zeta^{i,k_i+1}}\eqref{eq:dNs1.3RhoNe1ModifiedRedEqCCCond}\colon\
F^1_{\zeta^{ik_i}}+\hat{\mathrm D}_2F^2_{\zeta^{i,k_i+1}}=0
\end{gather*}
imply that the differential function $\hat{\mathrm D}_2F^2$
does not depend on $\zeta^{ik}$ with $k>k_i+1$ and is affine with respect to~$\zeta^{i,k_i+1}$,
i.e., its derivative \smash{$\hat{\mathrm D}_2F^2_{\zeta^{i,k_i+1}}$} does not depend on $\zeta^{i,k_i+1}$.
Integrating the differential consequence \smash{$\p_{\zeta^{i,k_i+1}}\eqref{eq:dNs1.3RhoNe1ModifiedRedEqCCCond}$}
with respect to the jet variable~$\zeta^{ik_i}$, we obtain
\begin{gather*}
F^1=\tilde F^1-\hat{\mathrm D}_2H,\quad\mbox{where}\quad
\tilde F^1:=F^1\Big|_{\zeta^{ik_i}=\zeta^{ik_i}_0}^{},\quad
H:=\int_{\zeta^{ik_i}_0}^{\zeta^{ik_i}}F^2_{\zeta^{i,k_i+1}}\Big|_{\zeta^{ik_i}=\varsigma}^{}{\rm d}\varsigma.
\end{gather*}
The tuple $(\tilde F^1,\tilde F^2)$ with $\tilde F^2:=F^2-\hat{\mathrm D}_1H$
is a conserved current of~\eqref{eq:dNs1.3RhoNe1ModifiedRedEq} that is equivalent to $(F^1,F^2)$.
We also have $\ord_i\tilde F^1<k_i$,
and for the other value~$i'$ of~$i$, $\ord_{i'}\tilde F^1\leqslant k_{i'}$.
We replace the tuple $(F^1,F^2)$ by $(\tilde F^1,\tilde F^2)$,
re-denote $(\tilde F^1,\tilde F^2)$ by $(F^1,F^2)$ and iterate the above procedure.
As a result, we conclude that up to adding null divergences, we can assume that
$k_1\leqslant0$ and $k_2=-\infty$.

The differential consequence
\begin{gather*}
\p_{\zeta^{20}}\eqref{eq:dNs1.3RhoNe1ModifiedRedEqCCCond}\colon\
F^1_{w_{1,0}}+\hat{\mathrm D}_2F^2_{\zeta^{20}}=0
\end{gather*}
implies that the differential function $\hat{\mathrm D}_2F^2$ is affine with respect to~$\zeta^{20}$,
i.e., its derivative \smash{$\hat{\mathrm D}_2F^2_{\zeta^{20}}$} does not depend on $\zeta^{20}$.
We integrate this differential consequence with respect to~$w_{1,0}$ to derive
\begin{gather*}
F^1=\tilde F^1-\hat{\mathrm D}_2H,\quad\mbox{where}\quad
\tilde F^1:=F^1\Big|_{w_{1,0}=w_{1,0}^0}^{},\quad
H:=\int_{w_{1,0}^0}^{w_{1,0}}F^2_{\zeta^{20}}\Big|_{w_{1,0}=\varsigma}^{}{\rm d}\varsigma.
\end{gather*}
The tuple $(\tilde F^1,\tilde F^2)$ with $\tilde F^2:=F^2-\hat{\mathrm D}_1H$
is a conserved current of~\eqref{eq:dNs1.3RhoNe1ModifiedRedEq} that is equivalent to $(F^1,F^2)$.
Moreover, $\tilde F^1_{w_{1,0}}=F^1_{w_{1,0}}+\hat{\mathrm D}_2H_{w_{1,0}}=0$,
$\ord_1\tilde F^1\leqslant0$ and $\ord_2\tilde F^1=-\infty$.
Therefore, up to adding null divergences, we can in addition assume that $F^1_{w_{1,0}}=0$.

Acting as above with the differential consequences
\begin{gather*}
\p_{\zeta^{11}}\eqref{eq:dNs1.3RhoNe1ModifiedRedEqCCCond}\colon\
F^1_{\zeta^{10}}+\hat{\mathrm D}_2F^2_{\zeta^{11}}=0,
\\
\p_{w_{1,0}}\eqref{eq:dNs1.3RhoNe1ModifiedRedEqCCCond}\colon\
F^1_{w_{0,0}}+\hat{\mathrm D}_2F^2_{w_{1,0}}=0,
\end{gather*}
we can replace the initial conserved current
with the equivalent conserved current that in addition satisfies the constraints \smash{$F^1_{\zeta^{10}}=F^1_{w_{0,0}}=0$}.
%\[\p_{w_{0,0}}\eqref{eq:dNs1.3RhoNe1ModifiedRedEqCCCond}\colon\ \hat{\mathrm D}_2F^2_{w_{0,0}}=0.\]
Then, differentiating the equation~\eqref{eq:dNs1.3RhoNe1ModifiedRedEqCCCond} with respect to~$w_{0,0}$,
we derive the equation $\p_{w_{0,0}}\eqref{eq:dNs1.3RhoNe1ModifiedRedEqCCCond}\colon$
\smash{$\hat{\mathrm D}_2F^2_{w_{0,0}}=0$}.
As a result, the differential functions~$F^1$ and~$F^2$ satisfy the system
\begin{gather*}
F^1_\varsigma=0,\quad \varsigma\in\{\zeta^{1k},\zeta^{2k},k\in\mathbb N_0,\,w_{1,0},\,w_{0,0}\},\\
\hat{\mathrm D}_2F^2_\tau=0,\quad \tau\in\{\zeta^{1k},k\in\mathbb N,\,\zeta^{2l},l\in\mathbb N_0,\,w_{1,0},\,w_{0,0}\},
\end{gather*}
which integrates to
\begin{gather*}
F^1=F^1(w_{0,1},w_{0,2},w_{1,2},\theta^*),\\
F^2=\bar F^2(w_{0,1},w_{0,2},w_{1,2},\theta^*,\zeta^{10})
+\kappa^1(z_1,\zeta^{10})w_{1,0}+\kappa^0(z_1,\zeta^{10})w_{0,0}
+\alpha(z_1,\zeta^{1*},\zeta^{2*}),
\end{gather*}
where $\bar F^2$, $\kappa^0$, $\kappa^1$ and~$\alpha$ are arbitrary functions of their arguments.

Differentiating the differential consequence
\begin{gather*}
\p_{\zeta^{10}}\eqref{eq:dNs1.3RhoNe1ModifiedRedEqCCCond}\colon\
F^1_{w_{0,1}}+\hat{\mathrm D}_2F^2_{\zeta^{10}}+F^2_{w_{1,0}}=0
\end{gather*}
in addition with respect to~$\zeta^{10}$ gives the equation
\smash{$\hat{\mathrm D}_2F^2_{\zeta^{10}\zeta^{10}}+2\kappa^1_{\zeta^{10}}=0$}.
In view of Theorem~\ref{thm:dNs1.3RhoNe1Integrals} and the above representation for~$F^2$,
this equation integrates, as an inhomogeneous equation with respect to~$\bar F^2_{\zeta^{10}\zeta^{10}}$, to
\begin{gather*}
\bar F^2_{\zeta^{10}\zeta^{10}}+\kappa^1_{\zeta^{10}\zeta^{10}}w_{1,0}+\kappa^0_{\zeta^{10}\zeta^{10}}w_{0,0}=-2\kappa^1_{\zeta^{10}}z_2+\beta,
\end{gather*}
where $\beta$ is an arbitrary $z_2$-integral of~\eqref{eq:dNs1.3RhoNe1ModifiedRedEq}.
Splitting the last equation with respect to~$w_{1,0}$ and~$w_{0,0}$
and separately differentiating it with respect to ${\zeta^{1k}}$, $k\in\mathbb N$, and $\zeta^{2l}$, $l\in\mathbb N_0$,
we derive the equations
\smash{$\bar F^2_{\zeta^{10}\zeta^{10}}=-2\kappa^1_{\zeta^{10}}z_2+\beta$},
$\kappa^0_{\zeta^{10}\zeta^{10}}=\kappa^1_{\zeta^{10}\zeta^{10}}=0$,
$\beta_{\zeta^{1k}}=0$, $k\in\mathbb N$, and $\beta_{\zeta^{2l}}=0$, $l\in\mathbb N_0$.
Hence the function~$\beta$ depends at most on~$(z_1,\zeta^{10})$,
$\kappa^0=\kappa^{01}(z_1)\zeta^{10}+\kappa^{00}(z_1)$ and
$\kappa^1=\kappa^{11}(z_1)\zeta^{10}+\kappa^{10}(z_1)$ for some functions~$\kappa^{00}$, $\kappa^{01}$, $\kappa^{10}$ and~$\kappa^{11}$ of~$z_1$,
and
\begin{gather*}
\bar F^2=-\kappa^{11}(\zeta^{10})^2z_2+\bar\beta
+F^{21}(w_{0,1},w_{0,2},w_{1,2},\theta^*)\zeta^{10}+F^{20}(w_{0,1},w_{0,2},w_{1,2},\theta^*),
\end{gather*}
where $\bar\beta=\bar\beta(z_1,\zeta^{10})$ is a second antiderivative of~$\beta$ with respect to $\zeta^{10}$,
$\bar\beta_{\zeta^{10}\zeta^{10}}=\beta$.
Re-denoting $\alpha+\bar\beta$ by~$\alpha$, we set $\bar\beta=0$.
Then the differential consequence $\p_{\zeta^{10}}\eqref{eq:dNs1.3RhoNe1ModifiedRedEqCCCond}$
reduces to
\begin{gather}\label{eq:dNs1.3RhoNe1ModifiedRedEqCCCondByZeta10}
F^1_{w_{0,1}}+\hat{\mathrm D}_2F^{21}-\frac{\kappa^{11}}2(w_{0,2})^2+\kappa^{01}w_{0,1}+\kappa^{10}=0.
\end{gather}
The integration of~\eqref{eq:dNs1.3RhoNe1ModifiedRedEqCCCondByZeta10}
with respect to~$w_{0,1}$ leads to
\begin{gather*}
F^1-F^1\Big|_{w_{0,1}=w_{0,1}^0}^{}+\hat{\mathrm D}_2H-w_{0,2}F^{21}\Big|_{w_{0,1}=w_{0,1}^0}^{}
\\\qquad{}
=\left(\frac{\kappa^{11}}2(w_{0,2})^2-\kappa^{10}\right)(w_{0,1}-w_{0,1}^0)-\frac{\kappa^{01}}2\big((w_{0,1})^2-(w_{0,1}^0)^2\big)
\end{gather*}
with
\begin{gather*}
H:=\int_{w_{0,1}^0}^{w_{0,1}}F^{21}\Big|_{w_{0,1}=\varsigma}^{}{\rm d}\varsigma+\kappa^{10}(w_{0,0}-w_{0,1}^0z_2).
\end{gather*}
In other words, we derive the representation $F^1=\tilde F^1-\hat{\mathrm D}_2H$,
where
\begin{gather*}
\tilde F^1:=\bar F^1+\frac{\kappa^{11}}2(w_{0,2})^2w_{0,1}-\frac{\kappa^{01}}2(w_{0,1})^2,
\\
\bar F^1:=F^1\Big|_{w_{0,1}=w_{0,1}^0}^{}+w_{0,2}F^{21}\Big|_{w_{0,1}=w_{0,1}^0}^{}
-\frac{\kappa^{11}}2w_{0,1}^0+\frac{\kappa^{01}}2(w_{0,1}^0)^2,
\end{gather*}
and thus the differential function~$\bar F^1$ depends at most on $(w_{0,2},w_{1,2},\theta^*)$.
Replacing the conserved current $(F^1,F^2)$ by the equivalent conserved current $(\tilde F^1,\tilde F^2)$, where
$\tilde F^1:=F^1+\hat{\mathrm D}_2H$ and
$\tilde F^2:=F^2-\hat{\mathrm D}_1H$, we have
\[
F^1=\bar F^1(w_{0,2},w_{1,2},\theta^*)
+\frac{\kappa^{11}}2(w_{0,2})^2w_{0,1}-\frac{\kappa^{01}}2(w_{0,1})^2.
\]
Substituting the obtained expression for~$F^1$ into~\eqref{eq:dNs1.3RhoNe1ModifiedRedEqCCCondByZeta10},
we derive the equation ${\hat{\mathrm D}_2F^{21}+\kappa^{10}=0}$
integrating to $F^{21}=-\kappa^{10}z_2+\mu(z_1)$.
Re-denoting $\alpha+\mu\zeta^{10}$ by $\alpha$, we set $\mu=0$.
As a result, on this stage we obtain the following representation for~$F^2$:
\begin{gather*}
\begin{split}
F^2={}&F^{20}(w_{0,1},w_{0,2},w_{1,2},\theta^*)+\alpha(z_1,\zeta^{1*},\zeta^{2*})
\\&{}
-\kappa^{11}(\zeta^{10})^2z_2-\kappa^{10}z_2\zeta^{10}+\kappa^1(z_1,\zeta^{10})w_{1,0}+\kappa^0(z_1,\zeta^{10})w_{0,0}.
\end{split}
\end{gather*}
Substituting the above expressions for $F^1$ and~$F^2$ into the differential consequence
\begin{gather*}
\p_{w_{0,1}}\eqref{eq:dNs1.3RhoNe1ModifiedRedEqCCCond}\colon\
\hat{\mathrm D}_1F^1_{w_{0,1}}+\hat{\mathrm D}_2F^2_{w_{0,1}}+F^2_{w_{0,0}}=0,
\end{gather*}
we derive the equation
\begin{gather*}
\hat{\mathrm D}_2F^{20}_{w_{0,1}}
+\frac12(\kappa^{11}_{z_1}+\kappa^{01})(w_{0,2})^2+\kappa^{11}w_{0,2}w_{1,2}
-\kappa^{01}_{z_1}w_{0,1}+\kappa^{00}=0.
\end{gather*}
Successively splitting it with respect to~$\theta^k$ from the highest appearing~$k$ to $k=2$ implies
that the function $\chi:=F^{20}_{w_{0,1}}$ depends at most on $(w_{0,1},w_{0,2},w_{1,2},\theta^0,\theta^1)$
and additionally satisfies the equations
\begin{gather}\label{eq:Commut1}
\chi_{\theta^1}-\frac{(w_{1,2})^2}{w_{0,2}}\chi_{w_{1,2}}=0,
\\\label{eq:Commut2}
\begin{split}&
w_{0,2}\chi_{w_{0,1}}-\frac{w_{1,2}}{w_{0,2}}\chi_{w_{0,2}}-\left(\frac{w_{1,2}}{w_{0,2}}\right)^2\chi_{w_{1,2}}
+\frac{w_{1,2}}{w_{0,2}}\theta^1\chi_{\theta^0}\\&\qquad{}
+\frac12(\kappa^{11}_{z_1}+\kappa^{01})(w_{0,2})^2+\kappa^{11}w_{0,2}w_{1,2}-\kappa^{01}_{z_1}w_{0,1}+\kappa^{00}=0.
\end{split}
\end{gather}
The nontrivial differential consequences of these equations are
\begin{gather*}%\label{eq:Commut3}
%[\eqref{eq:Commut1},\eqref{eq:Commut2}]:\quad
\chi_{w_{0,2}}+\frac{w_{1,2}}{w_{0,2}}\chi_{w_{1,2}}
+\left(\frac{w_{0,2}}{w_{1,2}}-\theta^1\right)\chi_{\theta^0}-\kappa^{11}(w_{0,2})^2=0,
%\\[\eqref{eq:Commut1},\eqref{eq:Commut3}]:\quad 0\equiv0,\nonumber
\\%\label{eq:Commut4}
%[\eqref{eq:Commut2},\eqref{eq:Commut3}]:\quad
-\chi_{w_{0,1}}-(\kappa^{11}_{z_1}+\kappa^{01})w_{0,2}=0,
\quad%\label{eq:Commut6}
%[\eqref{eq:Commut2},\eqref{eq:Commut4}]:\quad
-\kappa^{01}_{z_1}+(\kappa^{11}_{z_1}+\kappa^{01})\frac{w_{1,2}}{w_{0,2}}=0.
\end{gather*}
The last differential consequence splits with respect to~$w_{1,2}$ to $\kappa^{01}_{z_1}=0$ and $\kappa^{11}_{z_1}=-\kappa^{01}$,
which integrates to $\kappa^{01}=c_1$, $\kappa^{11}=-c_1z_1+c_0$, where $c_0$ and~$c_1$ are arbitrary constants.
Then the second differential consequence reduces to $\chi_{w_{0,1}}=0$.
Taking into account the obtained equations, we combine~\eqref{eq:Commut2} with the first differential consequence to
$\chi_{\theta^0}=-\kappa^{00}$.
Jointly with~\eqref{eq:Commut1}, this implies the ansatz
$\chi=\tilde\chi(z_1,w_{0,2})-\kappa^{00}\theta^0$,
which reduces the first differential consequence to the equation
$\tilde\chi_{w_{0,2}}+z_1\kappa^{00}-\kappa^{11}(w_{0,2})^2=0$,
where $z_1$ and~$w_{0,2}$ are considered as independent variables,
i.e., $\tilde\chi=\frac13\kappa^{11}(w_{0,2})^3-z_1\kappa^{00}w_{0,2}+\chi^0(z_1)$. %for some function $\chi^0$ of~$z_1$.
As a result,
\begin{gather*}
\chi=\frac{\kappa^{11}}3(w_{0,2})^3-\kappa^{00}z_2+\chi^0,\\
F^{20}=\left(\frac{\kappa^{11}}3(w_{0,2})^3-\kappa^{00}z_2+\chi^0\right)w_{0,1}+\psi(w_{0,2},w_{1,2},\theta^*).
\end{gather*}

Denoting $\varphi:=w_{0,2}\bar F^1-\psi$,
we substitute $\psi=w_{0,2}\bar F^1-\varphi$ into the remainder of the condition~\eqref{eq:dNs1.3RhoNe1ModifiedRedEqCCCond}
and rewrite it as
\begin{gather*}
w_{1,2}\bar F^1_{w_{1,2}}-\bar F^1=
\frac{w_{0,2}}{w_{1,2}}\hat{\mathrm D}_2\varphi
-\frac{\kappa^{11}}{12}\frac{(w_{0,2})^5}{w_{1,2}}
+\frac{(w_{0,2})^2}{w_{1,2}}\left(\frac{\kappa^{10}}2w_{0,2}+\kappa^{00}z_2-\chi^0\right).
\end{gather*}
We replace the conserved current $(F^1,F^2)$ by the equivalent conserved current $(\tilde F^1,\tilde F^2)$,
where
\[
\tilde F^1:=F^1+\hat{\mathrm D}_2H,\quad
\tilde F^2:=F^2-\hat{\mathrm D}_1H,\quad
H:=-w_{0,2}\int_{w_{1,2}^0}^{w_{1,2}}\frac\varphi{(w_{1,2})^2}\bigg|_{w_{1,2}=\varsigma}^{}{\rm d}\varsigma,
\]
and thus set $\varphi=0$.
Hence we should integrate the equation
\begin{gather*}
\left(\frac{\bar F^1}{w_{1,2}}\right)_{w_{1,2}}=
-\frac1{12}\kappa^{11}\frac{(w_{0,2})^5}{(w_{1,2})^3}
+\frac{(w_{0,2})^2}{(w_{1,2})^3}\left(\frac12\kappa^{10}w_{0,2}+\kappa^{00}z_2-\chi^0\right).
\end{gather*}
Its general solution can be represented as
\begin{gather*}
\bar F^1=\frac{w_{1,2}}{w_{0,2}}\varrho(w_{0,2},\theta^*)
+\frac1{12}(w_{0,2})^5w_{1,2}\left(\frac{-c_1z_1+c_0}{2(w_{1,2})^2}-\frac{c_1w_{0,2}}{6(w_{1,2})^3}\right)
+\frac{w_{0,2}}2\hat{\mathrm D}_1(\lambda w_{0,2})
\\ \qquad{}
+\hat{\mathrm D}_1\big(z_2\mu_{z_1}w_{0,2}-\mu(w_{0,2})^2-\nu w_{0,2}\big),
\end{gather*}
where $\lambda$ and $\nu$ are second antiderivatives of~$\kappa^{10}$ and~$\chi^0$, respectively,
and $\mu$ is a third antiderivative of~$\kappa^{00}$.

As a result, we derive the following expressions for the components of conserved currents
of the equation~\eqref{eq:dNs1.3RhoNe1ModifiedRedEq}:
\begin{gather*}
F^1=\bar F^1-\tfrac12(c_1z_1-c_0)(w_{0,2})^2w_{0,1}-\tfrac12c_1(w_{0,1})^2,
\\[1ex]
F^2=w_{0,2}\bar F^1+\alpha(z_1,\zeta^{1*},\zeta^{2*})
+(c_1z_1-c_0)(\zeta^{10})^2z_2
+(-c_1z_1+c_0)\zeta^{10}w_{1,0}+c_1\zeta^{10}w_{0,0}
\\ \qquad{}
+\lambda_{z_1z_1}(w_{1,0}-z_2\zeta^{10})
+\big(\tfrac13(-c_1z_1+c_0)(w_{0,2})^3+\nu_{z_1z_1}\big)w_{0,1}+\mu_{z_1z_1z_1}(w_{0,0}-z_2w_{0,1}),
\end{gather*}
where the expression for~$\bar F^1$ is given in the previous displayed equation.
The conserved currents associated with the parameter functions
$\lambda=\lambda(z_1)$, $\mu=\mu(z_1)$ and $\nu=\nu(z_1)$,
\begin{gather}\label{eq:dNs1.3RhoNe1ModifiedRedEqEquivCCs}
\begin{split}&
\big(w_{0,2}\hat{\mathrm D}_1(\lambda w_{0,2}),\,
(w_{0,2})^2\hat{\mathrm D}_1(\lambda w_{0,2})+2\lambda_{z_1z_1}(w_{1,0}-z_2\zeta^{10})
\big),
\\[1ex]&
\big(\hat{\mathrm D}_1\big(z_2\mu_{z_1}w_{0,2}-\mu(w_{0,2})^2\big),\,
w_{0,2}\hat{\mathrm D}_1\big(z_2\mu_{z_1}w_{0,2}-\mu(w_{0,2})^2\big)+\mu_{z_1z_1z_1}(w_{0,0}-z_2w_{0,1})\big),
\\[1ex]&
\big(\hat{\mathrm D}_1(\nu w_{0,2}),\, w_{0,2}\hat{\mathrm D}_1(\nu w_{0,2})-\nu_{z_1z_1}w_{0,1}\big),
\end{split}
\end{gather}
are equivalent to conserved currents from the family associated with the parameter function~$\alpha$;
moreover, the conserved current associated with the parameter function~$\mu$
is equivalent to the conserved current associated with the parameter function~$\lambda$,
where $\lambda=-\frac12\mu$.
This can be verified directly using the definition of equivalent conserved currents
or via computing the conservation-law characteristics corresponding to these conserved currents,
see Remark~\ref{rem:dNs1.3RhoNe1ModifiedRedEqEquivCCs}.
The conserved currents associated with the parameter functions
$\varrho=\varrho(w_{0,2},\theta^*)$ and $\alpha=\alpha(z_1,\zeta^{1*},\zeta^{2*})$
and the constant parameters~$c_0$ and~$c_1$ are listed in the theorem's statement.
\end{proof}

Lemma~\ref{lem:dNs1.3RhoNe1ModifiedRedEqGenFormOfCCs} in particular implies
that although the equation~\eqref{eq:dNs1.3RhoNe1ModifiedRedEq} admits large family of $z_2$-integrals,
it possesses no $z_1$-integrals.

\begin{theorem}\label{thm:dNs1.3RhoNe1ModifiedRedEqCLChars}
The quotient space of conservation-law characteristics
of the equation~\eqref{eq:dNs1.3RhoNe1ModifiedRedEq} with respect to their equivalence
(i.e., modulo adding trivial characteristics, which vanishes on the solutions of~\eqref{eq:dNs1.3RhoNe1ModifiedRedEq})
is naturally isomorphic to the subspace spanned by the following differential functions:
\begin{gather*}%\label{eq:dNs1.3RhoNe1ModifiedRedEqTrivCLCharCondes}
-\left(\frac{w_{0,2}}{w_{1,2}}\hat{\mathrm D}_2\right)^k\varrho_{\theta^k},
\quad
(-\mathrm D_1)^k\alpha_{\zeta^{1k}}-(w_{0,2}-z_2\mathrm D_1)(-\mathrm D_1)^k\alpha_{\zeta^{2k}},
\\[.5ex]
\frac{(w_{0,2})^4}{3w_{1,2}}-\frac{(w_{0,2})^4}{12}\theta^2+w_{1,0}+w_{0,1}w_{0,2}-2z_2\zeta^{10},
\\[.5ex]
\frac{(w_{0,2})^4}{12}\theta^1\theta^2-\frac{(w_{0,2})^4}{3w_{1,2}}\theta^1
+\frac{(w_{0,2})^5}{6(w_{1,2})^2}
+w_{0,0}-z_1(w_{1,0}+w_{0,1}w_{0,2})+2z_1z_2\zeta^{10},
\end{gather*}
where~$\varrho$ is an arbitrary function at most of~$w_{0,2}$ and a finite number of~$\theta^k$,
and $\alpha$ is an arbitrary $z_2$-integral of~\eqref{eq:dNs1.3RhoNe1ModifiedRedEq}.
\end{theorem}

\begin{proof}
For each of the conserved currents $(F^1,F^2)$ of~\eqref{eq:dNs1.3RhoNe1ModifiedRedEq}
that are listed in Lemma~\ref{lem:dNs1.3RhoNe1ModifiedRedEqGenFormOfCCs},
we perform the procedure described, e.g., in \cite[p.~266]{olve1993A}
to construct the unique conservation-law characteristic $\lambda\{w\}$
of the conservation law containing this conserved current.
More specifically, we expand the total divergence $\mathrm D_1F^1+\mathrm D_2F^2$ of $(F^1,F^2)$
and iteratively make formal integration by parts in
(or, equivalently, apply the Lagrange identity to) each obtained summand
up to deriving a term with the left-hand side $L:=w_{1,2}+w_{0,2}w_{0,3}$ of the equation~\eqref{eq:dNs1.3RhoNe1ModifiedRedEq}
as a multiplier.
The other summands are represented as total derivatives of conserved currents
that are trivial due to vanishing on the solution set of~\eqref{eq:dNs1.3RhoNe1ModifiedRedEq}.
In the course of this cumbersome and nontrivial computation, we use the following identities:
\begin{gather*}
\mathrm D_2\zeta^{1k}=\mathrm D_2\mathrm D_1^{\,k}I^1=\mathrm D_1^{\,k}\mathrm D_2I^1=\mathrm D_1^{\,k}L,
\\
\mathrm D_2\zeta^{2k}=\mathrm D_2\mathrm D_1^{\,k}I^2=\mathrm D_1^{\,k}\mathrm D_2I^2=-\mathrm D_1^{\,k}(w_{0,2}+z_2\mathrm D_1)L,
\\
(\mathrm D_1+w_{0,2}\mathrm D_2)w_{0,2}=L,
\quad
(\mathrm D_1+w_{0,2}\mathrm D_2)w_{1,2}=\mathrm D_1L-w_{1,2}w_{0,3},
\\
(\mathrm D_1+w_{0,2}\mathrm D_2)\theta^0=-z_1L,
\quad
(\mathrm D_1+w_{0,2}\mathrm D_2)\theta^1=2\frac L{w_{1,2}}-\frac{w_{0,2}}{(w_{1,2})^2}\mathrm D_1L,
\\[-.5ex]
(\mathrm D_1+w_{0,2}\mathrm D_2)\theta^{k+2}=w_{0,2}\theta^{k+2}_{w_{l,2}}\mathrm D_1^l\frac L{w_{0,2}}.
\end{gather*}
Here we outline only computations for the second families of conserved currents:
\begin{gather*}
\mathrm D_10+\mathrm D_2\alpha
=\alpha_{\zeta^{1k}}\mathrm D_1^{\,k}L-\alpha_{\zeta^{2k}}\mathrm D_1^{\,k}(w_{0,2}+z_2\mathrm D_1)L
\\[1ex]\quad{}
=\big((-\mathrm D_1)^k\alpha_{\zeta^{1k}}-(w_{0,2}-z_2\mathrm D_1)(-\mathrm D_1)^k\alpha_{\zeta^{2k}}\big)L
+\sum_{k'=1}^k\mathrm D_1\Big(\big((-\mathrm D_1)^{k'-1}\alpha_{\zeta^{1k}}\big)\mathrm D_1^{\,k-k'}L\Big)
\\[.5ex]\qquad{}
-\mathrm D_1\Big(z_2\big((-\mathrm D_1)^k\alpha_{\zeta^{2k}}\big)L\Big)
-\sum_{k'=1}^k\mathrm D_1\Big(\big((-\mathrm D_1)^{k'-1}\alpha_{\zeta^{2k}}\big)\mathrm D_1^{\,k-k'}(w_{0,2}+z_2\mathrm D_1)L\Big).
\end{gather*}

The above procedure does not work directly for
the first family conserved currents from Lemma~\ref{lem:dNs1.3RhoNe1ModifiedRedEqGenFormOfCCs},
but we can use the relation of the equation~\eqref{eq:dNs1.3RhoNe1ModifiedRedEq}
to the inviscid Burgers equation~\eqref{eq:dNs1.3RhoNe1InviscidBurgersEq}.
As a result, we show that
for each fixed value of the parameter function~$\varrho$,
the corresponding conserved current belongs
to the conservation law of~\eqref{eq:dNs1.3RhoNe1ModifiedRedEq} with the characteristic
from the first family in the theorem's statement.

Note also that the restriction of the differential operators~$\mathrm D_1$ and~$\hat{\mathrm D}_1$
on differential functions depending only on~$(w_{*,2})$ are well-defined and coincide with each other.
\end{proof}

\begin{remark}\label{rem:dNs1.3RhoNe1ModifiedRedEqEquivCCs}
The conserved currents~\eqref{eq:dNs1.3RhoNe1ModifiedRedEqEquivCCs} belong to
the conservation laws with characteristics
$z_2\mu_{z_1z_1}-w_{0,2}\mu_{z_1}$, $-2(z_2\lambda_{z_1z_1}w_{0,2}-\lambda_{z_1})$ and $-\nu_{z_1}$,
which are elements of the second family of characteristics from Theorem~\ref{thm:dNs1.3RhoNe1ModifiedRedEqCLChars},
where $\alpha=-\mu\zeta^{21}$, $\alpha=2\lambda\zeta^{21}$ and $\alpha=\nu\zeta^{11}$, respectively.
This is why the conserved currents~\eqref{eq:dNs1.3RhoNe1ModifiedRedEqEquivCCs}
are not presented in Lemma~\ref{lem:dNs1.3RhoNe1ModifiedRedEqGenFormOfCCs}.
\end{remark}

Let~$V$ and~$V_0$ denote the linear span of conserved currents of the equation~\eqref{eq:dNs1.3RhoNe1ModifiedRedEq}
from Lemma~\ref{lem:dNs1.3RhoNe1ModifiedRedEqGenFormOfCCs}
and the subspace of trivial conserved currents belonging to~$V$.

\begin{lemma}\label{lem:dNs1.3RhoNe1ModifiedRedEqTrivCCs}
The subspace~$V_0$ consists of the tuples
\begin{gather*}
\big(\hat{\mathrm D}_2\hat\varrho+(-c_0+c_1w_{0,2}-\tfrac12c_2\theta^0)w_{1,2}\theta^0,\\\quad{}
w_{0,2}\hat{\mathrm D}_2\hat\varrho+(-c_0+c_1w_{0,2}-\tfrac12c_2\theta^0)w_{0,2}w_{1,2}\theta^0
+\hat{\mathrm D}_1\hat\alpha+c_0\zeta^{10}+(c_2z_1+c_1)\zeta^{20}\big),
\end{gather*}
where~$\hat\varrho$ is an arbitrary function at most of~$w_{0,2}$ and a finite number of~$\theta^k$,
$\hat\alpha$ is an arbitrary $z_2$-integral of~\eqref{eq:dNs1.3RhoNe1ModifiedRedEq}
and $c_0$, $c_1$ and $c_2$ are arbitrary constants.
\end{lemma}

\begin{proof}
Since the single equation~\eqref{eq:dNs1.3RhoNe1ModifiedRedEq} is a normal, totally nondegenerate system of differential equations,
in view of \cite[Theorem 4.26]{olve1993A},
a conserved current of~\eqref{eq:dNs1.3RhoNe1ModifiedRedEq} is trivial
if and only if the associated characteristic identically vanishes.
Denote by ${\rm CC}_1(\alpha)$, ${\rm CC}_2(\varrho)$, ${\rm CC}_3$, ${\rm CC}_4$
the conserved currents listed in Lemma~\ref{lem:dNs1.3RhoNe1ModifiedRedEqGenFormOfCCs}
and by ${\rm Ch}_1(\alpha)$, ${\rm Ch}_2(\varrho)$, ${\rm Ch}_3$, ${\rm Ch}_4$
the conservation-law characteristics listed in Theorem~\ref{thm:dNs1.3RhoNe1ModifiedRedEqCLChars}.
According to  Lemma~\ref{lem:dNs1.3RhoNe1ModifiedRedEqGenFormOfCCs}
and Theorem~\ref{thm:dNs1.3RhoNe1ModifiedRedEqCLChars},
any conserved current of~\eqref{eq:dNs1.3RhoNe1ModifiedRedEq} is equivalent to
a conserved current
${\rm CC}_1(\alpha)+{\rm CC}_2(\varrho)+b_1{\rm CC}_3+b_2{\rm CC}_4$,
where $\alpha$ is a $z_2$-integral of~\eqref{eq:dNs1.3RhoNe1ModifiedRedEq},
$\varrho$ is a function at most of~$w_{0,2}$ and a finite number of~$\theta^k$,
and $b_1$ and $b_2$ are constants.
The latter conserved current is contained in the conservation law with the characteristic
${\rm Ch}_1(\alpha)+{\rm Ch}_2(\varrho)+b_1{\rm Ch}_3+b_2{\rm Ch}_4$.
Further it is convenient to use the modified coordinates~\eqref{eq:dNs1.3RhoNe1ModifiedRedEqCoords2},
where $w_{1,2}$ is replaced by~$z_1$.
If the above characteristic vanishes, then differentiating the corresponding equality
with respect to each of the coordinates, splitting the obtained equations if possible
and taking into account the initial equality as well,
we derive the equations $b_1=b_2=0$ and
\begin{gather}\label{eq:dNs1.3RhoNe1ModifiedRedEqTrivCLCharCondsA}
(-\mathrm D_1)^k\alpha_{\zeta^{1k}}=c_0,\quad
%(-\mathrm D_1)^{k+1}\alpha_{\zeta^{2k}}=-c_2,\quad
(-\mathrm D_1)^k\alpha_{\zeta^{2k}}=c_2z_1+c_1,
\\\label{eq:dNs1.3RhoNe1ModifiedRedEqTrivCLCharCondsB}
\left(\frac{w_{0,2}}{w_{1,2}}\hat{\mathrm D}_2\right)^k\varrho_{\theta^k}=-c_0+c_1w_{0,2}-c_2\theta^0,
\end{gather}
where $c_0$, $c_1$ and $c_2$ are arbitrary constants.
The system~\eqref{eq:dNs1.3RhoNe1ModifiedRedEqTrivCLCharCondsA}
can be considered as an inhomogeneous system of linear partial differential equations with respect to~$\alpha$.
Its general solution is represented in the form
$\alpha=c_0\zeta^{10}+(c_2z_1+c_1)\zeta^{20}+\tilde\alpha$,
where $\tilde\alpha$ is the general solution of the homogeneous counterpart of this system,
$(-\mathrm D_1)^k\tilde\alpha_{\zeta^{1k}}=0$ and
$(-\mathrm D_1)^k\tilde\alpha_{\zeta^{2k}}=0$.
Here the operator $(-\mathrm D_1)^k\p_{\zeta^{1k}}$ and $(-\mathrm D_1)^k\p_{\zeta^{2k}}$
can be interpreted as the Euler operators in the dependent variables~$\zeta^{10}$ and~$\zeta^{20}$,
respectively, where $z_1$ is the only independent variable.
Hence Theorem~4.7 from \cite{olve1993A} implies the following (local) representation for~$\tilde\alpha$:
$\tilde\alpha=\mathrm D_1\hat\alpha$ for some $z_2$-integral~$\hat\alpha$ of~\eqref{eq:dNs1.3RhoNe1ModifiedRedEq},
where the operator~$\mathrm D_1$ can be replaced by~$\hat{\mathrm D}_1$.
In a similar way, we treat the equation~\eqref{eq:dNs1.3RhoNe1ModifiedRedEqTrivCLCharCondsB},
representing its general solution as $\varrho=-c_0\theta^0+c_1w_{0,2}\theta^0-\frac12c_2(\theta^0)^2+\tilde\varrho$,
where $\tilde\varrho$ is the general solution its homogeneous counterpart.
Since the restriction of the operator $(w_{0,2}/w_{1,2})\hat{\mathrm D}_2$
to the space with the coordinates $(w_{0,2},\theta^k)$ can be formally interpreted
as the total derivative operator with respect to~$-w_{0,2}$ with successive derivatives~$\theta^k$, $k\in\mathbb N$,
of~$\theta^0$,
we obtain that (locally) $\tilde\varrho=(w_{0,2}/w_{1,2})\hat{\mathrm D}_2\hat\varrho$,
where $\hat\varrho$ is an arbitrary function at most of~$w_{0,2}$ and a finite number~of~$\theta^k$.
\end{proof}

Lemmas~\ref{lem:dNs1.3RhoNe1ModifiedRedEqGenFormOfCCs} and~\ref{lem:dNs1.3RhoNe1ModifiedRedEqTrivCCs}
imply the following theorem.

\begin{theorem}\label{lem:dNs1.3RhoNe1ModifiedRedEqCLs}
The space~$\Omega$ of conservation laws of the equation~\eqref{eq:dNs1.3RhoNe1ModifiedRedEq}
is naturally isomorphic to the quotient of the space~$V$ by the subspace~$V_0$.
\end{theorem}

\subsection{Relation to symmetry-like objects of inviscid Burgers equation}\label{sec:InviscidBurgersEqInducedObjects}

The study of~\cite[Appendix]{baik1989a}
on the generalized symmetries of the inviscid Burgers equation~\eqref{eq:dNs1.3RhoNe1InviscidBurgersEq}
was extended in~\cite{popo2025b} to other local symmetry-like objects of this equation,
which include cosymmetries, conserved currents, conservation-law characteristics and conservation laws;
see also a short preliminary description of the above results in~\cite[Section~6]{vane2021a}.

It turns out that the differential substitution~$w_{0,2}=h$ induces
a homomorphism~$\boldsymbol{\bar\upsilon}\colon\Sigma\to\check\Sigma$
between the algebras~$\Sigma$ and~$\check\Sigma$
of canonical representatives of equivalences classes of generalized symmetries
of the equations~\eqref{eq:dNs1.3RhoNe1ModifiedRedEq} and~\eqref{eq:dNs1.3RhoNe1InviscidBurgersEq}.
This homomorphism can be represented as the result of the following successive operations:
(\emph{i}) the second prolongation of the generalized vector fields from~$\Sigma$,
(\emph{ii}) neglecting all the components of the obtained prolongations
that are associated with the derivatives of~$w$, except for~$w_{0,2}$, and
(\emph{iii}) replacing derivatives of~$w_{0,2}$ by the respective derivatives of~$h$.
In other words, the second total derivative of the characteristic of any element of~$\Sigma$
is, after substituting~$h$ for~$w_{0,2}$, the characteristic of an element of~$\check\Sigma$.
The characteristics of generalized vector fields spanning the algebra~$\Sigma$
are listed in Theorem~\ref{thm:dNs1.3RhoNe1ModifiedRedEqGenSyms}.
The algebra~$\check\Sigma$ is spanned by the generalized vector fields with characteristics $h_2\check\varrho(h,\check\theta^*)$,
where $\check\varrho$ is an arbitrary function of~$h$ and a finite number of
\[
\check\theta^k:=\left(\frac h{h_1}\check{\mathrm D}_2\right)^k(z_2-hz_1),\quad k\in\mathbb N_0,
\]
which are in fact the modified jet coordinates~$\theta^k$ written in terms of derivatives of~$h$.
%In what follows we omit checks over $\theta^k$.
Thus, under the homomorphism~$\boldsymbol{\bar\upsilon}$,
the characteristics listed in Theorem~\ref{thm:dNs1.3RhoNe1ModifiedRedEqGenSyms} are respectively mapped
to the characteristics of generalized symmetries of the equation~\eqref{eq:dNs1.3RhoNe1InviscidBurgersEq}
with the following values of the differential parameter function~$\check\varrho$:
\begin{gather*}
-h,\quad
-h\check\theta^1,\quad
\check\theta^0\check\theta^1,\quad
1,\quad
2\check\theta^1,\quad
\check\theta^0+h\check\theta^1,\quad
0,\quad 0,\quad
\check R^2\frac{\check R\hat\varrho}{\check\theta^2},
\\[.5ex]
-\check R^2\frac{h^2(\check\theta^1)^2}{2\check\theta^2}-\check R\frac{h^2\check\theta^1}2,
\quad
-\check R^2\frac{h^3\check R^2(\check\theta^0)^2}{3\check\theta^2}-h^2\check\theta^1+3h\check\theta^0,
\quad
-\check R^2\frac{h^2\check R^2(\check\theta^0)^3}{3\check\theta^2}-\check\theta^0\check R(h\check\theta^0).
\end{gather*}
Here $\hat\varrho$ is a function of~$h$ and a finite number of~$\check\theta^k$ that is obtained
by the substitution $w_{0,2}=h$
into the function denoted by the same symbol in Theorem~\ref{thm:dNs1.3RhoNe1ModifiedRedEqGenSyms},
the operators~$\check R$ and~$\check{\mathrm D}_2$
are the pushforwards of the operators~$R$ from Lemma~\ref{lem:dNs1.3RhoNe1AuxiliaryLemma3}
and~$\hat{\mathrm D}_2$ by the composition of~$\varpi$ and the substitution $w_{0,2}=h$,
\begin{gather*}
\check R:=-\p_h+\check\theta^{k+1}\p_{\check\theta^k},\quad
\check{\mathrm D}_2=\frac{h_1}h\check R-h_1^{\,2}\frac{h_1\check\theta^2+1}{h^2}\p_{h_1},
\end{gather*}
and~$\varpi$ is the natural projection from the space with the coordinates~\eqref{eq:dNs1.3RhoNe1ModifiedRedEqCoords2}
to the space with the coordinates $(w_{0,2},w_{1,2},\theta^k)$.
It is obvious that the homomorphism~$\boldsymbol{\bar\upsilon}$ is neither injective nor surjective.
The kernel of~$\boldsymbol{\bar\upsilon}$ is spanned by the generalized vector fields
with the characteristics~$\beta$ and~$z_2\alpha$,
where both $\alpha$ and~$\beta$ run through the set of $z_2$-integrals of~\eqref{eq:dNs1.3RhoNe1ModifiedRedEq}.
Summing up, \emph{the equation~\eqref{eq:dNs1.3RhoNe1InviscidBurgersEq} admits
no nonlocal symmetries related to the differential substitution~$h=w_{0,2}$
but has generalized symmetries that are not induced by generalized symmetries
of the equation~\eqref{eq:dNs1.3RhoNe1ModifiedRedEq}.}
In other words, the elements from $\check\Sigma\setminus\boldsymbol{\bar\upsilon}(\Sigma)$
can be interpreted as nonlocal symmetries of~\eqref{eq:dNs1.3RhoNe1ModifiedRedEq},
and, therefore, as hidden nonlocal symmetries of~\eqref{eq:dN}.

The differential substitution~$w_{0,2}=h$ also induces the natural injective linear map
between the spaces of canonical representatives of equivalences classes
of cosymmetries, $\check\Gamma$ and~$\Gamma$,
(resp., of conservation-law characteristics, $\check\Lambda$ and~$\Lambda$)
of the equations~\eqref{eq:dNs1.3RhoNe1InviscidBurgersEq} and~\eqref{eq:dNs1.3RhoNe1ModifiedRedEq}
as well as the natural injective linear map between their spaces of conservation laws~$\check\Omega$ and~$\Omega$,
which act in the opposite direction in comparison with the case of generalized symmetries.

In particular, the space~$\check\Gamma$ of cosymmetries of the equation~\eqref{eq:dNs1.3RhoNe1InviscidBurgersEq}
consists of the differential functions depending at most on~$h$ and a finite number of~\smash{$\check\theta^k$}
and is embedded in the space~$\Gamma$ of cosymmetries of the equation~\eqref{eq:dNs1.3RhoNe1ModifiedRedEq}
just by substituting $w_{0,2}$ for~$h$,
which gives the first family~$\{\varrho\}$ of cosymmetries listed in Theorem~\ref{thm:dNs1.3RhoNe1Cosyms}.
Here $\varrho$ runs through the set of differential functions
depending at most on~$w_{0,2}$ and a finite number of~$\theta^k$.
Therefore, \emph{all the other elements of the space~$\Gamma$
can be interpreted as the canonical representatives of weak nonlocal cosymmetries
of the equation~\eqref{eq:dNs1.3RhoNe1InviscidBurgersEq}
that are associated with the differential substitution~$h=w_{0,2}$.}
Proper nonlocal cosymmetries among them should satisfy
the condition~\eqref{eq:dNs1.3RhoNe1ModifiedRedEqInTotalDers2'} with $\kappa^0=\kappa^1=0$
and thus are counterparts of local cosymmetries of the equation~\eqref{eq:dNs1.3RhoNe1InviscidBurgersEq}.

The descriptions of the corresponding embeddings
in the cases of conservation-law characteristics, conserved currents and conservation laws are analogous,
but the interpretations of the corresponding nonlocal objects differ from that for cosymmetries.

More specifically, the space~$\check\Lambda$ of conservation-law characteristics
of the equation~\eqref{eq:dNs1.3RhoNe1InviscidBurgersEq}
is spanned by the differential functions of the same form
as the elements of the first family of conservation-law characteristics
of~\eqref{eq:dNs1.3RhoNe1ModifiedRedEq} given in Theorem~\ref{thm:dNs1.3RhoNe1ModifiedRedEqCLChars},
where derivatives of~$h$ are substituted for the corresponding derivatives of~$w_{0,2}$.
This induces the natural embedding~$\iota$ of~$\check\Lambda$ in~$\Lambda$.
\emph{All the elements of $\Lambda\setminus\iota(\check\Lambda)$
can be interpreted as the canonical representatives of nonlocal conservation-law characteristics
of the equation~\eqref{eq:dNs1.3RhoNe1InviscidBurgersEq}
that are associated with the differential substitution~${h=w_{0,2}}$.}
Each of them is a weak nonlocal cosymmetry of the equation~\eqref{eq:dNs1.3RhoNe1InviscidBurgersEq}
in the above sense but not necessarily a proper nonlocal (and thus, local) cosymmetry of this equation.

Any conserved current of~\eqref{eq:dNs1.3RhoNe1InviscidBurgersEq} is equivalent to
a tuple of the form $(h^{-1}h_1\check\varrho,\,h_1\check\varrho)$,
where $\check\varrho$ is a~function at most of~$h$ and a finite number of~\smash{$\check\theta^k$}.
The space~\smash{$\check V$} of such tuples corresponds to
the first family of conserved currents of~\eqref{eq:dNs1.3RhoNe1ModifiedRedEq}
presented in Lemma~\ref{lem:dNs1.3RhoNe1ModifiedRedEqGenFormOfCCs}.
In other words, it is naturally embedded in the space~$V$ of conserved currents of~\eqref{eq:dNs1.3RhoNe1ModifiedRedEq}.
The subspace~$\check V_0$ of trivial conserved currents in~$\check V$
coincides, up to substituting $w_{0,2}$ for~$h$,
with the intersection of the first family of Lemma~\ref{lem:dNs1.3RhoNe1ModifiedRedEqGenFormOfCCs}
and the subspace~$V_0$ of~$V$.
As a result, the quotient space~$\check V/\check V_0$ can be naturally embedded in the quotient space~$V/V_0$.
The space~$\check\Omega$ of conservation laws of the equation~\eqref{eq:dNs1.3RhoNe1InviscidBurgersEq}
is naturally isomorphic to the space~$\check V/\check V_0$.
The last claim and Theorem~\ref{lem:dNs1.3RhoNe1ModifiedRedEqCLs}
jointly with the embedding of~$\check V/\check V_0$ in~$V/V_0$
imply the natural embedding~$\hat\iota$
of the space~$\check\Omega$ of conservation laws of the equation~\eqref{eq:dNs1.3RhoNe1InviscidBurgersEq}
in the space~$\Omega$ of conservation laws of the equation~\eqref{eq:dNs1.3RhoNe1ModifiedRedEq}.
\emph{All the elements of $\Omega\setminus\hat\iota(\check\Omega)$
can be interpreted as the canonical representatives of nonlocal conservation laws
of the equation~\eqref{eq:dNs1.3RhoNe1InviscidBurgersEq}
that are associated with the differential substitution~${h=w_{0,2}}$.}

\subsection{Symmetry-like objects of the intermediate equation}\label{sec:IntermediateEq}

In the context of the results of Section~\ref{sec:InviscidBurgersEqInducedObjects},
it is instructive to also study local symmetry-like objects of the equation
\begin{gather}\label{eq:IntermediateEq}
q_{12}+q_2q_{22}=0. %q_{1,1}+q_{0,1}q_{0,2}=0
\end{gather}
It can be considered as the intermediate equation
for the equations~\eqref{eq:dNs1.3RhoNe1InviscidBurgersEq} and~\eqref{eq:dNs1.3RhoNe1ModifiedRedEq}
and is related to them by the differential substitutions $h=q_2$ and $q=w_2$, respectively.
A distinguished feature of the equation~\eqref{eq:IntermediateEq}
in comparison with its counterparts~\eqref{eq:dNs1.3RhoNe1InviscidBurgersEq} and~\eqref{eq:dNs1.3RhoNe1ModifiedRedEq}
is that it is the Euler--Lagrange equation for a Lagrangian, $-\frac12q_1q_2-\frac12(q_2)^2$.

We compute the maximal Lie invariance pseudoalgebra of the equation~\eqref{eq:IntermediateEq}
using the packages {\sf DESOLV} \cite{carm2000a} and {\sf Jets} \cite{BaranMarvan} for {\sf Maple}.
This pseudoalgebra is spanned by the vector fields
\begin{gather*}%\label{eq:IntermediateEqMIA}
\begin{split}&
\breve P^1=\p_{z_1},\quad
\breve D^1=z_1\p_{z_1}-q\p_q,\quad
\breve K=z_1^2\p_{z_1}+z_1z_2\p_{z_2}+\tfrac12z_2^{\,2}\p_q,\\&
\breve D^2=z_2\p_{z_2}+2q\p_q,\quad
\breve P^2=\p_{z_2},\quad
\breve H=z_1\p_{z_2}+z_2\p_q,\quad
\breve R(\alpha)=\alpha(z_1)\p_q,
\end{split}
\end{gather*}
where
the parameter function~$\alpha$ runs through the set of smooth functions of~$z_1$.

Following the consideration of the equation~\eqref{eq:dNs1.3RhoNe1ModifiedRedEq}
in Section~\ref{sec:dNs1.3RhoNe1SymLikeObjects},
we (locally) represent the equation~\eqref{eq:IntermediateEq} in the Kovalevskaya form,
solving this equation with respect to the derivative~$q_{22}$.
Under the notation $q_{k,l}:=\p^{k+l}q/\p z_1^{\,k}\p z_2^{\,l}$, $k,l\in\mathbb N_0$,
this means that the derivatives $q_{k,l}$ with $k\in\mathbb N_0$ and $l\in\mathbb N_0+2$
and the other derivatives of~$q$ are considered to constitute
the tuples of the principal and the parametric derivatives of the equation~\eqref{eq:IntermediateEq},
respectively.
In other words,
the jet variables~$z_1$, $z_2$, $q_{k,l}$ with $k\in\mathbb N_0$ and $l\in\{0,1\}$
are chosen as the coordinates on the manifold~$\breve{\mathcal L}$
defined by the equation~\eqref{eq:IntermediateEq}
and its differential consequences in the jet space $\mathrm J^\infty(\mathbb R^2_{z_1,z_2}\times\mathbb R_q)$.
A~differential function~$\breve f$ of~$q$
that does not depend on the principal derivatives of the equation~\eqref{eq:IntermediateEq}
is denoted\ by~\smash{$\breve f\{q\}$}.
We rewrite all related objects defined in the beginning of Section~\ref{sec:dNs1.3RhoNe1SymLikeObjects}
in terms of~$q$ whenever this is possible, preserving the notation for these objects.
In particular,
\begin{gather*}
I^1:=q_{1,0}+\frac12(q_{0,1})^2,\quad
\zeta^{1k}:=\mathrm D_1^{\,k}I^1,
\\
\hat{\mathrm D}_1=\p_{z_1}+\left(\zeta^{10}-\frac12(q_{0,1})^2\right)\p_{q_{0,0}}+q_{k+1,1}\p_{q_{k,1}}+\zeta^{i,k+1}\p_{\zeta^{ik}},
\\[1ex]
\hat{\mathrm D}_2=\p_{z_2}+q_{0,1}\p_{q_{0,0}}
-\hat{\mathrm D}_1^k\left(\frac{q_{1,1}}{q_{0,1}}\right)\p_{q_{k,1}},\quad
\theta^k:=\left(\frac{q_{0,1}}{q_{1,1}}\hat{\mathrm D}_2\right)^k(z_2-q_{0,1}z_1).
\end{gather*}
In the modified coordinates $(q_{0,0},q_{1,0},q_{0,1},\theta^k,\zeta^{1k})$,
the operators~$\hat{\mathrm D}_1$ and~$\hat{\mathrm D}_2$ take the form
\begin{gather*}
\begin{split}
\hat{\mathrm D}_1={}&\left(\zeta^{10}-\frac{(q_{0,1})^2}2\right)\p_{q_{0,0}}
+q_{1,1}\p_{q_{0,1}}+\frac{(q_{1,1})^2}{q_{0,1}}(q_{1,1}\theta^2+2)\p_{q_{1,1}}
-q_{1,1}\theta^{k+1}\p_{\theta^k}+\zeta^{1,k+1}\p_{\zeta^{1k}},
\end{split}
\\[1ex]
\hat{\mathrm D}_2
=q_{0,1}\p_{q_{0,0}}-\frac{q_{1,1}}{q_{0,1}}\p_{q_{0,1}}
-\left(\frac{q_{1,1}}{q_{0,1}}\right)^2(q_{1,1}\theta^2+1)\p_{q_{1,1}}
+\frac{q_{1,1}}{q_{0,1}}\theta^{k+1}\p_{\theta^k}.
\end{gather*}

To construct the space of $z_2$-integrals of the equation~\eqref{eq:IntermediateEq},
the spaces of canonical representatives of its conserved currents, of its conservation-law characteristics
and of its trivial conserved currents,
it suffices to consider the corresponding objects for the equation~\eqref{eq:dNs1.3RhoNe1ModifiedRedEq},
select those among them that do not depend on the jet variables~$w_{k,0}$, $k\in\mathbb N$,
and represent the selected objects in terms of~$q$, substituting $w_{0,1}=q$.

\begin{theorem}\label{thm:IntermediateEqIntegrals}
A differential function $\breve\alpha\{q\}$ is a $z_2$-integral of the equation~\eqref{eq:IntermediateEq},
${\hat{\mathrm D}_2\breve\alpha\{q\}=0}$,
if and only if it is a sufficiently smooth function of $z_1$, $I^1$
and a finite number of total derivatives of $I^1$ with respect to~$z_1$,
\[
\breve\alpha=\breve\alpha\big(z_1,(\zeta^{1k})_{k=0,\dots,r}\big)=
\breve\alpha(z_1,I^1,\mathrm D_1I^1,\dots,\mathrm D_1^{\,r}I^1)
\quad\mbox{with}\quad r\in\mathbb N.
\]
\end{theorem}

\begin{lemma}\label{lem:IntermediateEqGenFormOfCCs}
Any conserved current of the equation~\eqref{eq:IntermediateEq} is equivalent to
a linear combination of the tuples
\begin{gather*}
\left(\frac{q_{1,1}}{q_{0,1}}\breve\varrho,\,q_{1,1}\breve\varrho\right),
\quad
\big(0,\,\breve\alpha\big),
\end{gather*}
where $\breve\varrho$ is an arbitrary function at most of~$q_{0,1}$ and a finite number of~$\theta^k$,
and $\breve\alpha$ is an arbitrary $z_2$-integral of~\eqref{eq:IntermediateEq}.
\end{lemma}

\begin{theorem}
The quotient space of conservation-law characteristics
of the equation~\eqref{eq:IntermediateEq} with respect to their equivalence
(i.e., modulo adding trivial characteristics, which vanishes on the solutions of~\eqref{eq:IntermediateEq})
is naturally isomorphic to the subspace spanned by the differential functions
\begin{gather*}
\left(\frac{q_{0,1}}{q_{1,1}}\hat{\mathrm D}_2\right)^k\breve\varrho_{\theta^k},
\quad
(-\mathrm D_1)^k\breve\alpha_{\zeta^{1k}}-(q_{0,1}-z_2\mathrm D_1)(-\mathrm D_1)^k\breve\alpha_{\zeta^{2k}},
\end{gather*}
where~$\breve\varrho$ is an arbitrary function at most of~$q_{0,1}$ and a finite number of~$\theta^k$,
and $\breve\alpha$ is an arbitrary $z_2$-integral of~\eqref{eq:IntermediateEq}.
\end{theorem}

Let~$\breve V$ and~$\breve V_0$ denote the linear span of conserved currents of the equation~\eqref{eq:IntermediateEq}
from Lemma~\ref{lem:IntermediateEqGenFormOfCCs}
and the subspace of trivial conserved currents belonging to~$\breve V$.

\begin{lemma}\label{lem:IntermediateEqTrivCCs}
The subspace~$\breve V_0$ consists of the tuples
\begin{gather*}
\big(\hat{\mathrm D}_2\hat\varrho-c_0q_{1,1}\theta^0,\,
q_{0,1}\hat{\mathrm D}_2\hat\varrho-c_0q_{0,1}q_{1,1}\theta^0
+\hat{\mathrm D}_1\hat\alpha+c_0\zeta^{10}\big),
\end{gather*}
where~$\hat\varrho$ is an arbitrary function at most of~$q_{0,1}$ and a finite number of~$\theta^k$,
$\hat\alpha$ is an arbitrary $z_2$-integral of~\eqref{eq:IntermediateEq}
and $c_0$ is an arbitrary constant.
\end{lemma}

\begin{theorem}\label{lem:IntermediateEqCLs}
The space~$\breve\Omega$ of conservation laws of the equation~\eqref{eq:IntermediateEq}
is naturally isomorphic to the quotient of the space~$\breve V$ by the subspace~$\breve V_0$.
\end{theorem}

The description of cosymmetries and generalized symmetries of the equation~\eqref{eq:IntermediateEq}
is a bit more involved.

\begin{theorem}
A differential function $\breve f\{q\}$ is a cosymmetry or, equivalently,
the characteristic of a generalized symmetry of the equation~\eqref{eq:IntermediateEq}
if and only if it is a linear combination of the differential functions
\begin{gather*}
\breve\varrho,\quad \breve\alpha,\quad q-\tfrac12z_1(q_{0,1})^2,
\end{gather*}
where~$\breve\varrho$ is an arbitrary function at most of~$q_{0,1}$ and a finite number of~$\theta^k$,
and $\breve\alpha$ is an arbitrary $z_2$-integral of~\eqref{eq:IntermediateEq}.
\end{theorem}

\begin{proof}
Since the equation~\eqref{eq:IntermediateEq} is the Euler--Lagrange equation for a Lagrangian,
the identity operator is both a Noether and an inverse Noether operators for this equation.
In other words, its cosymmetries are characteristics of its generalized symmetries, and vice versa.
The common determining equation for these objects is
$%\begin{gather}\label{eq:IntermediateEqGenInvCond}
%\hat{\mathrm D}_1\hat{\mathrm D}_2f+q_{0,1}\hat{\mathrm D}_2^{\,2}f-\frac{q_{1,1}}{q_{0,1}}\hat{\mathrm D}_2f=
\hat{\mathrm D}_2\big(\hat{\mathrm D}_1+q_{0,1}\hat{\mathrm D}_2\big)\breve f=0,
$ %\end{gather}
or, equivalently, \smash{$\big(\hat{\mathrm D}_1+q_{0,1}\hat{\mathrm D}_2\big)\breve f=\breve\kappa$},
where \smash{$\breve f=\breve f\{q\}$} is the unknown differential function,
and $\breve\kappa$ is a $z_2$-integral of~\eqref{eq:IntermediateEq}.
We use Lemma~\ref{lem:dNs1.3RhoNe1AuxiliaryLemma2b} and
select those solutions of the equation~\eqref{eq:dNs1.3RhoNe1ModifiedRedEqInTotalDers2'}
with \smash{$\kappa^0=\breve\kappa|^{}_{q=w_{0,1}}$} and $\kappa^1=0$
that do not depend on the jet variables~$w_{k,0}$, $k\in\mathbb N$.
Making the differential substitution $w_{0,1}=q$ in the selected solutions leads to the theorem's statement.
\end{proof}

The relation of symmetry-like objects of the equation~\eqref{eq:IntermediateEq} of any of the above kinds
to those of the equations~\eqref{eq:dNs1.3RhoNe1InviscidBurgersEq} and~\eqref{eq:dNs1.3RhoNe1ModifiedRedEq}
straightforwardly follows from the way of their construction
and is analogous to its counterpart for the equations~\eqref{eq:dNs1.3RhoNe1InviscidBurgersEq} and~\eqref{eq:dNs1.3RhoNe1ModifiedRedEq}
that is discussed in the end of Section~\ref{sec:dNRedEq1.3MIA}
or in Section~\ref{sec:InviscidBurgersEqInducedObjects}.

\section{Conclusion}\label{sec:Conclusion}

The study of the equation~\eqref{eq:dNs1.3RhoNe1ModifiedRedEq} in the present paper
has shown that this equation has remarkable properties
both as a submodel of the dispersionless Nyzhnyk equation~\eqref{eq:dN}
and as a partial differential equation considered independently
or in its relation to the inviscid Burgers equation~\eqref{eq:dNs1.3RhoNe1InviscidBurgersEq}.

The fact that the equation~\eqref{eq:dNs1.3RhoNe1ModifiedRedEq} arises
in the course of codimension-one Lie reductions
of the dispersionless Nyzhnyk equation~\eqref{eq:dN} in~\cite{vinn2024a}
gave us the initial inspiration for a deeper analysis of this equation.
The peculiarity of~\eqref{eq:dNs1.3RhoNe1ModifiedRedEq} revealed itself
already at this stage.
It contains no parameters, but corresponds,
as a reduced equation under a proper choice of ansatzes for reduction
in the spirit of \cite[Section~B]{fush1994b} and~\cite{popo1995b},
to the entire family of the subalgebras~$\mathfrak s_{1.3}^\rho$,
which are parameterized by the arbitrary nonvanishing function~$\rho$ of~$t$ with $\rho\not\equiv1$
and are in general $G$-inequivalent.
This is the only nontrivial Lie codimension-one submodel of~\eqref{eq:dN}
some of whose Lie symmetries are not induced by those of the equation~\eqref{eq:dN}
and thus are hidden for this equation.
Moreover, the values of the parameter function~$\rho$ in the subalgebras~$\mathfrak s_{1.3}^\rho$,
which is not involved in both the reduced equation~\eqref{eq:dNs1.3RhoNe1ModifiedRedEq}
and its maximal Lie invariance algebra~$\mathfrak a_{1.3}$,
define for some Lie symmetries of~\eqref{eq:dNs1.3RhoNe1ModifiedRedEq}
whether they are induced or not.
The differential substitution~$w_{22}=h$ lowers the order of the equation~\eqref{eq:dNs1.3RhoNe1ModifiedRedEq}
and maps it to the inviscid Burgers equation~\eqref{eq:dNs1.3RhoNe1InviscidBurgersEq}.
As a result, we constructed the large family~\eqref{eq:dNs1.3InvarSolutions}
of exact solutions of the dispersionless Nyzhnyk equation~\eqref{eq:dN},
which are parameterized by the second antiderivative of the general solution
of the inviscid Burgers equation~\eqref{eq:dNs1.3RhoNe1InviscidBurgersEq}
and the arbitrary nonvanishing function~$\rho$ of~$t$ with $\rho\not\equiv1$.
An essential part of more explicit exact solutions of~\eqref{eq:dN}
in terms of elementary or special functions or in a parametric form
that were constructed in~\cite{vinn2024a} using $G$-inequivalent codimension-two Lie reductions
belong, up to the $G$-equivalence, to the family~\eqref{eq:dNs1.3InvarSolutions}.
These are all the solutions presented in~\cite[Section~8.1]{vinn2024a}.
Their construction can be interpreted
as a result of carrying out two-step Lie reductions of~\eqref{eq:dN},
for each of which the first step is the codimension-one Lie reduction of~\eqref{eq:dN} to~\eqref{eq:dNs1.3RhoNe1ModifiedRedEq}
with respect to a subalgebra from the family~$\{\mathfrak s_{1.3}^\rho\}$
and the second step is a further Lie reduction of the equation~\eqref{eq:dNs1.3RhoNe1ModifiedRedEq}
with respect to its Lie symmetry induced by a Lie symmetry of~\eqref{eq:dN}.
The relatively simple integration of obtained reduced ordinary differential equations
can be explained by the presence of the differential substitution~$w_{22}=h$
mapping the equation~\eqref{eq:dNs1.3RhoNe1ModifiedRedEq} to the equation~\eqref{eq:dNs1.3RhoNe1InviscidBurgersEq}.
%\looseness=-1

The above properties of the submodel~\eqref{eq:dNs1.3RhoNe1ModifiedRedEq} hinted
that for finding exact solutions of~\eqref{eq:dN} in a~closed form,
it might be productive to carry out the exhaustive classification of Lie reductions of~\eqref{eq:dNs1.3RhoNe1ModifiedRedEq}
with respect to both its induced and noninduced Lie symmetries
following the optimized procedure of reducing submodels from \cite[Section~B]{kova2023b}.
We have exhaustively implemented this procedure for the submodel~\eqref{eq:dNs1.3RhoNe1ModifiedRedEq},
relating the Lie reductions of this submodel to specific Lie reductions of~\eqref{eq:dNs1.3RhoNe1InviscidBurgersEq}.
As a result, we have constructed exact solutions
of both~\eqref{eq:dNs1.3RhoNe1InviscidBurgersEq} and~\eqref{eq:dNs1.3RhoNe1ModifiedRedEq}
in an explicit form in terms of elementary or Lambert functions or in a parametric form.
In Theorem~\ref{thm:dNRedEq1.3CorrInvSolutions},
these solutions are properly extended by hidden symmetries and mapped to
solutions of the equation~\eqref{eq:dN},
which results in a better, more closed, representation for them than~\eqref{eq:dNs1.3InvarSolutions}.
The families of these solutions of~\eqref{eq:dN} are much wider than their counterparts
from~\cite[Section~8.1]{vinn2024a}.
As a by-product of comprehensively carrying out the Lie-reduction procedure for~\eqref{eq:dNs1.3RhoNe1ModifiedRedEq},
we have observed once more
that inequivalent reductions may lead to equivalent solutions.
Using the Lie reduction of~\eqref{eq:dN} to~\eqref{eq:dNs1.3RhoNe1ModifiedRedEq},
we have also considered for the first time the induction of point symmetries in the course of a Lie reduction,
which is a more complicated phenomenon than the similar induction of Lie-symmetry vector fields.

The study of the submodel~\eqref{eq:dNs1.3RhoNe1ModifiedRedEq} in the present paper
has gone well beyond the scope of Lie reductions.
When computing the point-symmetry pseudogroup~$G_{1.3}$ of the equation~\eqref{eq:dNs1.3RhoNe1ModifiedRedEq}
by the algebraic method, we have shown
that this pseudogroup coincides with the stabilizer of the maximal Lie invariance algebra~$\mathfrak a_{1.3}$
of~\eqref{eq:dNs1.3RhoNe1ModifiedRedEq}
in the pseudogroup of local diffeomorphisms in the space $\mathbb R^3_{z_1,z_2,w}$.
In other words, the algebraic condition that
the pushforwards of vector fields from~$\mathfrak a_{1.3}$ by transformations from~$G_{1.3}$
map~$\mathfrak a_{1.3}$ (on)to~$\mathfrak a_{1.3}$ completely defines the pseudogroup~$G_{1.3}$.
The submodel~\eqref{eq:dNs1.3RhoNe1ModifiedRedEq} is only the second, but much simpler, example of this kind.
The first example is given by the dispersionless Nyzhnyk equation~\eqref{eq:dN} itself~\cite[Remark~21]{boyk2024a}.
Usually, the point-symmetry (pseudo)group of a system of differential equations
is properly contained in the stabilizer of the maximal Lie invariance algebra of this system
in the pseudogroup of local diffeomorphisms in the underlying space
that the corresponding tuple of independent and dependent variables runs through.
The above phenomenon is just one of the displays of defining properties of Lie symmetries
of the equation~\eqref{eq:dNs1.3RhoNe1ModifiedRedEq}.
It has also turned out that this equation is Lie-remarkable.
More specifically, it is completely defined
by 11- and 12-dimensional subalgebras of the algebra~$\mathfrak a_{1.3}$
in the classes of genuine and general partial differential equations of order not greater than three
in two independent variables, respectively.
Furthermore, a six-dimensional subalgebra of the former subalgebra
completely defines the local diffeomorphisms that stabilize the algebra~$\mathfrak a_{1.3}$. %\looseness=-1

We have also found all the local symmetry-like objects associated with the equation~\eqref{eq:dNs1.3RhoNe1ModifiedRedEq},
which include generalized symmetries, cosymmetries, conservation-law characteristics and conservation laws.
This is the first so comprehensive study of local symmetry-like objects for a submodel
of a well-known system of differential equations.
Complete descriptions even of particular kinds of these objects in nontrivial cases
exist in the literature only for a minor part of such systems themselves, not to mention submodels,
see, e.g., \cite{holb2023a,ivan2007a,kova2023a,opan2020c,popo2025a,popo2025b,popo2020b,serg2024a,serg2016a,vane2021a}
for recent results and the review in \cite[Section~1]{opan2020e}.
Moreover, complete descriptions of all the local symmetry-like objects of a model
in a single paper are rather exceptional \cite{opan2020c,popo2025a,popo2025b,popo2020b,serg2016a,vane2021a}.
Since the possible number of arguments of the parameter functions in such objects
associated with the equation~\eqref{eq:dNs1.3RhoNe1ModifiedRedEq} grows when the object order does,
standard techniques like recursion operators and
the estimation of the dimension of the space of objects in question up to an arbitrary fixed order
do not work when a complete description of this space is required.
Even the best computer packages for finding local symmetry-like objects
such as {\sf Jets} \cite{BaranMarvan,marv2009a} and {\sf GeM}~\cite{chev2007a} for {\sf Maple}
are inefficient at computing such objects for the equation~\eqref{eq:dNs1.3RhoNe1ModifiedRedEq}
even at low orders, starting from order three.
This can be explained by the fact that for local symmetry-like objects of any specific kind,
the corresponding space of them for the equation~\eqref{eq:dNs1.3RhoNe1ModifiedRedEq}
is of complicated structure.
In particular, it is parameterized by several arbitrary functions of an arbitrary finite number of arguments
that are differential functions from two different infinite sequences,
see \cite{olve2021a,opan2020c,popo2020b} for similar spaces of local symmetry-like objects. %\looseness=1

Using the package {\sf Jets} \cite{BaranMarvan,marv2009a} for {\sf Maple},
in~\cite[Section~2]{boyk2024a} we showed that
generalized symmetries of the equation~\eqref{eq:dN} at least up to order five are exhausted,
modulo the equivalence of generalized symmetries, by its Lie symmetries.
Conservation laws characteristics up to order two were considered in~\cite{moro2021a}.
The comparison of these results
with Theorems~\ref{thm:dNs1.3RhoNe1ModifiedRedEqGenSyms} and~\ref{thm:dNs1.3RhoNe1ModifiedRedEqCLChars}
shows that the dispersionless Nyzhnyk equation~\eqref{eq:dN} admits
many nontrivial hidden generalized symmetries and hidden conservation laws related to the Lie reductions
with respect to each subalgebra from the family~$\{\mathfrak s_{1.3}^\rho\}$.

The homomorphism~$\boldsymbol{\bar\upsilon}\colon\Sigma\to\check\Sigma$
between the algebras~$\Sigma$ and~$\check\Sigma$
of canonical representatives of equivalences classes of generalized symmetries
of the equations~\eqref{eq:dNs1.3RhoNe1ModifiedRedEq} and~\eqref{eq:dNs1.3RhoNe1InviscidBurgersEq}
that is induced by the differential substitution~$w_{0,2}=h$
is neither injective nor surjective.
As a result, the equation~\eqref{eq:dNs1.3RhoNe1ModifiedRedEq}
possesses nonlocal symmetries associated with the differential substitution~$w_{0,2}=h$,
but this is not the case for the equation~\eqref{eq:dNs1.3RhoNe1InviscidBurgersEq}.
These nonlocal symmetries of~\eqref{eq:dNs1.3RhoNe1ModifiedRedEq}
can be considered as hidden nonlocal symmetries of~\eqref{eq:dN}.
The analogous natural induced linear maps
between the corresponding spaces of canonical representatives of equivalences classes
of cosymmetries and of conservation-law characteristics
as well as the corresponding spaces of conservation laws act in the opposite direction
and are injective, but not surjective.
This is why the equation~\eqref{eq:dNs1.3RhoNe1InviscidBurgersEq} admits
nonlocal conservation laws that are associated with the differential substitution~$w_{0,2}=h$.

For the equation~\eqref{eq:IntermediateEq},
the spaces of canonical representatives of equivalences classes
of its cosymmetries, of its conserved currents, of its conservation-law characteristics
and of its trivial conserved currents
have easily been constructed from the suitable subspaces of their counterparts
for the equation~\eqref{eq:dNs1.3RhoNe1ModifiedRedEq}
with analogous linear maps induced by the differential substitution~$w_{0,2}=h$.
Due to the Lagrangian nature of the equation~\eqref{eq:IntermediateEq},
this has also directly given the algebra of canonical representatives of equivalences classes of generalized symmetries.
The relations of symmetry-like objects of the equation~\eqref{eq:IntermediateEq}
to those of the equations~\eqref{eq:dNs1.3RhoNe1InviscidBurgersEq} and~\eqref{eq:dNs1.3RhoNe1ModifiedRedEq}
are then obvious.

The extended symmetry analysis of other Nyzhnyk models
in the dispersive symmetric case and in the both dispersive and dispersionless asymmetric cases,
including the complete descriptions of various hidden symmetry-like objects,
can be carried out similarly to that for the dispersionless symmetric case in~\cite{boyk2024a,vinn2024a} and this paper.
Another possible direction for further studies related to Nyzhnyk models
is to construct integrable generalizations of these models in more than three independent variables,
which reduce to Nyzhnyk models under certain additional constraints.
Recently, such a generalization of the asymmetric dispersionless Nyzhnyk system has been derived in \cite[Eq.~(21)]{serg2025a}
in dimension 1+3.

\section*{Acknowledgments} 

The authors are grateful to Serhii Koval, Dmytro Popovych, Galyna Popovych and Artur Sergyeyev
for helpful discussions and suggestions.
The authors also sincerely thank the anonymous reviewer for a number of helpful suggestions and comments,
which led to essentially improving the presentation of results.
R.O.P. also expresses his gratitude for the hospitality shown by the University of Vienna during his long-term stay there.
This work was supported in part by grants from the Simons Foundation
(1290607 and SFI-PD-Ukraine-00014586, O.O.V., V.M.B.).
The work of R.O.P. was supported in part by the Ministry of Education, Youth and Sports of the Czech Republic
(M\v SMT \v CR) under RVO funding for I\v C47813059.
The authors express their deepest thanks to the Armed Forces of Ukraine and the civil Ukrainian people
for their bravery and courage in defense of peace and freedom in Europe and in the entire world from russism.


\footnotesize
\begin{thebibliography}{10}

\bibitem{abra2006a}
Abraham-Shrauner B., Govinder K.S. and Arrigo D.J.,
Type-II hidden symmetries of the linear 2D and 3D wave equations,
{\it J.~Phys. A} {\bf 39} (2006), 5739--5747. %, no.20

\bibitem{abra2006b}
Abraham-Shrauner B. and Govinder K.S.,
Provenance of type II hidden symmetries from nonlinear partial differential equations,
{\it J.~Nonlinear Math. Phys.} {\bf 13} (2006), 612--622. %, no.4

\bibitem{CRC_v1}
Ames W.F., Anderson R.L., Dorodnitsyn V.A., Ferapontov E.V., Gazizov R.K., Ibragimov N.H. and Svirshchevskii S.R.,
{\it CRC handbook of Lie group analysis of differential equations. Vol.~1. Symmetries, exact solutions and conservation laws},
edited by N.H.~Ibragimov, CRC Press, Boca Raton, FL, 1994.

\bibitem{andr2001a}
Andriopoulos K., Leach P.G.L. and Flessas G.P.,
Complete symmetry groups of ordinary differential equations and their integrals: some basic considerations.
{\it J.~Math. Anal. Appl.} {\bf 262} (2001), 256--273. %, no. 1

\bibitem{baik1989a}
Baikov V.A., Gazizov R.K. and Ibragimov N.Kh.,
Perturbation methods in group analysis,
{\it J.~Soviet Math.} {\bf 55} (1991), 1450--1490. %, no. 1

\bibitem{BaranMarvan}
Baran H. and Marvan M.,
Jets. A software for differential calculus on jet spaces and diffieties.
Available at \texttt{http://jets.math.slu.cz}.

\bibitem{bihl2012b}
Bihlo A., Dos Santos Cardoso-Bihlo E. and Popovych R.O.,
Complete group classification of a class of nonlinear wave equations,
{\it J.~Math. Phys.} {\bf 53} (2012), 123515, arXiv:1106.4801.

\bibitem{bihl2015a}
Bihlo A., Dos Santos Cardoso-Bihlo E.M. and Popovych R.O.,
Algebraic method for finding equivalence groups,
{\it J.~Phys. Conf. Ser.} {\bf 621} (2015), 012001, arXiv:1503.06487. %, 17~pp.

\bibitem{bihl2011b}
Bihlo A. and Popovych R.O.,
Point symmetry group of the barotropic vorticity equation,
in \emph{Proceedings of 5th Workshop ``Group Analysis of Differential Equations \& Integrable Systems'' (June 6--10, 2010, Protaras, Cyprus)},
University of Cyprus, Nicosia, 2011, pp. 15--27, arXiv:1009.1523.

\bibitem{blum2010A}
Bluman G.W., Cheviakov A.F. and Anco S.C.,
{\it Applications of symmetry methods to partial differential equations},
Springer, New York, 2010.

\bibitem{blum1989A}
Bluman G.W. and Kumei S.,
{\it Symmetries and differential equations},
Springer, New York, 1989.

\bibitem{boyk2024a}
Boyko V.M., Popovych R.O. and Vinnichenko O.O.,
Point- and contact-symmetry pseudogroups of dispersionless Nizhnik equation,
{\it Commun. Nonlinear Sci. Numer. Simul.} {\bf 132} (2024), 107915, arXiv:2211.09759. % 19 pp.

\bibitem{carm2000a}
Carminati J. and Vu K.,
Symbolic computation and differential equations: Lie symmetries,
{\it J.~Symbolic Comput.} {\bf 29} (2000), 95--116.

\bibitem{chap2024a}
Chapovskyi Ye.Yu., Koval S.D. and Zhur O.,
Subalgebras of Lie algebras. Example of $\mathfrak{sl}_3(\mathbb R)$ revisited,
2024, arXiv:2403.02554.

\bibitem{chev2007a}
Cheviakov A.F.,
GeM software package for computation of symmetries and conservation laws of differential equations,
{\it Comput. Phys. Comm.} {\bf 176} (2007), 48--61.

\bibitem{card2011a}
Dos Santos Cardoso-Bihlo E., Bihlo A. and Popovych R.O.,
Enhanced preliminary group classification of a class of generalized diffusion equations,
{\it Commun. Nonlinear Sci. Numer. Simulat.} {\bf 16} (2011), 3622--3638, arXiv:1012.0297.

\bibitem{card2013a}
Dos Santos Cardoso-Bihlo E. and Popovych R.O.,
Complete point symmetry group of the barotropic vorticity equation on a rotating sphere,
{\it J.~Engrg.\ Math.} {\bf 82} (2013), 31--38, arXiv:1206.6919.

\bibitem{card2021a}
Dos Santos Cardoso-Bihlo E. and Popovych R.O.,
On the ineffectiveness of constant rotation in the primitive equations and their symmetry analysis,
{\it Commun. Nonlinear Sci. Numer. Simul.} {\bf 101} (2021), 105885, arXiv:1503.04168. %, 15 pp.

\bibitem{fush1994a}
Fushchych W. and Popowych R.,
Symmetry reduction and exact solutions of the Navier--Stokes equations.~I,
{\it J.~Nonlinear Math. Phys.} {\bf 1} (1994), 75--113, arXiv:math-ph/0207016.

\bibitem{fush1994b}
Fushchych W. and Popowych R.,
Symmetry reduction and exact solutions of the Navier--Stokes equations.~{II},
{\it J.~Nonlinear Math. Phys.} {\bf 1} (1994), 158--188, arXiv:math-ph/0207016.

\bibitem{gorg2019a}
Gorgone M. and Oliveri F.,
Lie remarkable partial differential equations characterized by Lie algebras of point symmetries,
{\it J.~Geom. Phys.} {\bf 144} (2019), 314--323, arXiv:2108.02171.

\bibitem{holb2023a}
Holba~P., Complete classification of local conservation laws for a family of PDEs generalizing Cahn--Hilliard and Kuramoto--Sivashinsky equations,
{\it Stud. Appl. Math.} {\bf 151} (2023), 171--182, arXiv:2108.08693.

\bibitem{hydo1998a}
Hydon P.E.,
Discrete point symmetries of ordinary differential equations,
{\it Proc. R. Soc. Lond. Ser. A Math. Phys. Eng. Sci.} {\bf 454} (1998), 1961--1972. %, no.~1975

\bibitem{hydo1998b}
Hydon P.E.,
How to find discrete contact symmetries,
{\it J.~Nonlinear Math. Phys.} {\bf 5} (1998), 405--416. %, no. 4

\bibitem{hydo2000b}
Hydon P.E.,
How to construct the discrete symmetries of partial differential equations,
{\it Eur. J. Appl. Math.} {\bf 11} (2000), 515--527. %, no.~5

\bibitem{hydo2000A}
Hydon P.E.,
{\it Symmetry methods for differential equations},
Cambridge University Press, Cambridge, 2000.

\bibitem{ivan2007a}
Ivanova N.M.,
Conservation laws of multidimensional diffusion-convection equations,
{\it Nonlinear Dynam.} {\bf 49} (2007), 71--81, arXiv:math-ph/0604057.

\bibitem{kamk1959B}
Kamke E.,
{\it Differentialgleichungen. L\"osungsmethoden und L\"osungen.
Teil II: Partielle Differentialgleichungen erster Ordnung f\"ur eine gesuchte Funktion},
%Mathematik und ihre Anwendungen in Physik und Technik, Reihe A, Band 18
Akademische Verlagsgesellschaft Geest \& Portig K.-G., Leipzig, 1959.

\bibitem{kapi1978a}
Kapitanskii L.V.,
Group analysis of the Euler and Navier--Stokes equations in the presence of rotational symmetry and new exact solutions of these equations,
{\it Sov. Phys. Dokl.} {\bf 23} (1978), 896--898.

\bibitem{katk1965a}
Katkov V.L.,
Group classification of solutions of the Hopf equation,
{\it J.~Appl. Mech. Tech. Phys.} {\bf 6} (1965), no.~6, 71--71.

\bibitem{king1998a}
Kingston J.G. and Sophocleous C.,
On form-preserving point transformations of partial differential equations,
{\it J.~Phys.~A} {\bf 31} (1998), 1597--1619.

\bibitem{kont2019a}
Kontogiorgis S., Popovych R.O. and Sophocleous C.,
Enhanced symmetry analysis of two-dimensional Burgers system,
{\it Acta Appl. Math.} {\bf 163} (2019), 91--128, arXiv:1709.02708.

\bibitem{kova2023b}
Koval S.D., Bihlo A. and Popovych R.O.,
Extended symmetry analysis of remarkable (1+2)-dimensional Fokker--Planck equation,
{\it European J. Appl. Math.} {\bf 34} (2023), 1067--1098, arXiv:2205.13526.

\bibitem{kova2023a}
Koval S.D. and Popovych R.O.,
Point and generalized symmetries of the heat equation revisited,
{\it J.~Math. Anal. Appl.} {\bf 527} (2023), 127430, arXiv:2208.11073. % no. 2,

\bibitem{krau1994a}
Krause J.,
On the complete symmetry group of the classical Kepler system,
{\it J.~Math. Phys.} {\bf 35} (1994), 5734--5748. %, no. 11

\bibitem{malt2024a}
Maltseva D.S. and Popovych R.O.,
Complete point-symmetry group, Lie reductions and exact solutions of Boiti--Leon--Pempinelli system,
{\it Phys.~D} {\bf 460} (2024), 134081, arXiv:2103.08734.

\bibitem{mann2014a}
Manno G., Oliveri F., Saccomandi G. and Vitolo R.,
Ordinary differential equations described by their Lie symmetry algebra,
{\it J.~Geom. Phys.} {\bf 85} (2014), 2--15.

\bibitem{mann2007b}
Manno G., Oliveri F. and Vitolo R.,
On differential equations characterized by their Lie point symmetries,
{\it J.~Math. Anal. Appl.} {\bf 332} (2007), 767--786. %, no. 2

\bibitem{marv2009a}
Marvan M.,
Sufficient set of integrability conditions of an orthonomic system,
{\it Found. Comp. Math.} {\bf 9} (2009), 651--674, arXiv:nlin/0605009.

\bibitem{moro2021a}
Morozov O.I. and Chang J.-H.,
The dispersionless Veselov--Novikov equation: symmetries, exact solutions, and conservation laws,
{\it Anal. Math. Phys.} {\bf 11} (2021), 126. %26~pp. , no.~3

\bibitem{nizh1980a}
Nizhnik L.P.,
Integration of multidimensional nonlinear equations by the inverse problem method,
%Integration of multidimensional nonlinear equations by the method of the inverse problem,
{\it Soviet Phys. Dokl.} {\bf 25} (1980), 706--708.

\bibitem{nucc1996b}
Nucci M.C.,
The complete Kepler group can be derived by Lie group analysis,
{\it J.~Math. Phys.} {\bf 37} (1996), 1772--1775. %, no. 4

\bibitem{oliv2004a}
Oliveri F.,
Lie symmetries of differential equations: direct and inverse problems,
{\it Note Mat.} {\bf 23} (2004), no.~2, 195--216.

\bibitem{olve1993A}
Olver P.J.,
{\it Application of Lie groups to differential equations},
Springer, New York, 1993.

\bibitem{olve2021a}
Olver P.J.,
Higher-order symmetries of underdetermined systems of partial differential equations and Noether's second theorem,
{\it Stud. Appl. Math.} {\bf 147} (2021), 904--913.

\bibitem{opan2020a}
Opanasenko S., Bihlo A., Popovych R.O. and Sergyeyev A.,
Extended symmetry analysis of isothermal no-slip drift flux model,
{\it Phys. D} {\bf 402} (2020), 132188, arXiv:1705.09277.

\bibitem{opan2020c}
Opanasenko S., Bihlo A., Popovych R.O. and Sergyeyev A.,
Generalized symmetries, conservation laws and Hamiltonian structures of an isothermal no-slip drift flux model,
{\it Phys. D} {\bf 411} (2020), 132546, arXiv:1908.00034.

\bibitem{opan2020e}
Opanasenko S. and Popovych R.O.,
Generalized symmetries and conservation laws of (1+1)-dimensional Klein--Gordon equation,
{\it J.~Math. Phys.} {\bf 61} (2020), 101515, arXiv:1810.12434.

\bibitem{ovsi1982A}
Ovsiannikov L.V.,
{\it Group analysis of differential equations},
Academic Press, New York -- London, 1982.

\bibitem{pavl2006a}
Pavlov M.V.,
Modified dispersionless Veselov--Novikov equation and corresponding hydrodynamic chains, 2006, arXiv:nlin/0611022.

\bibitem{poch2017a}
Pocheketa O.A. and Popovych R.O.,
Extended symmetry analysis of generalized Burgers equations,
{\it J.~Math. Phys.} {\bf 58} (2017), 101501, arXiv:1603.09377.

\bibitem{poly2012A}
Polyanin A.D. and Zaitsev V.F.,
{\it Handbook of nonlinear partial differential equations}, second edition,
Chapman \& Hall/CRC, Boca Raton, FL, 2012.

\bibitem{popo2024a}
Popovych D.R., Bihlo A. and Popovych R.O.,
Generalized symmetries of Burgers equation,
2024, arXiv: 2406.02809.

\bibitem{popo2025a}
Popovych D.R., Bihlo A. and Popovych R.O.,
Conservation laws and variational symmetry of the Liouville equation,
in preparation.

\bibitem{popo2025b}
Popovych D.R., Dos Santos Cardoso-Bihlo E.M. and Popovych R.O.,
Higher-order symmetry-like objects of inviscid Burgers equation,
in preparation.

\bibitem{popo2003a}
Popovych R.O., Boyko V.M., Nesterenko M.O. and Lutfullin M.W.,
Realizations of real low-dimensional Lie algebras,
{\it J.~Phys.~A} {\bf 36} (2003), 7337--7360, arXiv:math-ph/0301029. %, no.~26

\bibitem{popo2020b}
Popovych R.O. and Cheviakov A.F.,
Variational symmetries and conservation laws of the wave equation in one space dimension,
{\it Appl. Math. Lett.} {\bf 104} (2020), 106225, arXiv:1912.03698.

\bibitem{popo1995b}
Popowych R.,
On Lie reduction of the Navier--Stokes equations,
{\it J.~Nonlinear Math. Phys.} {\bf 2} (1995), 301--311. %, no.3--4

\bibitem{rose1986a}
Rosenhaus V.,
The unique determination of the equation by its invariance group and field-space symmetry,
{\it Algebras Groups Geom.} {\bf 3} (1986), 148--166.

\bibitem{serg2024a}
Sergyeyev A.,
Complete description of local conservation laws for generalized dissipative Westervelt equation,
{\it Qual. Theory Dyn. Syst.} {\bf 23} (2024), 209.

\bibitem{serg2025a}
Sergyeyev A., Multidimensional integrable systems from contact geometry,
{\it Bol. Soc. Mat. Mex.} {\bf 31} (2025), 26, arXiv:2501.04474.

\bibitem{serg2016a}
Sergyeyev A. and Vitolo~R.,
Symmetries and conservation laws for the Karczewska--Rozmej--Rutkowski--Infeld equation,
{\it Nonlinear Anal. Real World Appl.} {\bf 32} (2016), 1--9, arXiv:1511.03975.

\bibitem{vane2021a}
Vaneeva O.O., Popovych R.O. and Sophocleous C.,
Enhanced symmetry analysis of two-dimensional degenerate Burgers equation,
{\it J.~Geom. Phys.} {\bf 169} (2021), 104336, arXiv:1908.01877.

\bibitem{vinn2024a}
Vinnichenko O.O., Boyko V.M. and Popovych R.O.,
Lie reductions and exact solutions of dispersionless Nizhnik equation,
{\it Anal. Math. Phys.} {\bf 14} (2024), 82, arXiv:2308.03744. %56 pp.,

\bibitem{yeho2004a}
Yehorchenko I.,
Group classification with respect to hidden symmetry,
{\it Proceedings of Institute of Mathematics of NAS of Ukraine. Mathematics and its Applications} {\bf 50} (2004), Part 1, 290--297.

\bibitem{zakh1994a}
Zakharov V.E.,
Dispersionless limit of integrable systems in 2+1 dimensions,
in {\it Singular limits of dispersive waves (Lyon, 1991)}, {\it NATO Adv. Sci. Inst. Ser. B: Phys.}, 320, Plenum, New York, 1994, pp.~165--174.

\end{thebibliography}
\end{document}